\documentclass[%
 reprint,
 superscriptaddress,
 amsmath,amssymb,
 aps,
 prx,               
 longbibliography,  
 floatfix,
]{revtex4-2}

\AtBeginDocument{%
  }

\usepackage{amsmath,amssymb,mathtools}
\usepackage{amsthm}
\usepackage{graphicx}
\usepackage{subcaption}
\graphicspath{{figures/}}
\usepackage{textcomp}
\usepackage{xcolor}
\usepackage{ragged2e}
\usepackage{braket}
\usepackage{booktabs}
\usepackage{physics}
\usepackage{float}

\usepackage{algorithm}
\usepackage{algpseudocode}

\theoremstyle{definition}
\newtheorem{definition}{Definition}
\theoremstyle{plain}
\newtheorem{theorem}{Theorem}
\newtheorem{lemma}[theorem]{Lemma}
\newtheorem{proposition}[theorem]{Proposition}
\newtheorem{corollary}[theorem]{Corollary}

\algrenewcommand\algorithmiccomment[1]{\hfill{\footnotesize$\triangleright$ #1}}

\usepackage{xspace}

\usepackage{hyperref}
\hypersetup{
  colorlinks=true,
  linkcolor=blue,
  citecolor=blue,
  urlcolor=blue,
}

\algnewcommand\algorithmicinput{\textbf{Input:}}
\algnewcommand\algorithmicoutput{\textbf{Output:}}
\algnewcommand\Input{\item[\algorithmicinput]}
\algnewcommand\Output{\item[\algorithmicoutput]}

\newcommand{\Fop}[1]{\mathcal{F}_{#1}}       
\newcommand{\Pop}[1]{\hat{P}_{#1}}            
\newcommand{\Vop}[1]{\hat{V}_{#1}}            
\DeclareMathOperator{\Inv}{Inv}
\newcommand{\SWAP}{\text{SWAP}\xspace}
\newcommand{\CNOT}{\text{CNOT}\xspace}
\newcommand{\CZ}{\text{CZ}\xspace}
\newcommand{\FSWAP}{\text{FSWAP}\xspace}
\newcommand{\Z}{\text{Z}\xspace}
\newcommand{\SWAPs}{\text{SWAPs}\xspace}
\newcommand{\CNOTs}{\text{CNOTs}\xspace}

\newcommand{\FSWAPs}{\text{FSWAPs}\xspace}

\newcommand{\Bswap}[1]{B_{#1}}
\newcommand{\Ffull}[1]{F_{#1}}

\DeclareMathOperator{\polylog}{polylog}

\begin{document}

\title{Asymptotically Optimal Depth Fermionic Permutation on 2D Grid Quantum Architecture without Ancillas}

\author{Dantong Li}
\email{dantong.li@yale.edu}
\affiliation{Department of Computer Science, Yale University, New Haven, Connecticut 06511, USA}
\affiliation{Yale Quantum Institute, Yale University, New Haven, Connecticut 06511, USA}

\author{Shifan Xu}
\email{shifan.xu@yale.edu}
\affiliation{Department of Applied Physics, Yale University, New Haven, Connecticut 06511, USA}
\affiliation{Yale Quantum Institute, Yale University, New Haven, Connecticut 06511, USA}

\author{Yongshan Ding}
\email{yongshan.ding@yale.edu}
\affiliation{Department of Computer Science, Yale University, New Haven, Connecticut 06511, USA}
\affiliation{Yale Quantum Institute, Yale University, New Haven, Connecticut 06511, USA}
\affiliation{Department of Applied Physics, Yale University, New Haven, Connecticut 06511, USA}

\date{\today}

\begin{abstract}
Simulating fermionic systems on qubit hardware involves many nonlocal interactions, and efficient routing of these interactions is critical to the overall cost of fermionic simulation algorithms. Recent works reduce this routing overhead to polylogarithmic depth under all-to-all connectivity, but degrade to $O(\sqrt{N}\log N)$ for $N$ fermionic modes on 2D nearest-neighbor architectures. We present a fermionic permutation protocol tailored to 2D grid architectures that achieves the optimal $O(\sqrt{N})$ depth with $O(N\sqrt{N})$ nearest-neighbor gates and no ancilla qubits, measurements, or classical feedforward. For $N=L^2$ modes, the CNOT depth is at most $10L+12$. This matches the $\Omega(\sqrt{N})$ lower bound, which holds even when $O(N)$ ancillas and classical feedforward are permitted. We further construct an $O(\sqrt{N})$-depth transformation between the Jordan--Wigner, Bravyi--Kitaev, and Parity encodings on the 2D grid via a Hilbert-curve layout, extending our result to all three encodings. Benchmarks on a standalone fermionic permutation circuit, the fermionic fast Fourier transform, and Trotter simulation of a sparse SYK model demonstrate that the practical crossover is near the analytical depth-bound breakeven at $L=6$ ($N=36$), followed by consistent reductions in circuit depth and spacetime volume and higher estimated no-fault probability in the early fault-tolerant regime, growing with system size.
\end{abstract}


\maketitle

\begin{figure*}[t]
\centering
\includegraphics{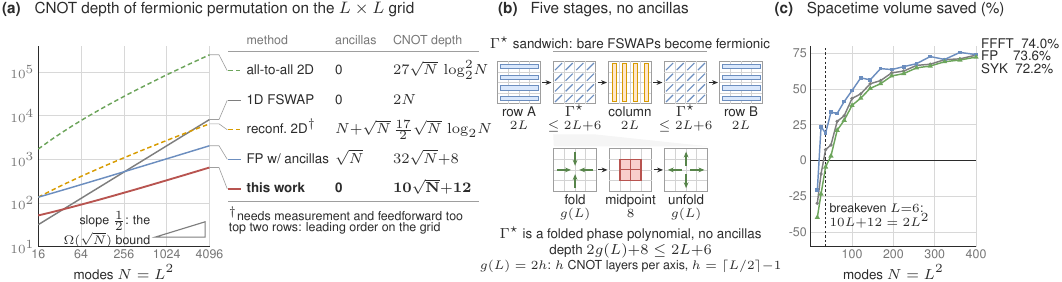}
\caption{Summary of results. (a)~CNOT depth of fermionic permutation on the $L\times L$ nearest-neighbor grid for the methods of Table~\ref{tab:fp-comparison}, from the formulas listed in the table at right (leading order for the two $\polylog N$ protocols compiled to the grid); the table also gives each method's ancilla count. Ours is ancilla-free at $O(\sqrt N)$ depth and runs parallel to the slope-$\tfrac12$ guide, the $\Omega(\sqrt N)$ lower bound that holds even with $O(N)$ ancillas and feedforward. (b)~The protocol: three sorting stages of depth $2L$ each---a row sort, a column sort, and a row sort---with the column sort conjugated by the fixed diagonal operator $\Gamma^\star$, which turns its bare vertical \FSWAP{}s into fermionic ones. Each $\Gamma^\star$ is a folded phase polynomial: a CNOT fold of depth $g(L)=2h$ ($h$ layers along each axis), an $8$-layer midpoint, and an unfold of depth $g(L)$, for depth $2g(L)+8\le 2L+6$ with no ancillas; the five stages total at most $10L+12$. (c)~Spacetime volume saved by the folded construction as a function of system size: standalone random fermionic permutation against the 1D \FSWAP{} network (circuit-level volume, circles), and the 2D FFFT (squares) and sparse-SYK Trotter step (triangles) against their baselines (logical volume $S=QD$ at code distance $d=11$; CT-FFFT and the 1D \FSWAP{} network, respectively), from the data of Section~\ref{sec:eval}. The dashed line marks the analytical breakeven $L=6$, where $10L+12=2L^2$; the measured crossovers fall at $N=25$--$49$, and the saving reaches $72$--$74\%$ at $N=400$.}
\label{fig:summary}
\end{figure*}

\section{Introduction}
\label{sec:intro}

Simulating interacting fermions is central to quantum chemistry, condensed matter physics, and lattice gauge theory~\cite{Abrams_1997, Ortiz_2001, Lanyon_2010, Wecker_2014, Reiher_2017, Bauer_2020, McArdle_2020, di2024quantum}, and it is a task on which quantum computers can sidestep the sign problem that obstructs classical Monte Carlo methods~\cite{manin1980computable, feynman2018simulating}. Doing so requires a fermion-to-qubit mapping, because qubit operators have a local tensor-product structure while fermionic operators obey nonlocal anticommutation relations. The Jordan--Wigner (JW) transformation~\cite{jordan1928paulische} serializes the modes along a chain and represents each operator as a Pauli-$Z$ string of length $O(N)$; the Bravyi--Kitaev (BK)~\cite{bk} and Parity~\cite{parity} encodings spread this cost more evenly with tree-structured parity bookkeeping. All three fix a static labeling of modes to qubits, and no static labeling keeps every interaction local once the interaction graph changes between simulation steps. The active modes must then be routed dynamically by a \emph{fermionic permutation}---the fermionic analog of qubit routing, and the dynamic counterpart to the static \emph{initial placement} that prior fermion-mapping work~\cite{fermihedral, hatt} optimizes. Unlike ordinary routing, a fermionic permutation must also preserve anticommutation: in the JW encoding, every pair of modes whose order changes needs a phase correction. This makes dynamic fermionic routing a dominant cost in Trotterized simulation of the Fermi--Hubbard and Sachdev--Ye--Kitaev (SYK) models and in the fermionic fast Fourier transform (FFFT) used throughout plane-wave electronic structure~\cite{maskara2025fastsimulationfermionsreconfigurable, constantinides2025lowdepthfermionroutingancillas, ffft, ffft-ct}.

The textbook approach permutes modes along the snake JW chain with \FSWAP{} networks~\cite{fswap1Dnetwork}, at $O(N)$ depth and $O(N^2)$ gates on an $L\times L$ grid, leaving the grid's second axis unused. Two recent works reduce this dramatically under \emph{reconfigurable} or \emph{all-to-all} connectivity by decomposing any permutation into $O(\log N)$ layers of structured primitives. Maskara et al.~\cite{maskara2025fastsimulationfermionsreconfigurable} use merge-sort-style interleaves and reach $O(\log N)$ depth with $\Theta(N)$ ancillas, mid-circuit measurement, and classical feedforward; Constantinides et al.~\cite{constantinides2025lowdepthfermionroutingancillas} use staircase permutations from recursive bisection and reach $O(\log^2 N)$ depth with no ancillas, for the whole class of product-preserving ternary-tree encodings that includes JW, BK, and Parity. Neither decomposition is tailored to a planar layout. On a 2D grid, where information needs $\Theta(\sqrt N)$ time merely to cross the array~\cite{devulapalli2022routing}, each nonlocal layer acquires a $\Theta(\sqrt N)$ routing overhead, and our grid compilations of the two protocols (Appendices~\ref{app:reconf-2dnn} and~\ref{app:constantinides-2dnn}) have depths $O(\sqrt N\log N)$~\cite{maskara2025fastsimulationfermionsreconfigurable} and $O(\sqrt N\log^2 N)$~\cite{constantinides2025lowdepthfermionroutingancillas} on 2D nearest-neighbor hardware~\cite{google2025quantum,kim2023evidence}---above the $\Omega(\sqrt N)$ lower bound by logarithmic factors. The one existing grid-native ingredient is the diagonal parity correction $\Gamma$ of Jiang et al.~\cite{gamma}, introduced for specific FFFT and Hubbard-model subroutines and implemented there with $\Theta(\sqrt N)$ traveling ancillas at depth $13L+O(1)$.

In this article we close the gap between these protocols and the lower bound (Figure~\ref{fig:summary}). We present a fermionic permutation protocol for the 2D nearest-neighbor grid with CNOT depth at most $10L+12$ for $N=L^2$ modes---the asymptotically optimal $O(\sqrt N)$---using $O(N\sqrt N)$ nearest-neighbor entangling gates and \textbf{no ancillas, mid-circuit measurements, or classical feedforward} [Figure~\ref{fig:summary}(a)]. The optimality is strong: the $\Omega(\sqrt N)$ bound from grid expansion holds in the LOCC model with $O(N)$ ancillas, measurements, and feedforward~\cite{devulapalli2022routing}, so our construction, which uses none of these, is optimal against protocols with strictly more power. The depth bound meets the $2L^2$ of the 1D \FSWAP{} network at $L=6$ ($N=36$) and is quadratically better beyond it. On the same grid we give an $O(\sqrt N)$-depth transformation between the JW, BK, and Parity encodings through a Hilbert space-filling curve layout~\cite{hilbert1935stetige}, so the result extends to all three encodings. Numerically, in standalone fermionic permutation, the 2D FFFT, and Trotterized sparse-SYK simulation, the crossovers fall near the analytical $L=6$ breakeven, varying with workload and metric, and at $N=400$ the folded construction reduces spacetime volume by $72$--$74\%$ against each application's baseline, with the highest estimated no-fault probability in the early fault-tolerant regime [Figure~\ref{fig:summary}(c)].

Two ideas make this possible [Figure~\ref{fig:summary}(b)]. First, rather than linearizing the modes and compiling a one-dimensional circuit onto the grid, we start from Hall's Row--Column--Row decomposition, a permutation factorization intrinsically adapted to the 2D architecture, and realize each of its three stages by odd--even transposition networks running on all rows or all columns in parallel. Second, we absorb the fermionic phase correction into this grid-local data movement with a diagonal unitary $\Gamma$, chosen so that conjugation by $\Gamma$ promotes every parity-uncorrected vertical swap of the column stage to the full fermionic swap at once. Our \emph{folded} construction of $\Gamma$ selects a representative equivalent to that of Ref.~\cite{gamma}, maps its dense quadratic phase into a local prefix basis by CNOT wavefronts that fold in from both ends of every row and column, applies a constant-depth midpoint, and unfolds. This removes the traveling ancillas and reduces the per-$\Gamma$ depth from $13L+O(1)$ to $2L+O(1)$, an approximately $85\%$ reduction in the leading constant; the coefficient two is optimal for a universal diagonal correction (Section~\ref{sec:optimality}). Concurrently and independently, Aigner et al.~\cite{aigner2026fermionlattices} also reached $O(\sqrt N)$-depth ancilla-free fermionic routing on the grid, by dynamically switching between row- and column-oriented JW encodings; their switch plays the same parity-correction role as $\Gamma$, and its displayed CNOT-ladder construction of depth $4L+O(1)$ gives $14\sqrt N+O(1)$ per permutation under our accounting (Table~\ref{tab:fp-comparison}). What distinguishes our construction is the depth-$2L+O(1)$ correction, which brings the total to $10\sqrt N+12$, and the matching $\Omega(\sqrt N)$ lower bound.

Fermionic permutation complements, rather than replaces, the static mapping optimized by Fermihedral~\cite{fermihedral} and HATT~\cite{hatt}: those fix one minimum-Pauli-weight layout for the whole computation, which suits Hamiltonian simulation but not interaction-driven primitives such as the FFFT, and Pauli weight controls depth only under all-to-all connectivity---on a grid a low-weight term can still be spread across distant qubits. Conversely, for a Hamiltonian already local under one layout (e.g.\ 2D Fermi--Hubbard) routing buys little, and specialized circuits~\cite{fswap1Dnetwork} or classical algorithms~\cite{majoranaprop} stay competitive. Our contributions are:
\begin{itemize}
    \item \textbf{Optimal depth without ancillas.} A fermionic permutation algorithm on the 2D nearest-neighbor grid with depth $10\sqrt N+12$, $O(N\sqrt N)$ entangling gates, and zero ancillas, measurements, or feedforward (Section~\ref{sec:method})---a quadratic depth and $\sqrt N$-factor gate-count improvement over the 1D \FSWAP{} network, and a depth-constant improvement from $32$ to $10$ over the ancilla-based construction of Ref.~\cite{gamma}.
    \item \textbf{Folded $\Gamma$.} A phase-polynomial synthesis that eliminates the $O(\sqrt N)$ ancillas of Ref.~\cite{gamma} and reduces the per-$\Gamma$ depth from $13L+O(1)$ to $2L+O(1)$, with an optimal leading coefficient (Sections~\ref{sec:gamma-impl} and~\ref{sec:optimality}).
    \item \textbf{BK and Parity encodings.} An $O(\sqrt N)$-depth conversion between JW, BK, and Parity on the grid via a Hilbert-curve layout for power-of-two $L$, lifting the result to all three encodings (Section~\ref{sec:extensions}).
    \item \textbf{Evaluation.} Benchmarks of standalone fermionic permutation, the 2D FFFT, and sparse-SYK Trotter simulation under circuit-level and early-fault-tolerant cost models (Section~\ref{sec:eval}), showing crossovers near the $L=6$ breakeven and $72$--$74\%$ spacetime-volume reductions at $N=400$.
\end{itemize}

Section~\ref{sec:background} fixes notation: the JW mapping on the snake-ordered grid, the \FSWAP{} gate, and the odd--even transposition network. Section~\ref{sec:method} presents the protocol---Hall's decomposition (Section~\ref{sec:hall}), the $\Gamma$ operator (Section~\ref{sec:gamma}), its ancilla-free circuit (Section~\ref{sec:gamma-impl}), and the optimality argument (Section~\ref{sec:optimality}). Section~\ref{sec:extensions} treats BK and Parity, Section~\ref{sec:eval} reports the experiments, and Section~\ref{sec:conclusion} concludes. The folded-$\Gamma$ proofs and the 2D compilations of the two $O(\polylog N)$ protocols are in the appendices.

\section{Background}
\label{sec:background}

\subsection{Majorana operators and the Jordan--Wigner mapping}
\label{sec:jw}

We consider $N = L^2$ fermionic modes mapped to an $L \times L$ nearest-neighbor qubit array.
Each fermionic mode~$j$ ($0 \le j \le N{-}1$) gives rise to two \emph{Majorana operators} $\gamma_{2j}$ and $\gamma_{2j+1}$ as self-adjoint operators satisfying the Clifford algebra
\begin{equation}
\{\gamma_i,\, \gamma_k\} \;=\; 2\,\delta_{ik}, \qquad \gamma_i = \gamma_i^\dagger,
\label{eq:majorana-car}
\end{equation}
for $i,k \in \{0, \ldots, 2N{-}1\}$.
The Majorana picture is the natural language for fermion-to-qubit encodings (Section~\ref{sec:ternary-tree}), where each encoding is specified by a ternary tree whose leaves carry Majorana operators.

Under the Jordan--Wigner transformation (JWT)~\cite{jordan1928paulische}, the $2N$ Majorana operators map to qubit operators as
\begin{equation}
\gamma_{2j} \;\mapsto\; X_j \prod_{\ell=0}^{j-1} Z_\ell, \qquad
\gamma_{2j+1} \;\mapsto\; Y_j \prod_{\ell=0}^{j-1} Z_\ell.
\label{eq:jwt}
\end{equation}
The standard fermionic creation and annihilation operators are recovered as
\begin{equation}
a_j^\dagger = \tfrac{1}{2}(\gamma_{2j} - i\gamma_{2j+1}), \qquad
a_j = \tfrac{1}{2}(\gamma_{2j} + i\gamma_{2j+1}),
\label{eq:majorana-to-ca}
\end{equation}
which under Eq.~\eqref{eq:jwt} reproduces the familiar JWT form $a_j^\dagger \mapsto \frac{1}{2}(X_j - iY_j)\prod_{\ell<j} Z_\ell$ and preserves the canonical anticommutation relations $\{a_j, a_k^\dagger\} = \delta_{jk}$.

The JWT requires choosing a linear ordering of the $N$ modes.
On our $L \times L$ grid, we thread the qubits into a one-dimensional chain using the \emph{snake} (boustrophedon) order shown in Figure~\ref{fig:snake}: even rows are traversed left-to-right, odd rows right-to-left, giving the JWT index $\mathrm{jw}(r,c) = rL + c$ for even $r$ and $rL + (L{-}1{-}c)$ for odd $r$.
This ordering makes every horizontal neighbor JWT-adjacent (indices differ by~1), but vertical neighbors are separated by up to $2L{-}1$ sites along the chain. This vertical separation is the source of the long parity strings that fermionic permutation algorithms must address.

\begin{figure}[t]
    \centering
    \includegraphics[width=\linewidth]{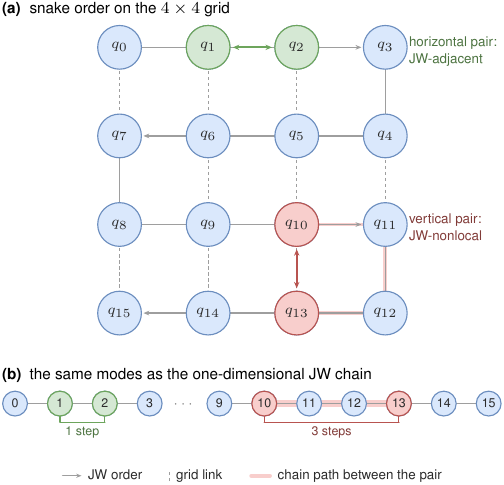}
    \caption{Snake JW order on a $4\times4$ grid. (a)~Qubits are threaded row by row, reversing direction at every row turn; arrows follow the JW index. A horizontal grid pair (green, $q_1,q_2$) is JW-adjacent. A vertical grid pair (red, $q_{10},q_{13}$) is JW-nonlocal: its chain path crosses the row turn, so a vertical swap needs a phase correction along every mode between them. (b)~The same modes as the one-dimensional JW chain: the horizontal pair is one step apart, the vertical pair three; in general vertical neighbors are up to $2L-1$ steps apart.}
    \label{fig:snake}
\end{figure}

More specifically, fixing a JW linear order over $N$ sites, for $\pi \in S_N$, the symmetric group of all $N!$ permutations of the $N$ modes, a fermionic permutation is defined as:
\begin{equation}
\Fop{\pi} \;=\; \Pop{\pi}\, \Vop{\pi},
\qquad
\Vop{\pi} \;=\; \prod_{(i,j) \in \Inv(\pi)} \CZ_{i,j},
\label{eq:fperm}
\end{equation}
where $\Pop{\pi}$ physically permutes the modes according to $\pi$, and
$\Inv(\pi) = \{(i,j) : i < j \;\wedge\; \pi(i) > \pi(j)\}$ is the inversion set. As data movement it is \emph{routing}---bringing interacting modes adjacent, as the FFFT does between its butterfly levels (Section~\ref{sec:eval-ffft}). Equivalently, $\Fop{\pi}$ is the unique unitary (up to global phase) that permutes the Majorana operators: $\Fop{\pi}\,\gamma_{2j+\alpha}\,\Fop{\pi}^\dagger = \gamma_{2\pi(j)+\alpha}$ for all modes $j$ and $\alpha \in \{0,1\}$. As the conjugation it is a \emph{fermionic Hamiltonian rewrite} that makes each term local in turn, which is how it accelerates Hamiltonian-simulation applications such as the sparse SYK Trotter steps of Section~\ref{sec:eval-syk}.

A concrete example shows why this primitive matters in Hamiltonian simulation. Take a single hopping term between two modes $j$ and $k$ that sit at opposite ends of the snake JW chain. Under the JW map of Eq.~\eqref{eq:jwt} it becomes a Pauli string of weight $\Theta(N)$,
\begin{equation}
a_j^\dagger a_k + a_k^\dagger a_j \;\mapsto\; \tfrac{1}{2}\big(X_j\,\overline{Z}\,X_k + Y_j\,\overline{Z}\,Y_k\big),\quad \overline{Z}=\!\!\prod_{\ell=j+1}^{k-1}\!\! Z_\ell,
\label{eq:hop-string}
\end{equation}
whose Trotter exponential compiles to the standard snake-order \CNOT staircase of depth $O(N)$ spanning the whole chain. Sandwiching the term in a fermionic permutation $\Fop{\pi}$ that brings $j$ and $k$ to adjacent JW positions leaves the physics unchanged but collapses the long string to a local weight-2 operator $\tfrac{1}{2}(XX{+}YY)$ in $O(1)$ depth. On the 2D grid our protocol supplies $\Fop{\pi}$ in $O(\sqrt N)$ depth, turning the $O(N)$ staircase into $O(\sqrt N)$ routing plus $O(1)$ interaction.

\subsection{Fermionic SWAP operations}
\label{sec:fswap-bg}

The $\FSWAP$ gate on two JWT-adjacent qubits $j, j{+}1$ acts as
\begin{equation}
\FSWAP_{j,j+1}\ket{s_j\, s_{j+1}} \;=\; (-1)^{s_j \cdot s_{j+1}} \ket{s_{j+1}\, s_j}.
\label{eq:fswap-adj}
\end{equation}
The $\FSWAP$ is locally equivalent to iSWAP and requires exactly 2 entangling gates, executed as $H_a\to\CNOT(a,b)\to\CNOT(b,a)\to H_b$~\cite{fswap}, strictly better than the naive $\SWAP$ (3~CNOTs) $+$ $\CZ$ (1~CNOT) $= 4$~CNOTs. Under the Clifford-tracked lattice-surgery cost convention the two Hadamards exchange the logical $X/Z$ labels and modify the subsequent parity measurements without separate Hadamard timesteps~\cite{litinski2018lattice}, so we count the $\FSWAP$ as two CNOT layers throughout.

For horizontal neighbors $(r,c)$ and $(r,c{\pm}1)$, the JWT indices differ by~1, so the $\FSWAP$ is a local 2-qubit gate. By contrast, for vertical neighbors $(r,c)$ and $(r{+}1,c)$ with JWT indices $j < k$ (WLOG, since $\Ffull{j,k} = \Ffull{k,j}$), the gap satisfies $k - j = 2|d| + 1$, where $d$ is the horizontal distance to the snake turning corner. The full fermionic SWAP then requires a \emph{parity correction}:
\begin{equation}
    \Ffull{j,k}\ket{\mathbf{s}} = (-1)^{s_j s_k \;+\; (s_j + s_k)\sum_{\ell=j+1}^{k-1} s_\ell}\; \ket{\mathbf{s}'},
\label{eq:full-fswap}
\end{equation}
where $\mathbf{s}'$ is obtained from $\mathbf{s}$ by exchanging the two occupations---$s'_j = s_k$, $s'_k = s_j$, and $s'_\ell = s_\ell$ otherwise. We write $P = \sum_{\ell=j+1}^{k-1} s_\ell$ for the \emph{parity string}, which is the core difficulty of 2D fermionic simulation.

For a vertical grid-neighbor pair with JWT indices $j < k$, the \emph{bare} FSWAP applies the standard FSWAP gate (Eq.~\eqref{eq:fswap-adj}) to the two physically adjacent qubits while \emph{ignoring} the parity string:
\begin{equation}
\Bswap{j,k}\ket{\mathbf{s}} = \FSWAP_{j,k}\ket{\mathbf{s}} = (-1)^{s_j s_k}\,\ket{\mathbf{s}'}.
\label{eq:bare-fswap}
\end{equation}
We denote by $\Ffull{j,k}$ the full fermionic SWAP of Eq.~\eqref{eq:full-fswap}, which additionally incorporates the parity correction $(-1)^{(s_j+s_k)P}$. Thus, the bare FSWAP is a simple nearest-neighbor gate, but it does not by itself supply the required parity correction for vertical moves. In Section~\ref{sec:gamma} we introduce a diagonal unitary~$\Gamma$ whose conjugation action converts every bare FSWAP into its full fermionic counterpart, thereby achieving the necessary correction using only nearest-neighbor gates and no ancillas.

Throughout this paper, we use the shorthand $\Bswap{t}$ and $\Ffull{t}$ to denote the (product of parallel) bare and full FSWAPs applied in round~$t$, respectively.

\subsection{Odd--even transposition sort}
\label{sec:oet-bg}

Any permutation on a line of $L$ elements can be realized by an \emph{odd--even transposition} (OET) network~\cite{oet, fswap1Dnetwork}: as shown in Fig.~\ref{fig:one-row}, in round $t\in\{0,\ldots,L-1\}$ the network conditionally swaps all pairs in
$\mathcal{M}_t=\{(i,i+1):0\le i<L-1,\ i\equiv t\pmod 2\}$.
Each round consists of pairwise-disjoint adjacent transpositions, hence all swaps within a round execute in parallel.

On a 1D chain of $N$ qubits under the JWT, an arbitrary fermionic permutation $\Fop{\pi}$ can be implemented by an OET sort using local FSWAPs.
On the 2D $L \times L$ grid ($N=L^2$), applying the same 1D approach along the snake order yields depth $2N = O(N)$ and $O(N^2)$ gates.
Our algorithm improves this to $10\sqrt{N} +O(1) = O(\sqrt{N})$ depth and $O(N\sqrt{N})$ gates.

\begin{figure}[t]
    \centering
    \includegraphics[width=\linewidth]{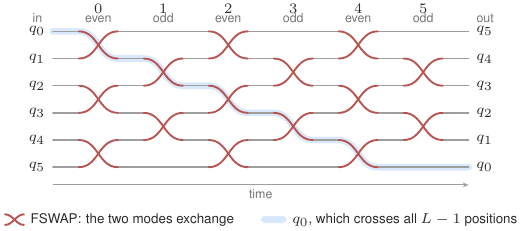}
    \caption{Odd--even transposition network on one row of $L=6$ modes, realizing the reversal $q_i\mapsto q_{L-1-i}$. Round $t$ applies \FSWAP{}s (red crossings) to every pair in $\mathcal M_t$, alternating between even and odd pairs; $L$ rounds suffice for any permutation of the row. The highlighted mode $q_0$ crosses all $L-1$ positions.}
    \label{fig:one-row}
\end{figure}

\section{Ancilla-Free Fermionic Permutation}
\label{sec:method}

\subsection{Protocol overview}
\label{sec:overview}

We present an ancilla-free protocol for implementing arbitrary fermionic permutations on an $L \times L$ grid in $10L + O(1) = O(L) = O(\sqrt{N})$ circuit depth using $O(N\sqrt{N})$ nearest-neighbor gates and zero ancilla qubits, which is a quadratic improvement in depth and a $\sqrt{N}$-factor improvement in gate count over the current best baseline.

The protocol exploits a classical result showing that any permutation on a 2D grid decomposes into three within-line stages: Row-Column-Row (Section~\ref{sec:hall}). Each stage is realized by a depth-$O(L)$ OET network of parallel FSWAPs (Section~\ref{sec:oet-bg}). Under the snake JWT, row-stage FSWAPs are natively JWT-adjacent and require no phase correction. The column stage, however, involves JWT-non-adjacent pairs whose full fermionic SWAPs demand parity-string corrections spanning up to $\Theta(L)$ intermediate qubits.

\begin{figure*}[t]
    \centering
    \includegraphics{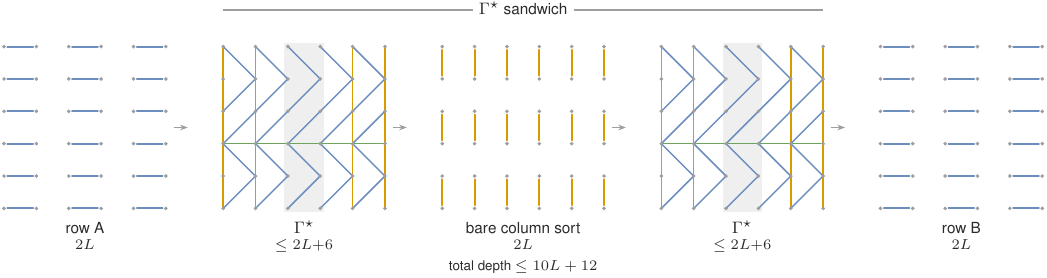}
    \caption{The five stages of Algorithm~\ref{alg:main}, drawn on a $6\times6$ grid.
    Each sorting stage is one OET network running on all lines in parallel; one round is drawn per sorting tile, and each stage runs $L$ rounds of depth-two \FSWAPs, giving the $2L$ below those tiles.
    The row stages use horizontal \FSWAPs (blue bars), which act on JWT-adjacent pairs and need no correction; the column stage uses bare vertical \FSWAPs (orange bars), which do.
    Conjugating that stage by the fixed phase operator $\Gamma^\star$ promotes every bare vertical \FSWAP to a full fermionic \SWAP at once.
    The two $\Gamma^\star$ tiles sketch its phase pattern in the folded basis, up to a constant number of local patches: thin strokes are phase edges, not gates, one on a diagonal of every cell (blue), one on every vertical edge outside the shaded central columns (orange), and one on each edge of a single horizontal path (green).
    Figure~\ref{fig:gamma-overview} gives the circuit that realizes it.}
    \label{fig:algorithm-overview}
\end{figure*}

\begin{algorithm}[H]
\caption{Ancilla-free fermionic permutation on an $L \times L$ grid}
\label{alg:main}
\begin{algorithmic}[1]
\Input Permutation $\pi \in S_N$ under a fixed snake JW order ($N = L^2$).
\Output A circuit implementing $\Fop{\pi}$ on the $N$ data qubits.
\State $(\pi^{\mathrm{rowA}}, \pi^{\mathrm{col}}, \pi^{\mathrm{rowB}}) \gets \textsc{Hall3StagePlan}(\pi)$
\Comment{classical}
\State \textsc{ExecuteRowStage}($\pi^{\mathrm{rowA}}$) \Comment{$O(\sqrt{N})$ depth}
\State \textsc{ApplyGamma}() \Comment{$O(\sqrt{N})$ depth, fixed $\Gamma^\star$}
\State \textsc{ExecuteBareColumnSort}($\pi^{\mathrm{col}}$) \Comment{$O(\sqrt{N})$ depth}
\State \textsc{ApplyGamma}() \Comment{same $\Gamma^\star$}
\State \textsc{ExecuteRowStage}($\pi^{\mathrm{rowB}}$) \Comment{$O(\sqrt{N})$ depth}
\end{algorithmic}
\end{algorithm}

To overcome this nonlocal phase requirement, we develop a diagonal unitary $\Gamma$ (Section~\ref{sec:gamma}); $\Gamma^\star$ denotes the fixed representative our synthesis uses, which is what Algorithm~\ref{alg:main} and Figure~\ref{fig:algorithm-overview} show. It is implementable with depth-$O(\sqrt{N})$ nearest-neighbor gates on the bare grid and zero ancillas, whose conjugation action converts every bare (parity-uncorrected) vertical FSWAP into its full fermionic counterpart:
\[
\Gamma \;\cdot\; \Bswap{j,k} \;\cdot\; \Gamma
\;=\;
\Ffull{j,k}.
\]
Because $\Gamma^2 = I$, inserting $\Gamma\Gamma$ between consecutive rounds of the column sort lets a single $\Gamma$ before and a single $\Gamma$ after the entire sort promote every bare round at once, yielding the five-stage pipeline of Algorithm~\ref{alg:main} and Figure~\ref{fig:algorithm-overview}.

\begin{figure*}[t]
    \centering
    \includegraphics[width=\textwidth]{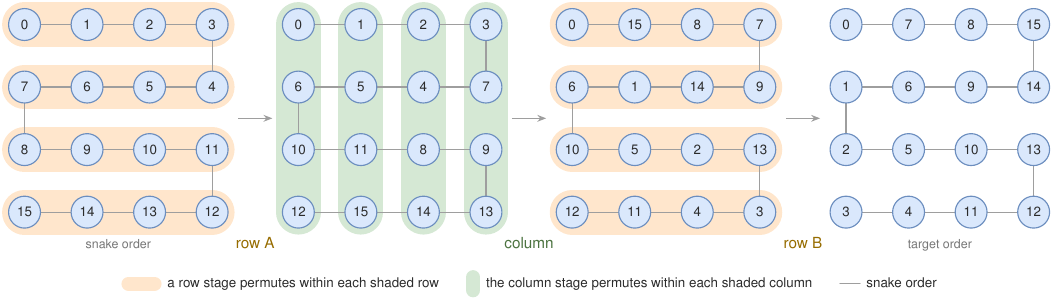}
    \caption{Hall's Row-Column-Row decomposition on a $4\times4$ grid: row~A permutes within each source row (shaded rows), the column stage moves each mode within its column to its destination row (shaded columns), and row~B permutes within each destination row to the final column. The grey path is the snake order.}
    \label{fig:3-stage}
\end{figure*}

\subsection{Hall's three-stage Row-Column-Row routing}
\label{sec:hall}

Hall's three-stage routing~\cite[Theorem~4]{hall-routing} decomposes any permutation $\pi$ on an $L \times L$ grid as $\pi = \pi^{\mathrm{rowB}} \circ \pi^{\mathrm{col}} \circ \pi^{\mathrm{rowA}}$ (Figure~\ref{fig:3-stage}), where $\pi^{\mathrm{rowA}}$ and $\pi^{\mathrm{rowB}}$ are within-row permutations and $\pi^{\mathrm{col}}$ is a within-column permutation. The construction builds a bipartite graph from source rows to destination rows, which is $L$-regular and hence admits a decomposition into $L$ perfect matchings by the edge coloring theorem~\cite{Lregularmatching}. These matchings determine intermediate column assignments: RowA places items into their intermediate columns, the column stage moves items to destination rows, and RowB routes items to their final columns. We treat \textsc{Hall3StagePlan} as a classical preprocessing step, computable in $O(N \log N)$ time via standard edge-coloring algorithms~\cite{edgecoloring}. We choose Row-Col-Row over Col-Row-Col because under the snake JWT, horizontal neighbors are JWT-adjacent, so row-stage FSWAPs are natively local; only the single column stage requires parity correction via $\Gamma$, giving two $\Gamma$ applications total instead of four.

\subsection{Row and column sorting stages}
\label{sec:row-col}

Both row and column stages sort within each line using the OET network of Section~\ref{sec:oet-bg}.
In the row stages, horizontal neighbors are JWT-adjacent, so each swap is a local FSWAP (2~CNOTs) requiring no parity correction.
All rows are JWT-disjoint and execute in parallel.
In the column stage, which operates inside the $\Gamma$ sandwich, each swap is a \emph{bare} FSWAP on a vertical neighbor pair---the same 2-CNOT gate applied without parity correction.
All columns are disjoint and execute in parallel.
In both cases, the OET network runs $L$ rounds; each FSWAP contributes depth~2, giving depth $2L = 2\sqrt{N}$ and $\le N\sqrt{N}$ entangling gates per stage. Sandwiched between two $\Gamma$ applications, the bare column sort is promoted to a full fermionic column sort.

\begin{figure*}[t]
    \centering
    \includegraphics{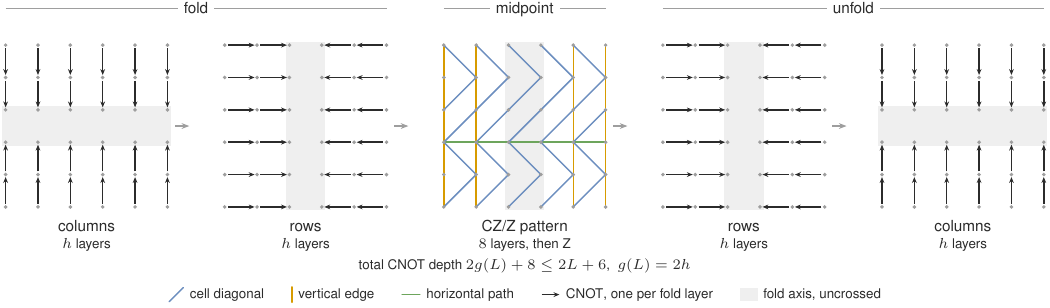}
    \caption{The ancilla-free $\Gamma^\star$ circuit of Algorithm~\ref{alg:gamma}, one tile per step, on a $6\times6$ grid ($m=3$, fold depth $h=m-1$ per axis, so the two folds together have depth $g(L)=2h$, which is $L-2$ for even $L$ and $L-1$ for odd $L$).
    The unfold repeats the fold's gates in reverse order, hence the identical arrow direction.
    The midpoint tile is the local edge set of Lemma~\ref{lem:gamma-local-factors}, with row $p$ the largest odd integer at most $m$.
    Two caveats on the drawing: $L=6$ is shown for compactness, whereas the eight-layer midpoint and the depth $2g(L)+8$ are the general construction (Appendix~\ref{app:gamma} proves $L\ge7$; $2\le L\le6$ use code-verified representatives); and the selected $\Gamma^\star$ differs from the pattern drawn by the residual \CZ gates of Appendix~\ref{app:gamma} ($4L-4$ edges and four cell patches, Fig.~\ref{fig:app-residual}), which Lemma~\ref{lem:app-zero-extra-placement} interleaves into unused slots of the row fold, midpoint, and row unfold layers at no extra depth.}
    \label{fig:gamma-overview}
\end{figure*}

\subsection{The $\Gamma$ operator}
\label{sec:gamma}

A valid $\Gamma$ correction is a diagonal unitary
\[
\Gamma\ket{\mathbf{s}} \;=\; (-1)^{f(\mathbf{s})}\ket{\mathbf{s}},
\]
where $f : \mathbb{F}_2^N \to \mathbb{F}_2$ is a degree-2 polynomial chosen using only the fixed snake JW geometry. Because every eigenvalue is $\pm 1$, $\Gamma$ is Hermitian and self-inverse: $\Gamma^\dagger = \Gamma = \Gamma^{-1}$. Its defining property is that conjugation converts each bare vertical \FSWAP into its full fermionic counterpart:
\begin{equation}
\Gamma \;\Bswap{j,k}\; \Gamma \;=\; \Ffull{j,k}.
\label{eq:gamma-conjugation}
\end{equation}
To see what condition this places on $f$, write $\gamma_{\mathbf{s}}=(-1)^{f(\mathbf{s})}$ and let $\mathbf{s}'$ be obtained from $\mathbf{s}$ by exchanging occupations $s_j$ and $s_k$. Since $\Bswap{j,k}$ and $\Ffull{j,k}$ already agree when $s_j=s_k$, only $s_j\oplus s_k=1$ constrains $\Gamma$. Comparing their phases in that case gives the \emph{parity-encoding condition}: for every vertical grid-neighbor pair with JW indices $j<k$ and every $\mathbf{s}\in\mathbb{F}_2^N$ with $s_j\oplus s_k=1$,
\begin{equation}
\gamma_{\mathbf{s}}\,\gamma_{\mathbf{s}'}
\;=\;(-1)^{\sum_{\ell=j+1}^{k-1}s_\ell}.
\label{eq:parity-encoding}
\end{equation}
This condition does not determine $f$ uniquely. We write $\Gamma^\star$ for the fixed representative used by our synthesis; it depends only on $L$ and the snake JW layout, not on $\pi$ or the active vertical swaps, and the same $\Gamma^\star$ is used on both sides of every column stage. Appendix~\ref{app:gamma} proves that the selected representative satisfies Eq.~\eqref{eq:parity-encoding} simultaneously for every vertical grid edge.

\subsection{Ancilla-free $\Gamma$ circuit}
\label{sec:gamma-impl}

For each $L$, we synthesize $\Gamma^\star$ by changing to a parity-XOR basis, applying a grid-wide nearest-neighbor \CZ/$Z$ phase pattern, and undoing the basis change. On a line of bits $x_0,\ldots,x_{L-1}$, a CNOT ladder stores cumulative parities. Our two-sided ladder uses $m=\lceil L/2\rceil$ and stores
\[
y_i=
\begin{cases}
\bigoplus_{j=0}^{i}x_j,&i<m,\\
\bigoplus_{j=i}^{L-1}x_j,&i\ge m.
\end{cases}
\]
We call this basis change a \emph{fold}: prefix- and suffix-XOR wavefronts propagate simultaneously from opposite boundaries toward the center, so each travels only about half a line. Folding first the columns and then the rows makes corner-anchored rectangular parities available across the grid. In this folded basis, the dense phase becomes a constant-depth local pattern of \CZ gates together with parallel \Z corrections spanning all data qubits. The cell diagonals of that pattern are not themselves nearest-neighbor pairs; eight midpoint layers synthesize them by \CNOT conjugation, after which a parallel layer of single-qubit \Z gates costs no further \CNOT depth. The pattern corrects every possible vertical edge, independent of the column-sort instance. Reversing the row and column folds restores the computational basis. Algorithm~\ref{alg:gamma} gives the circuit and Figure~\ref{fig:gamma-overview} shows its five blocks; Figures~\ref{fig:app-fold} and~\ref{fig:app-layers} in Appendix~\ref{app:gamma} illustrate the fold and draw every layer of the circuit for $L=7$.

The construction uses $O(N)$ nearest-neighbor \CNOT, \CZ, and \Z gates, no ancillas, and depth at most $2L+6$ for every $L\ge2$.

\begin{algorithm}[H]
\caption{\textsc{ApplyGamma}: ancilla-free $\Gamma^\star$ on an $L \times L$ grid}
\label{alg:gamma}
\begin{algorithmic}[1]
\State Fold all columns inward into the parity-XOR basis
\State Fold all rows inward into the parity-XOR basis
\State Apply the universal grid-wide nearest-neighbor \CZ/$Z$ phase pattern
\State Unfold all rows
\State Unfold all columns
\Statex Residual \CZ gates use idle slots of 2--4 (Lemma~\ref{lem:app-zero-extra-placement}).
\end{algorithmic}
\end{algorithm}

For every $L\ge2$, the fixed parity-correction operator $\Gamma^\star$ on the $L\times L$ row-snake grid uses zero ancillas and has nearest-neighbor CNOT depth
\begin{equation}
d_{\Gamma^\star}(L)\le
\begin{cases}
2L+4,&L\text{ even},\\
2L+6,&L\text{ odd}.
\end{cases}
\label{eq:gamma-depth}
\end{equation}
For $L\ge7$, Appendix~\ref{app:gamma} gives a closed-form phase and schedule attaining these bounds. For $2\le L\le6$, separately synthesized, code-verified circuits, their gate lists, and the verification scripts are provided in the GitHub repository (\texttt{common/gamma\_small.py}, \texttt{common/tests/test\_gamma.py}). Appendix~\ref{app:gamma} also proves correctness, and Section~\ref{sec:optimality} shows that the leading coefficient two is optimal.

\subsection{Correctness and resources}
\label{sec:full-alg}

Algorithm~\ref{alg:main} implements $\Fop{\pi}$ for any permutation $\pi \in S_N$. Fermionic permutations compose ($\Fop{\pi \circ \sigma} = \Fop{\pi} \cdot \Fop{\sigma}$), so it suffices to check each stage. In the row stages, horizontal neighbors are JWT-adjacent, so each \FSWAP is a valid fermionic adjacent transposition; rows are JWT-disjoint and execute in parallel. In the column stage, applying Eq.~\eqref{eq:gamma-conjugation} to each round and cancelling adjacent factors of $(\Gamma^\star)^2=I$ converts the bare vertical \FSWAPs into full fermionic \SWAPs, and columns are likewise parallel. Hall's routing gives $\pi = \pi^{\mathrm{rowB}} \circ \pi^{\mathrm{col}} \circ \pi^{\mathrm{rowA}}$, so the three stages compose to $\Fop{\pi^{\mathrm{rowB}}}\cdot\Fop{\pi^{\mathrm{col}}}\cdot\Fop{\pi^{\mathrm{rowA}}} = \Fop{\pi}$.

\paragraph{Resources.}
Each sorting stage has depth~$2L$ and $\le NL$ entangling gates.
Each folded $\Gamma^\star$ has depth at most $2L+6$ and $O(N)$ gates.
For every $L\ge2$, summing the three sorting stages and the two $\Gamma^\star$ blocks gives
\begin{equation}
d_{\mathrm{FP}}(L)\le10L+12,
\label{eq:fp-depth-exact}
\end{equation}
with $O(N\sqrt{N})$ entangling gates, zero ancillas, no mid-circuit measurements, and only nearest-neighbor interactions.
Classical preprocessing runs in $O(N\log N)$ time~\cite{edgecoloring}.
Table~\ref{tab:resources} gives the per-component breakdown.

\begin{table}[!htb]
\centering
\begin{tabular}{@{}lcc@{}}
\toprule
\textbf{Component} & \textbf{CNOT depth} & \textbf{CNOT gate count} \\
\midrule
Row A (OET sort)           & $2L$          & $\le NL$   \\
$\Gamma^\star$ (first)     & $\le2L+6$ & $O(N)$     \\
Bare Column Sort           & $2L$          & $\le NL$   \\
$\Gamma^\star$ (second)    & $\le2L+6$ & $O(N)$     \\
Row B (OET sort)           & $2L$          & $\le NL$   \\
\midrule
\textbf{Total}             & $\le10L+12$ & $O(N\sqrt{N})$ \\
\bottomrule
\end{tabular}
\caption{Resource breakdown for Algorithm~\ref{alg:main} ($N=L^2$). An \FSWAP contributes CNOT depth two.}
\label{tab:resources}
\end{table}

\subsection{Asymptotic optimality}
\label{sec:optimality}

Under a fixed JWT ordering on the $L\times L$ grid, worst-case fermionic permutation requires depth $\Omega(\sqrt{N})$ and $\Omega(N\sqrt{N})$ nearest-neighbor entangling gates. For the depth bound, we reduce ordinary qubit routing to fermionic permutation. In each row, choose $\lfloor L/2\rfloor$ disjoint horizontal pairs, order the two modes of each pair according to the fixed JWT, and encode one logical qubit as $\alpha\ket{01}+\beta\ket{10}$. Given any permutation $\sigma$ of these logical qubits, let $\widetilde{\sigma}$ move the corresponding mode pairs as blocks while preserving their internal order. Because every block contains exactly one fermion, each inversion of two blocks contributes the same minus sign. Therefore, on the encoded subspace, $\Fop{\widetilde{\sigma}}$ implements the ordinary logical-qubit permutation $\sigma$ up to the global phase $\operatorname{sgn}(\sigma)$. Contracting each pair gives an $L\times\lfloor L/2\rfloor$ grid $H$, so a depth-$d$ implementation of every fermionic permutation would yield an $O(d)$-depth protocol for arbitrary qubit routing on $H$. Devulapalli et al.~\cite{devulapalli2022routing} prove that such routing in the LOCC model satisfies $\mathrm{rt}_{\mathrm{LOCC}}(H)\ge 2/c(H)-1$. Since $c(H)=O(1/L)$, some fermionic permutation requires depth $\Omega(L)=\Omega(\sqrt{N})$, even with a constant number of locally attached ancillas per site, mid-circuit measurements, and classical feedforward. For the gate count, reflecting the logical blocks across the vertical axis transfers $\Theta(L\min(c,L-c))$ logical qubits across the cut after column $c$; summing this constant-capacity requirement over all cuts gives $\Omega(L^3)=\Omega(N\sqrt{N})$ gates.

For the $\Gamma$ correction itself, Appendix~\ref{app:gamma} gives a sharper light-cone bound. Any ancilla-free nearest-neighbor diagonal phase correction that satisfies Eq.~\eqref{eq:parity-encoding} must couple some pair of opposite corners. Consequently, conjugating a Pauli $X$ at one corner produces support at the other. A depth-$d$ local circuit can spread support by graph distance at most $d$, while opposite corners are distance $2L-2$ apart. Therefore $d\ge2L-2$, and the folded depth $2L+O(1)$ in Eq.~\eqref{eq:gamma-depth} is optimal up to an additive constant.

\section{Extensions to Bravyi--Kitaev and Parity Encodings}
\label{sec:extensions}

Our algorithm is formulated under the Jordan--Wigner encoding, but it extends naturally to the Bravyi--Kitaev (BK)~\cite{bk} and Parity~\cite{parity} encodings via efficient basis transformations. We assume $L$ is a power of two throughout this section. Converting between these encodings on the 2D grid can be performed in $O(\sqrt{N})$ depth by exploiting a Hilbert space-filling curve, so the overall depth for fermionic permutations under BK or Parity remains $O(\sqrt{N})$.

\subsection{Ternary tree encodings}
\label{sec:ternary-tree}

A broad class of fermion-to-qubit encodings can be described by \emph{ternary trees}~\cite{Jiang_2020, Yu_2025, Miller_2023, chiew2024ternarytreetransformationsequivalent, vlasov2024mutualtransformationsarbitraryternary}.
A full ordered ternary tree has $N$ parent nodes (each with exactly three children) and $2N{+}1$ leaves.
An encoding is specified by a tree shape, a labeling of the $N$ parent nodes by qubits, and a labeling of $2N$ leaves by Majorana operators, with one leaf unused (Figure~\ref{fig:ternary-tree-intro}).
The Pauli string for a given Majorana operator $\gamma_j$ is read off from its root-to-leaf path: at each parent node labeled by qubit~$k$, the left, middle, and right branches contribute $X_k$, $Y_k$, and $Z_k$ respectively.
Fermionic creation operators are recovered from pairs of Majorana operators via Eq.~\eqref{eq:majorana-to-ca}.

The JW, BK, and Parity encodings are all \emph{binary-shaped} ternary trees---meaning all qubit nodes lie in the binary subtree formed by left and right children---with qubits labeled by inorder traversal and Majorana leaves labeled left-to-right~\cite{Miller_2023}.
They differ only in tree shape: JW is a right spine, BK is a logarithmic-depth binary tree, and Parity is a left spine (Figure~\ref{fig:parity-bk-jw}).

\begin{figure}[t]
    \centering
    \includegraphics[width=\columnwidth]{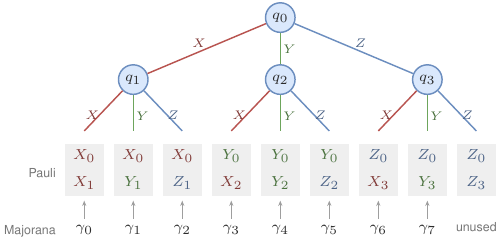}
    \caption{A ternary-tree encoding maps Majorana operators to Pauli strings via root-to-leaf paths. It is specified by the tree shape, a labeling of the $N$ parent nodes by qubits, and a labeling of $2N$ leaves by Majorana operators.}
    \label{fig:ternary-tree-intro}
\end{figure}

A $\CNOT_{j,k}$ gate between a parent--child pair in the tree corresponds to a standard tree rotation~\cite{Yu_2025, vlasov2024mutualtransformationsarbitraryternary}.
These rotations preserve both the inorder traversal of qubit nodes and the left-to-right ordering of Majorana leaves~\cite{constantinides2025lowdepthfermionroutingancillas}. They therefore transform between encodings without relabeling (Figure~\ref{fig:parity-bk-jw}).

\begin{figure*}[t]
    \centering
    \includegraphics[width=\textwidth]{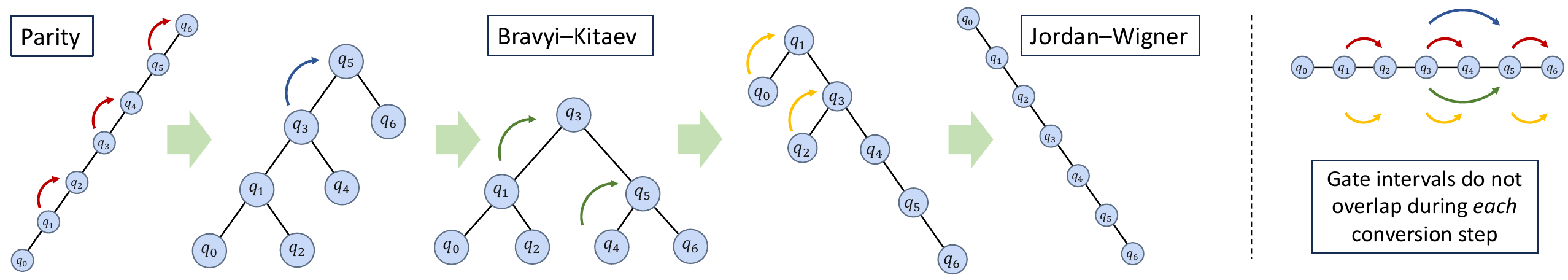}
    \caption{Parity $\to$ BK $\to$ JW transformation via parallel tree rotations.
    Each round consists of disjoint \CNOT gates whose intervals on the inorder-traversal line are non-overlapping.
    For BK$\to$JW, the maximum interval length decays geometrically; for Parity$\to$BK, it grows in the reverse order.
    On the Hilbert-curve-ordered grid, the per-round depth is $O(\sqrt{d_{\max}})$, and the geometric series sums to $O(\sqrt{N})$.}
    \label{fig:parity-bk-jw}
\end{figure*}

\subsection{Fermionic permutation under BK and Parity}
\label{sec:ferm-perm-bk-parity}

\begin{figure*}[t]
    \centering
    \includegraphics[width=\textwidth]{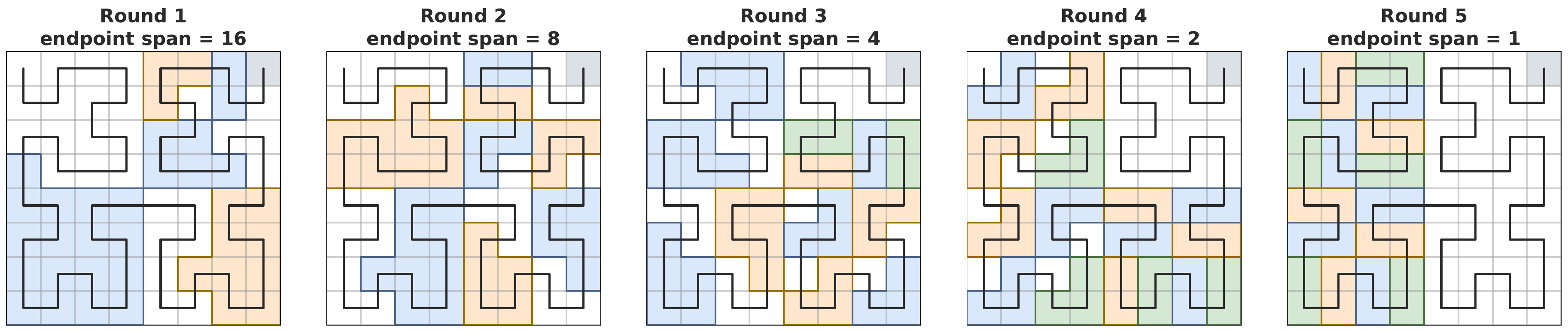}
    \caption{BK$\to$JW CNOT rounds for $N = 63$ on an $8 \times 8$ Hilbert-curve-ordered grid; the black path is the Hilbert curve and the one grey cell is the unused $64$th site. Each panel shows one round. Every outlined block is one inorder-traversal interval, coloured so that no two touching blocks share a colour (four colours always suffice, so a colour is reused by blocks that do not touch). Each interval lies in a pairwise-disjoint dyadic container whose Hilbert image is a compact rectangle, so all route--CNOT--unroute circuits in one round execute in parallel.}
    \label{fig:hilbert-bk-jw}
\end{figure*}

Both BK-to-JW and Parity-to-BK conversions proceed by $O(\log N)$ rounds of parallel tree rotations, each round consisting of disjoint \CNOT gates~\cite{constantinides2025lowdepthfermionroutingancillas}. Within each round, (1)~each qubit appears in at most one gate, and (2)~the inorder-traversal intervals spanned by different gates are non-overlapping. Appendix~\ref{app:hilbert-disjoint} proves that these intervals admit pairwise-disjoint dyadic containers of comparable length.

To execute these rounds efficiently on the 2D grid, we lay out qubits according to the \emph{Hilbert space-filling curve}: inorder-traversal index~$i$ maps to grid cell~$H(i)$. Under this layout, each dyadic container of length $O(d)$ becomes a compact rectangle of diameter $O(\sqrt{d})$. Moving one endpoint along a shortest path, applying the tree-rotation \CNOT, and reversing the path implements the rotation without ancillas and restores all intervening data. Because the rectangles are disjoint, all rotations in one round execute concurrently; Figure~\ref{fig:hilbert-bk-jw} shows the rounds of a BK$\to$JW conversion for $N=63$.

For BK$\to$JW, the total depth sums as
\[
\sum_{r=1}^{\log_2 N} O\!\left(\sqrt{N/2^r}\right)
\;=\; O(\sqrt{N}).
\]
Parity$\to$BK uses the reflected Hilbert placement of Corollary~\ref{cor:parity-jw-depth} and the same series in reverse, also $O(\sqrt{N})$; a Hall permutation then restores the standard Hilbert placement before BK$\to$JW, and restoring the snake layout by Hall routing costs another $O(\sqrt{N})$ depth~\cite{hall-routing}. Each conversion contains $O(N)$ rotations, each routed with at most $O(L)$ nearest-neighbor gates in the Hilbert layout, while changing back to the snake layout uses $O(NL)$ gates; thus the total gate count is $O(N\sqrt{N})$.

Any fermionic permutation in the Bravyi--Kitaev or Parity encoding on an $L\times L$ grid can therefore be implemented in depth $O(\sqrt{N})$ with $O(N\sqrt{N})$ nearest-neighbor gates and zero ancillas, in four steps: (1)~convert from the source encoding to JW using the tree-rotation \CNOT layers on the Hilbert layout; (2)~reroute the Hilbert layout into the snake layout---an ordinary qubit permutation with no fermionic phase correction---using standard 2D routing~\cite{hall-routing}; (3)~apply Algorithm~\ref{alg:main}; and (4)~reverse steps (2) and (1) to return to the original encoding. Each step has $O(\sqrt{N})$ depth, so the total remains $O(\sqrt{N})$.

Constantinides et al.~\cite{constantinides2025lowdepthfermionroutingancillas} show that any product-preserving ternary tree encoding~\cite{chiew2024ternarytreetransformationsequivalent} can be converted to JW in $O(\log^2 N)$ depth under all-to-all connectivity. Extending the $O(\sqrt{N})$ grid bound beyond the BK and Parity cases above---to arbitrary grid sizes and to general product-preserving ternary trees---is the subject of forthcoming work by the authors.

\section{Evaluation}
\label{sec:eval}

We evaluate our protocol on three applications: standalone fermionic permutations (\S\ref{sec:eval-fp}), the 2D fast fermionic Fourier transform (\S\ref{sec:eval-ffft}), and Trotter simulation of the sparse SYK model (\S\ref{sec:eval-syk}).

\subsection{Experimental setup, metrics, and baselines}
\label{sec:eval-setup}

Circuits are constructed on $L\times L$ grids of \texttt{cirq.GridQubit}s with nearest-neighbor connectivity~\cite{cirq} under the snake JW ordering. Every \FSWAP is compiled into exactly two CNOTs (Section~\ref{sec:fswap-bg}). All simulations were conducted on a PC with a 12-core CPU and 32 GB of RAM.

We use CNOT depth $D$ and logical spacetime volume $S=QD$. For the circuit schedule, the estimated no-fault probability is
\begin{equation}
\begin{aligned}
P_{\mathrm{nf}}
&=(1-p)^G(1-p)^I,\\
I&=QD-2G.
\end{aligned}
\label{eq:no-fault-estimate}
\end{equation}
Here $G$ counts CNOT locations and $I$ counts the remaining non-entangling patch locations in the CNOT schedule, including locations with single-qubit gates or idle locations. Thus each \FSWAP contributes two gate-error opportunities. In fault-tolerant operation, syndrome extraction continues on both active and idle data patches; for standalone FP we therefore use the uniform logical error approximation $p\in\{10^{-4},10^{-5}\}$ for CNOT, single-qubit, and idle locations~\cite{Fowler_2012,google2025quantum,beverland2022assessing}. Apart from the Hadamards inside each $\FSWAP$ (Section~\ref{sec:fswap-bg}), the only single-qubit gates in this benchmark are Pauli-$Z$ corrections, which are free Pauli-frame updates~\cite{mckay2017efficient,riesebos2017pauli}; the simulation applies the exact $\FSWAP$ unitary.

For the Clifford FP benchmark, in each CNOT-equivalent layer we apply two-qubit depolarizing noise of rate \(p\) to every active pair and independent single-qubit depolarizing noise of rate \(p\) to every idle qubit. We estimate the process fidelity $F_e=\Pr(P_{\mathrm{data}}=I)$ by propagating the sampled Pauli errors through the noisy forward circuit with Stim's frame simulator~\cite{stim} ($10^6$ shots per circuit) and counting a shot as successful when the final Pauli on every data qubit is the identity; this detects bit and phase errors alike and credits faults that cancel. The other benchmarks are non-Clifford, so we use the estimated no-fault probability of Eq.~\eqref{eq:no-fault-estimate} instead.

For the FFFT and sparse SYK benchmarks, we use a block-volume proxy motivated by FLASQ~\cite{huggins2025flasq}. One block is one logical patch over one logical timestep of $d$ surface-code syndrome-extraction cycles. In our proxy, each scheduled \CNOT layer takes one such timestep; hence $S=QD$ charges every active or idle patch, including ancilla patches. For $d=11$, $p_{\mathrm{phys}}=10^{-3}$, and $p_{\mathrm{th}}=10^{-2}$, the standard surface-code fit gives the logical failure probabilities per patch over one syndrome-extraction cycle and one logical block, respectively~\cite{Fowler_2012,google2025quantum,beverland2022assessing}:
\begin{equation}
\begin{aligned}
p_{\mathrm{cyc}}&=0.03\left(\frac{p_{\mathrm{phys}}}{p_{\mathrm{th}}}\right)^{(d+1)/2}=3\times10^{-8},\\
p_B&=1-(1-p_{\mathrm{cyc}})^d\simeq3.3\times10^{-7}.
\end{aligned}
\label{eq:logical-rate}
\end{equation}
An exact $\pi/4$ rotation consumes one $T$ state, whereas a generic rotation synthesized to accuracy $10^{-7}$ consumes $c_T=70$~\cite{ross2016optimal}, giving
\begin{equation}
n_T=n_{T}^{\mathrm{exact}}+c_Tn_{\mathrm{synth}},\qquad c_T=70,
\label{eq:magic-count}
\end{equation}
We assign one block per $T$-state consumption and $V_T=\alpha n_T$ blocks to magic-state preparation, where $\alpha$ is used to describe the cost ratio of a magic $T$ vs a standard \CNOT. We use $\alpha=1$ in both experiments, motivated by low-overhead magic-state preparation~\cite{litinski2019magic,gidney2024cultivation}. Applying $p_B$ independently to $S$, the $n_T$ consumption blocks, and the $V_T$ preparation blocks gives
\begin{equation}
F_{\mathrm{FT}}=(1-p_B)^{S+n_T+V_T}.
\label{eq:ft-fidelity}
\end{equation}
Its exponent $S+n_T+V_T=S+(1+\alpha)n_T$ ($S+2n_T$ at $\alpha=1$) is the full charged volume behind the no-fault probability, and it is the total we report alongside the routing volume $S$ below.

\paragraph{Baselines.}
We numerically benchmark three FP implementations in this section and compare their asymptotic scaling in Table~\ref{tab:fp-comparison}.
The \emph{1D FSWAP baseline} sorts modes along the snake chain via an OET network, achieving $O(N)$ depth and $O(N^2)$ gates with zero ancillas.
\emph{FP with ancillas} pairs our Hall RCR decomposition with the ancilla-based $\Gamma$ construction of Jiang et al.~\cite{gamma}, which places $L$ traveling ancilla qubits in an extra column adjacent to the grid; this reaches $O(\sqrt{N})$ depth at $13L + 4$ per $\Gamma$ application but spends $L = \sqrt{N}$ ancillas.
\emph{FP without ancillas} (this work) uses the folded circuit of Section~\ref{sec:gamma-impl}: for every $L\ge2$, each $\Gamma^\star$ has depth at most $2L+6$, giving total depth at most $10L+12$.
For asymptotic context, Table~\ref{tab:fp-comparison} also lists our 2D-NN compilations of recent $O(\polylog N)$ fermionic routing protocols~\cite{maskara2025fastsimulationfermionsreconfigurable,constantinides2025lowdepthfermionroutingancillas}, derived in Appendices~\ref{app:reconf-2dnn} and~\ref{app:constantinides-2dnn}. These bounds are listed for analytical context; the two protocols are not benchmarked numerically.

\begin{table}[t]
\centering
\resizebox{\columnwidth}{!}{%
{\setlength{\tabcolsep}{2.5pt}
\begin{tabular}{@{}lccc@{}}
\toprule
\textbf{Method} & \textbf{Depth} & \textbf{Gate count} & \textbf{Ancillas} \\
\midrule
1D \FSWAP baseline  & $2N$          & $O(N^2)$          & 0 \\
FP w/ ancillas~\cite{gamma} & $32\sqrt{N} + 8$ & $O(N^{3/2})$ & $\sqrt{N}$ \\
\textbf{FP w/o ancillas (ours)} & {\boldmath$10\sqrt{N}+12$} & {\boldmath$O(N^{3/2})$} & \textbf{0} \\
\midrule
Reconf 2D$^*$~\cite{maskara2025fastsimulationfermionsreconfigurable} & $\tfrac{17}{2}\sqrt{N}\log_2 N$ & $O(N^{3/2}\log N)$ & $N+\sqrt{N}$ \\
All-to-all 2D~\cite{constantinides2025lowdepthfermionroutingancillas} & $27\sqrt{N}(\log_2 N)^2$ & $O(N^{3/2}\log^2 N)$ & 0 \\
Dynamic switching (concurrent)$^{\ddagger}$~\cite{aigner2026fermionlattices} & $14\sqrt{N}$ & $O(N^{3/2})$ & 0 \\
\bottomrule
\end{tabular}
}%
}
\caption{Fermionic permutation circuit comparison on the $L\times L$ grid ($N=L^2$). Bottom rows show leading-order depth; $^*$ requires measurement and feedforward; $^{\ddagger}$our accounting of the displayed construction (three odd--even stages of depth $2L$ and two encoding switches of depth $4L-2$).}
\label{tab:fp-comparison}
\end{table}

Since $2L^2-(10L+12)=2(L-6)(L+1)$, the bounds in Table~\ref{tab:fp-comparison} break even at $L=6$ ($N=36$), and ours is quadratically smaller for $L>6$. Practical crossover points can differ because scheduled depths may lie below these worst-case guarantees.

\subsection{Fermionic permutation benchmarking}
\label{sec:eval-fp}

\begin{figure}[t]
    \centering
    \includegraphics[width=\columnwidth]{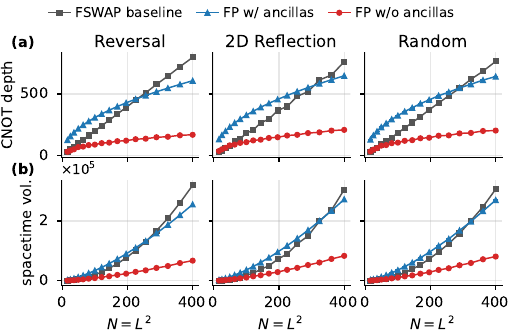}
    \caption{FP circuit resources for reversal, 2D reflection, and random permutations (20 random instances; mean $\pm$ standard deviation).
    The three curves are the 1D \FSWAP{} network (FSWAP baseline), the ancilla-based $\Gamma$ of Ref.~\cite{gamma} (FP w/ ancillas), and our ancilla-free folded $\Gamma$ (FP w/o ancillas). (a)~CNOT depth after scheduling.
    (b)~Spacetime volume. The ancilla-based construction incurs an additional penalty because its $\Theta(\sqrt{N})$ ancillas contribute even while idle.}
    \label{fig:fp-depth-spacetime}
\end{figure}
We benchmark on three permutation families: reversal, 2D reflection (matrix transpose), and random permutations (20 instances per $L$), for every integer $L=4,5,\ldots,20$ ($N=16$ to $400$).

Figure~\ref{fig:fp-depth-spacetime}(a) confirms the expected scaling: the 1D \FSWAP network grows as $O(N)$ while both 2D methods grow as $O(\sqrt{N})$, with folded $\Gamma$ achieving the lowest depth because its per-$\Gamma$ depth is $2L+O(1)$ rather than $13L+O(1)$.

At the depth-bound breakeven $N=36$ ($L=6$), folded $\Gamma$ becomes strictly shallower than 1D and remains so for every larger $L$, with a quadratically growing separation.
Figure~\ref{fig:fp-depth-spacetime}(b) shows that eliminating ancillas compounds the improvement; at \(N=400\), our method reduces the random-permutation mean spacetime volume by \(70\%\) relative to the ancilla method and \(74\%\) relative to the 1D baseline.

\begin{figure}[t]
    \centering
    \includegraphics[width=\columnwidth]{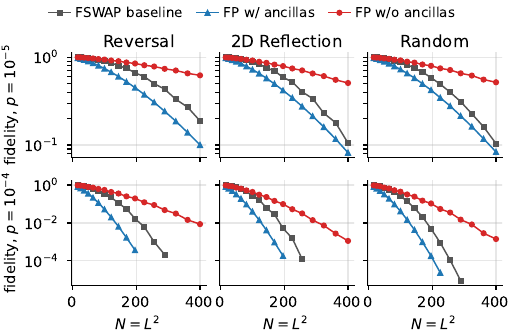}
    \caption{Process fidelity $F_e$ of the noisy FP circuit from Stim frame simulation ($10^6$~shots per circuit), for $p\in\{10^{-5},10^{-4}\}$; curves as in Fig.~\ref{fig:fp-depth-spacetime}, and random points are means over 20 instances. Points with fewer than 100 surviving shots are withheld. The 95\% finite-shot intervals, propagated through the instance means, lie within the markers at every plotted point and are therefore not drawn. Beyond the crossover region, folded $\Gamma$ gives the largest process fidelity.}
    \label{fig:fp-stim}
\end{figure}
Figure~\ref{fig:fp-stim} shows the same family-dependent crossover: folded $\Gamma$ first becomes best for reversal at $N=25$, for random permutations at $N=36$, and for transpose at $N=64$. At $p=10^{-5}$ and $N=400$, its process fidelity remains above $0.50$ for every family, compared with $0.10$--$0.19$ for 1D and $0.08$--$0.10$ with ancillas; at $p=10^{-4}$ all methods decay earlier.

\begin{figure*}[t]
  \centering
  \includegraphics[width=\textwidth]{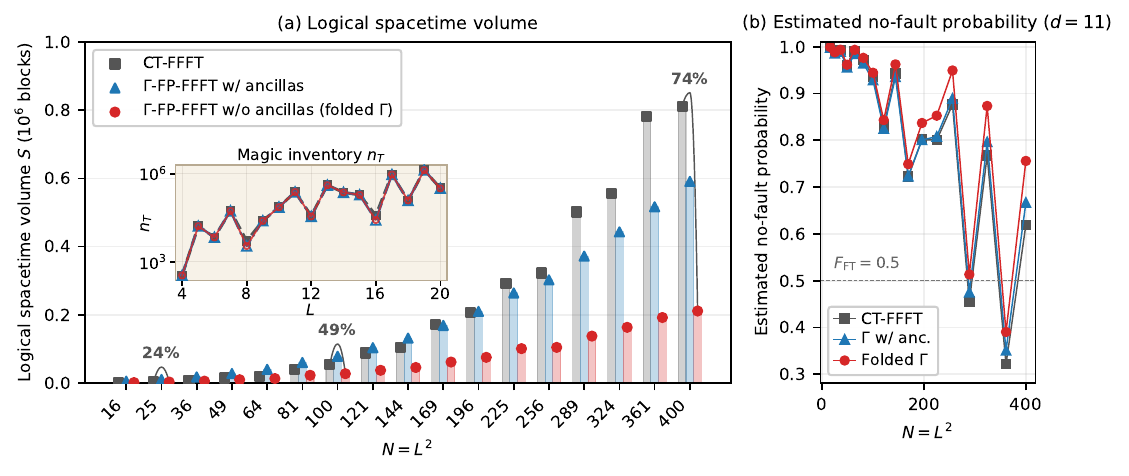}
  \caption{Early-fault-tolerant resources for the mode-preserving FFFT implementations.
  \textbf{(a)}~Logical spacetime volume $S=QD$; markers and percentage annotations compare folded-$\Gamma$ with CT-FFFT using $S$. The inset gives the magic inventory $n_T$ for all three methods; $\alpha=1$ makes $V_T=n_T$, so magic charges $2n_T$ blocks in Eq.~\eqref{eq:ft-fidelity}.
  \textbf{(b)}~Estimated no-fault probability from Eq.~\eqref{eq:ft-fidelity} at code distance $d=11$.}
  \label{fig:ffft-ft}
\end{figure*}

\subsection{2D fast fermionic Fourier transform}
\label{sec:eval-ffft}

The fermionic Fourier transform (FFFT) maps between the position and momentum bases of a fermionic system, acting on creation operators as
\begin{equation}
  \mathit{FFFT}^\dagger\, \tilde{a}_k^\dagger\, \mathit{FFFT}
  \;=\;
  \frac{1}{\sqrt{N}} \sum_{n=0}^{N-1} e^{-i\frac{2\pi}{N} n k}\, \hat{a}^\dagger_n,
  \label{eq:ffft-def}
\end{equation}
and is a standard primitive in fermionic quantum simulation~\cite{Verstraete_2009, ffft, maskara2025fastsimulationfermionsreconfigurable, constantinides2025lowdepthfermionroutingancillas}. We use the nearest-neighbor CT-FFFT implemented in OpenFermion~\cite{ffft-ct, openfermion} as the 1D baseline. It applies one length-$N$ transform along the snake chain with $O(N)$ depth and $O(N^2)$ two-qubit gates. Prior $O(\sqrt{N})$-depth constructions on a 2D grid require either $\Theta(\sqrt{N})$ traveling ancillas~\cite{gamma} or $O(N)$ compact-encoding ancillas~\cite{compactencoding}.

\paragraph{Construction.}
We use the mode-preserving 2D construction throughout. Starting from the canonical snake layout, we first reverse the odd rows to obtain raster order. We then apply $\Gamma^\star$, all $L$ column FFFTs in parallel, and $\Gamma^\star$ again, followed by the single-qubit twiddle rotations and all $L$ row FFFTs in parallel. A final FP restores every Fourier mode to the canonical snake layout. Each row and column transform is the same length-$L$ OpenFermion CT-FFFT used in the baseline. The column gates conserve occupation, so Eq.~\eqref{eq:parity-encoding} supplies their missing JW phases; since $(\Gamma^\star)^2=I$, the same outer pair corrects the complete column transform.

Each length-$L$ line transform has $O(L)$ depth and at most $O(L^2)$ two-qubit gates. Including the layout changes, $\Gamma^\star$ pair, and final FP, the complete construction therefore has $O(L)=O(\sqrt{N})$ depth and $O(L^3)=O(N\sqrt{N})$ two-qubit gates. Its non-Clifford rotation count is also at most $O(L^3)$; these rotations come only from the line FFFTs and twiddle phases, while the layout, $\Gamma^\star$, and FP stages are Clifford. Table~\ref{tab:ffft-comparison} summarizes the asymptotic comparison.

\begin{table}[t]
    \centering
    \resizebox{\columnwidth}{!}{%
    \begin{tabular}{@{}lccc@{}}
    \toprule
    \textbf{Method} & \textbf{Depth} & \textbf{Gates} & \textbf{Ancillas} \\
    \midrule
    CT-FFFT~\cite{ffft-ct, openfermion}      & $O(N)$         & $O(N^2)$        & 0 \\
    Compact encoding~\cite{compactencoding}  & $O(\sqrt{N})$  & $O(N\sqrt{N})$  & $O(N)$ \\
    Gamma-FP-FFFT w/ ancillas~\cite{gamma}   & $O(\sqrt{N})$  & $O(N\sqrt{N})$  & $\sqrt{N}$ \\
    \textbf{Gamma-FP-FFFT (ours)}            & $O(\sqrt{N})$  & $O(N\sqrt{N})$  & \textbf{0} \\
    \bottomrule
    \end{tabular}
    }
    \caption{FFFT circuit comparison on the $L\times L$ grid ($N=L^2$).}
    \label{tab:ffft-comparison}
\end{table}

\paragraph{Results.}
We compare three mode-preserving implementations for every integer $L=4,5,\ldots,20$: the 1D OpenFermion CT-FFFT, the 2D construction above with the ancilla-based $\Gamma$, and the same construction with our folded ancilla-free $\Gamma$.

Figure~\ref{fig:ffft-ft}(a) shows that folded $\Gamma$ first beats CT-FFFT in both $S$ and the charged volume $S+n_T+V_T$ at $L=5$ ($N=25$), one grid size before the $L=6$ circuit-level breakeven, and remains better through $L=20$. This application-level crossover comes from executing the row and column FFFTs in parallel instead of repeatedly reshuffling the full transform along one line. The ancilla method becomes consistently better than CT-FFFT in $S$ from $L=15$ ($N=225$). At $L=5$, including magic changes the reduction from $23.7\%$ in $S$ to $2.4\%$ in $S+n_T+V_T$. At $L=20$, folded $\Gamma$ reduces $S$ by $74.0\%$ and $S+n_T+V_T$ by $41.6\%$.

The inset shows the magic inventory $n_T$. The two 2D methods coincide because they differ only in their Clifford $\Gamma$ circuits. Their magic comes from the shared length-$L$ line FFFTs and twiddle rotations, whereas CT-FFFT applies one length-$N$ transform. The variation in $n_T$ therefore comes from the FFFT arithmetic rather than fermionic routing.

Figure~\ref{fig:ffft-ft}(b) gives the corresponding estimated no-fault probability. Folded $\Gamma$ again wins from $L=5$ onward and reaches $F_{\mathrm{FT}}=0.756$ at $L=20$, compared with $0.619$ for CT-FFFT and $0.667$ for the ancilla method. The nonmonotonicity follows from the same rotation inventory as Panel~(a).

\begin{figure*}[t]
  \centering
  \includegraphics[width=\textwidth]{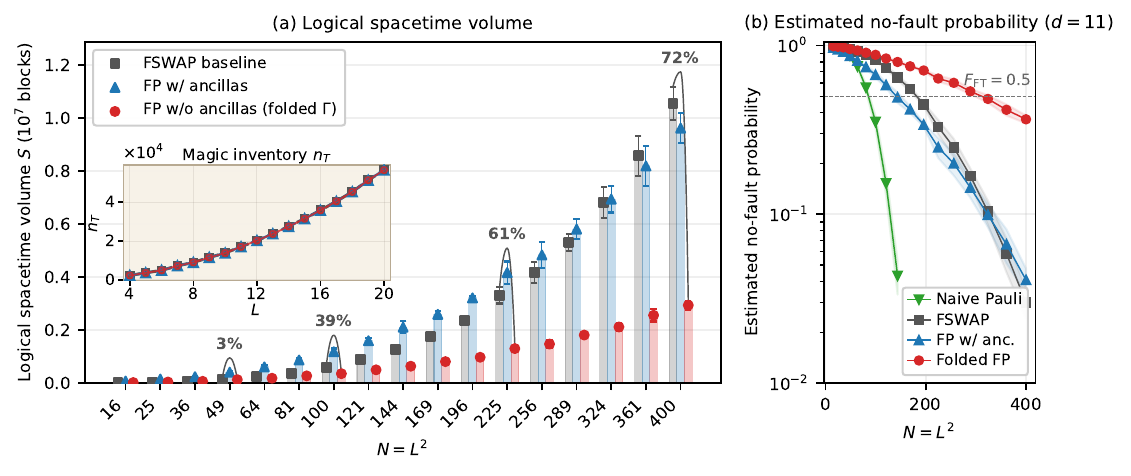}
  \caption{Fault-tolerant resources for one sparse-SYK Trotter step at $k=1$, averaged over 10 random instances for every integer $L=4,5,\ldots,20$.
  \textbf{(a)}~Mean logical spacetime volume $S=QD$ for the 1D \FSWAP network, ancilla-based $\Gamma$~\cite{gamma}, and folded ancilla-free $\Gamma$; markers, error bars, and percentage annotations compare folded-$\Gamma$ with the 1D \FSWAP baseline using $S$. The inset gives the raw magic volume cost $n_T$ for all three methods, which coincides because their interaction workload is common.
  \textbf{(b)}~Estimated no-fault probability from Eq.~\eqref{eq:ft-fidelity}, including the naive-Pauli baseline (no coloring or FP, $O(N^2)$ depth and $O(N^3)$ volume). The logarithmic scale omits points below $10^{-2}$ rather than extending them as flat tails. At $N=400$, the folded method retains $F_{\mathrm{FT}}=0.366$, compared with $0.030$ for 1D and $0.041$ with ancillas.}
  \label{fig:syk}
\end{figure*}

\subsection{Sparse SYK simulation}
\label{sec:eval-syk}

The sparse Sachdev--Ye--Kitaev (SYK) model is a strongly interacting fermionic Hamiltonian
\begin{equation}
  H = \sum_{\{i < j < k < l\} \in \mathcal{S}} J_{ijkl}\, \gamma_i \gamma_j \gamma_k \gamma_l,
  \label{eq:syk}
\end{equation}
which has long-range four-body interactions and is believed to be dual to toy models of quantum gravity~\cite{syk, Sachdev_1993, xu2020sparsemodelquantumholography}.
Here $\gamma_0, \ldots, \gamma_{2N-1}$ are the Majorana operators of Section~\ref{sec:jw}, $\mathcal{S}$ is a random sparse subset with each quartet included independently with probability $k \cdot 2N / \binom{2N}{4}$ (yielding $2kN$ terms in expectation), and couplings $J_{ijkl}\sim\mathcal{N}(0,6/(2N)^3)$ set an illustrative energy scale.  Their magnitudes do not affect the coloring, routing circuits, or resource counts reported here.
We fix $k=1$ as a representative constant-sparsity workload, for which the expected Hamiltonian term count is linear in $N$.

\paragraph{Construction.}
We implement one first-order Trotter step.
Under JW, each term becomes a multi-qubit Pauli rotation implemented by a nearest-neighbor \CNOT ladder and inverse around a $Z$ rotation. Unsupported sites in the span are dirty bridges restored by the inverse sweep; with $s$ supported and $b$ bridge sites, the two sweeps use $2(s-1+2b)$ \CNOTs. We color mode-conflicting terms and pack each color class so that the complete snake intervals, including bridges, are disjoint. A Trotter step then applies a forward FP, the parallel rotations, and the inverse FP for each color class. The rotations and color classes are identical across FP methods, isolating routing.
\paragraph{Results.}
\emph{Logical spacetime volume (Panel a).} Figure~\ref{fig:syk}(a) shows that folded $\Gamma$ has lower mean logical spacetime volume than the 1D network for every sampled $L>6$, while the ancilla method first does so at $N=361$. At $N=400$, folded $\Gamma$ reduces $S$ by $72.2\%$ relative to 1D and $69.5\%$ relative to the ancilla method.

For folded $\Gamma$, the magic charge falls from $(n_T+V_T)/S=27.0\%$ at $N=16$ to $3.8\%$ at $N=400$; at the latter size, the charged volume $S+n_T+V_T$ is lower by $71.4\%$ than 1D and $68.7\%$ than the ancilla method.

\emph{No-fault probability (Panel b).} This volume gap caps the reachable system size. Under the no-fault model of Eq.~\eqref{eq:ft-fidelity}, the folded method has mean $F_{\mathrm{FT}}=0.366$ at $N=400$, compared with $0.030$ for 1D and $0.041$ for the ancilla-based method. Its last sampled mean above $0.5$ is $N=289$, compared with $N=169$ for 1D, $N=121$ for the ancilla-based method, and $N=81$ for naive Pauli.

\section{Conclusion}
\label{sec:conclusion}
Implementing arbitrary fermionic permutations on 2D nearest-neighbor hardware is the dominant routing bottleneck in fermionic simulation, yet no prior method simultaneously saturates the $\Omega(\sqrt{N})$ lower bound and avoids ancillas. We close this gap via algorithm--architecture co-design: Hall's Row--Column--Row decomposition tailors the routing to the grid, while a folded phase-polynomial $\Gamma$ realizes the universal diagonal parity correction in depth $2L+O(1)$. The resulting ancilla-free protocol has depth $10L+O(1)$ and remains asymptotically optimal even against LOCC protocols granted $O(N)$ ancillas and feedforward. For power-of-two $L$, a Hilbert-curve layout extends the $O(\sqrt{N})$ result to the Bravyi--Kitaev and Parity encodings. The depth bounds intersect at $L=6$ ($N=36$), the benchmarks cross over near that size depending on workload and metric, and our method is quadratically more efficient for larger system sizes. The folded-$\Gamma$ improvement changes routing volume without changing qubit count or magic-state inventory, so its effect is largest in routing-intensive fermionic applications such as the FFFT and sparse SYK. Across all three benchmarks, our ancilla-free method achieves the lowest logical spacetime volume and highest estimated no-fault probability at scale; at \(N=400\), it reduces spacetime by \(72\text{--}74\%\) relative to the existing 1D baseline, without requiring routing ancillas.

Two directions follow. The Hilbert-layout argument of Section~\ref{sec:extensions} covers the BK and Parity trees for power-of-two $L$; extending the $O(\sqrt N)$ grid bound to every product-preserving ternary-tree encoding~\cite{chiew2024ternarytreetransformationsequivalent} and to arbitrary $L$ is the subject of forthcoming work by the authors. Second, fermionic permutation is one route to locality on the grid; the other is to pay qubits for it up front, as the superfast encodings do~\cite{bk,setia2018superfast}, and an end-to-end comparison of total simulation cost on 2D hardware under one fault-tolerant accounting would locate the regime in which each approach wins.

\section*{Data and Code Availability}
The code, benchmark data, and plotting scripts are available at \url{https://github.com/dtlics/fast-fermionic-permutation}; the version used for this article is archived on Zenodo at \href{https://doi.org/10.5281/zenodo.22646475}{doi:10.5281/zenodo.22646475}.

\begin{acknowledgments}
We thank the PennyLane team for helpful discussions on implementing the numerical experiments. This work is primarily supported by the U.S. Department of Energy, Office of Science, National Quantum Information Science Research Center, Co-design Center for Quantum Advantage (C2QA) under Contract No. DE-SC0012704. YD acknowledges partial support by the National Science Foundation (under awards CCF-2312754 and CCF-2338063), by Quantum CT (under NSF Engines award ITE-2302908), by Air Force Office of Scientific Research MURI (FA9550-26-1-B036), by Boehringer Ingelheim, and NSF NQVL-ERASE (under award OSI-2435244). External interest disclosure: YD is a consultant and equity holder of D-Wave Quantum, Inc.
\end{acknowledgments}

\appendix

\section{Folded \texorpdfstring{$\Gamma$}{Gamma} Construction and Proof}
\label{app:gamma}

This appendix proves the two claims about the universal correction
$\Gamma^\star$ used in the main construction.  First, conjugation by
$\Gamma^\star$ supplies the Jordan--Wigner parity for \emph{every} vertical
nearest-neighbor exchange.  Second, $\Gamma^\star$ has an ancilla-free
nearest-neighbor implementation of depth
\[
2L+4\quad\text{for even }L,
\qquad
2L+6\quad\text{for odd }L.
\]
The closed-form construction is given for $L\ge 7$; the finitely many cases
$2\le L\le 6$ use separately synthesized representatives, as stated in
Section~\ref{sec:gamma-impl}.  We conclude with a matching leading-order lower
bound of $2L-2$.

The proof has three conceptual steps.  We first construct a dense quadratic
phase whose change under a vertical exchange is exactly the required snake-JW
parity.  We then apply a two-sided cumulative-parity transform along both grid
axes; in this folded basis, the dense phase becomes a local pattern consisting
of one diagonal in each grid cell together with a sparse set of ordinary grid
edges.  Finally, we realize the cell diagonals in eight central layers and pack
all remaining phases into otherwise unused locations of the fold and unfold.

Throughout the appendix, all polynomial arithmetic is over $\mathbb F_2$, and
rows and columns are indexed by $0,\ldots,L-1$.  For sums of Boolean
polynomials, $\oplus$ denotes addition modulo two.  For sets of phase edges, the
same symbol denotes symmetric difference.

For a binary matrix $M$, define the Boolean quadratic form
\begin{equation}
q_M(z)=z^TMz
=\operatorname{diag}(M)^Tz
\oplus\!\bigoplus_{u<v}(M_{uv}\oplus M_{vu})z_uz_v,
\label{eq:app-boolean-quadratic}
\end{equation}
where $z_u^2=z_u$.  Its off-diagonal terms can be implemented by \CZ gates on
\[
E(M)=\bigl\{\{u,v\}:u<v,\ M_{uv}\oplus M_{vu}=1\bigr\},
\]
while $\operatorname{diag}(M)$ specifies single-qubit $Z$ phases.

\subsection{A universal phase polynomial}

We begin by writing the parity that a vertical exchange must acquire.  Let
$\tau_{r,c}$ exchange the occupations at $(r,c)$ and $(r+1,c)$.  In the
row-snake JW ordering, the sites strictly between these two endpoints lie

\begin{itemize}
    \item in rows $r$ and $r+1$ to the right of column $c$, when $r$ is even;
    \item in rows $r$ and $r+1$ to the left of column $c$, when $r$ is odd.
\end{itemize}

Accordingly, define
\begin{equation}
\chi_{r,c}(x)=
\begin{cases}
\displaystyle\bigoplus_{d>c}(x_{r,d}\oplus x_{r+1,d}),&r\text{ even},\\[2mm]
\displaystyle\bigoplus_{d<c}(x_{r,d}\oplus x_{r+1,d}),&r\text{ odd}.
\end{cases}
\label{eq:app-required-vertical-parity}
\end{equation}
The parity-encoding condition of Eq.~\eqref{eq:parity-encoding} is therefore
satisfied by a phase polynomial $f$ whenever
\[
f(x)\oplus f(\tau_{r,c}x)
=(x_{r,c}\oplus x_{r+1,c})\chi_{r,c}(x).
\]
The prefactor makes the identity automatic when the two exchanged occupations
are equal.

Let $\mathbb I(\mathcal P)$ denote the indicator of a statement $\mathcal P$,
and define
\begin{equation}
\begin{aligned}
U_{c,d}&=\mathbb I(c<d),\\
S_{r,s}&=\mathbb I(s<r)
\oplus\mathbb I(r=s\text{ and }r\text{ is odd}).
\end{aligned}
\label{eq:app-gamma-us}
\end{equation}
Thus row $r$ of $S$ contains ones strictly to the left of the diagonal when
$r$ is even, and up to and including the diagonal when $r$ is odd.  In
row--column tensor order, set
\begin{equation}
f_0(x)=q_{S\otimes U}(x)
=\bigoplus_{c<d}\ \bigoplus_{r,s}S_{r,s}x_{r,c}x_{s,d}.
\label{eq:app-gamma-f0}
\end{equation}

\begin{lemma}[The canonical phase has the required vertical derivative]
\label{lem:gamma-vertical-derivative}
For every vertical grid edge,
\begin{align}
f_0(x)\oplus f_0(\tau_{r,c}x)
={}&(x_{r,c}\oplus x_{r+1,c})\nonumber\\[-1mm]
&{}\times
\begin{cases}
\displaystyle\bigoplus_{d>c}(x_{r,d}\oplus x_{r+1,d}),&r\text{ even},\\[2mm]
\displaystyle\bigoplus_{d<c}(x_{r,d}\oplus x_{r+1,d}),&r\text{ odd}.
\end{cases}
\label{eq:app-gamma-derivative}
\end{align}
Consequently, $\Gamma_0\ket{x}=(-1)^{f_0(x)}\ket{x}$ satisfies
Eq.~\eqref{eq:parity-encoding} simultaneously for all vertical grid edges.
\end{lemma}

\begin{figure}[t]
  \centering
  \includegraphics{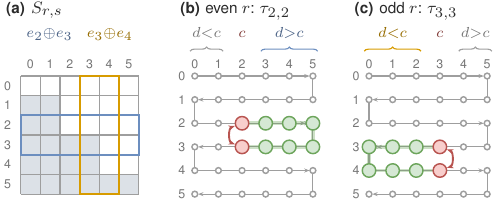}
  \caption{Exchange derivative of the canonical phase $f_0=q_{S\otimes U}$
  (Lemma~\ref{lem:gamma-vertical-derivative}) for $L=6$; the strict upper
  triangle $U_{c,d}=\mathbb I(c<d)$ is not shown.
  (a)~The matrix $S$ of Eq.~\eqref{eq:app-gamma-us}, shaded cells being ones.
  Adjacent rows differ only for even $r$ [blue, $S_{2,*}\oplus S_{3,*}=e_2\oplus
  e_3$, Eq.~\eqref{eq:app-S-row-difference}], and adjacent columns $s,s+1$
  differ only for odd $s$ [orange, $S_{*,3}\oplus S_{*,4}=e_3\oplus e_4$,
  Eq.~\eqref{eq:app-S-column-difference}].
  (b)~Exchanging $(2,2)\leftrightarrow(3,2)$ (red pair, even $r$): only the
  columns $d>c$ contribute, with the row difference as coefficient.
  (c)~Exchanging $(3,3)\leftrightarrow(4,3)$ (odd $r$): only the columns $d<c$
  contribute, with the column difference as coefficient.
  In both cases the gray brace marks the columns that contribute nothing, and
  the surviving variables (green) are the sites strictly between the exchanged
  pair along the snake order, i.e., $\chi_{r,c}$ of
  Eq.~\eqref{eq:app-required-vertical-parity}.}
  \label{fig:app-derivative}
\end{figure}

\begin{proof}
Write $a=x_{r,c}$ and $b=x_{r+1,c}$.  Only monomials containing $a$ or $b$
can change when these two variables are exchanged.  We separate those
monomials according to whether their other column lies to the right or to the
left of $c$; Figure~\ref{fig:app-derivative} illustrates both cases for $L=6$.

For a column $d>c$, the changing part is determined by the two rows
$S_{r,*}$ and $S_{r+1,*}$.  Directly from Eq.~\eqref{eq:app-gamma-us},
\begin{equation}
S_{r,*}\oplus S_{r+1,*}
=\mathbb I(r\text{ is even})(e_r\oplus e_{r+1})^T.
\label{eq:app-S-row-difference}
\end{equation}
Hence the total contribution from all columns to the right is
\[
(a\oplus b)\,\mathbb I(r\text{ is even})
\bigoplus_{d>c}(x_{r,d}\oplus x_{r+1,d}).
\]

For a column $d<c$, the variable in column $d$ is the first factor in
Eq.~\eqref{eq:app-gamma-f0}, so the relevant difference is between columns of
$S$ rather than rows.  Again from Eq.~\eqref{eq:app-gamma-us},
\begin{equation}
S_{*,r}\oplus S_{*,r+1}
=\mathbb I(r\text{ is odd})(e_r\oplus e_{r+1}).
\label{eq:app-S-column-difference}
\end{equation}
The contribution from all columns to the left is therefore
\[
(a\oplus b)\,\mathbb I(r\text{ is odd})
\bigoplus_{d<c}(x_{r,d}\oplus x_{r+1,d}).
\]
Exactly one of these two expressions is nonzero, which proves
Eq.~\eqref{eq:app-gamma-derivative}.  When $a\oplus b=1$, the remaining factor
is precisely the snake-JW parity in Eq.~\eqref{eq:app-required-vertical-parity};
when $a=b$, both sides vanish.  This is Eq.~\eqref{eq:parity-encoding}.
\end{proof}

The required exchange derivatives do not determine a unique phase polynomial.
If $h(x)$ is invariant under arbitrary permutations of the rows within each
column, then $h(x)=h(\tau_{r,c}x)$ for every vertical exchange, and
$f_0\oplus h$ is equally valid.  We will use this freedom to choose a
representative whose folded circuit is especially easy to schedule.  Importantly,
this choice depends only on $L$ and the fixed snake ordering, not on the
particular column permutation being implemented.

\subsection{The two-sided fold and the resulting local phase}

Let $m=\lceil L/2\rceil$.  On a one-dimensional register, define the two-sided
cumulative-parity transform $A$ by
\begin{equation}
(Ax)_i=
\begin{cases}
\displaystyle\bigoplus_{j=0}^{i}x_j,&i<m,\\[1mm]
\displaystyle\bigoplus_{j=i}^{L-1}x_j,&i\ge m.
\end{cases}
\label{eq:app-fold-map}
\end{equation}
The left half computes prefix parities from the left boundary, while the right
half computes suffix parities from the right boundary.  The two CNOT wavefronts
move toward the center in parallel.  We call this basis change a
\emph{two-sided fold}.  Figure~\ref{fig:app-fold} shows the fold of one line
together with the matrices $A$ and $D=A^{-1}$.

\begin{figure}[t]
  \centering
  \includegraphics{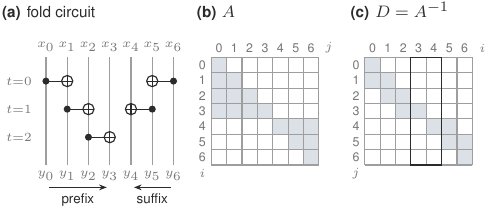}
  \caption{The two-sided fold of one line for $L=7$, $m=4$.
  (a)~Two \CNOT wavefronts accumulate prefix parities from the left boundary
  and suffix parities from the right boundary in $\max\{m-1,L-m-1\}=3$ layers
  $t=0,1,2$.
  (b)~The matrix $A$ of Eq.~\eqref{eq:app-fold-map}, $y=Ax$; shaded cells are
  the entries equal to one, so row $i$ lists the original bits contained in
  the folded bit $y_i$.  Throughout, $i$ indexes the folded bits $y$ and $j$
  the original bits $x$, so $D$ is written $D_{j,i}$.
  (c)~Its inverse $D$ [Eq.~\eqref{eq:app-fold-inverse}], $x=Dy$: every column
  $d_i$ is supported on two neighboring coordinates, except the two singletons
  $d_{m-1}=e_{m-1}$ and $d_m=e_m$ at the framed fold columns $m-1$, $m$.}
  \label{fig:app-fold}
\end{figure}

Let $D=A^{-1}$.  Its columns are
\begin{equation}
d_i=
\begin{cases}
e_i\oplus e_{i+1},&i<m-1,\\
e_i,&i\in\{m-1,m\},\\
e_i\oplus e_{i-1},&i>m.
\end{cases}
\label{eq:app-fold-inverse}
\end{equation}
Apply $A$ first along every column and then along every row, and write
\[
y=(A\otimes A)x=:Gx
\]
for the fully folded variables.  Each axis requires
$\max\{m-1,L-m-1\}$ CNOT layers, so the complete fold has depth
\begin{equation}
g(L)=2\max\{m-1,L-m-1\}
=\begin{cases}L-2,&L\text{ even},\\L-1,&L\text{ odd}.
\end{cases}
\label{eq:app-gather-depth}
\end{equation}

In the folded basis, the canonical phase is
\[
f_0(G^{-1}y)=q_{R\otimes C}(y),
\qquad
R=D^TSD,
\quad
C=D^TUD,
\]
because
\begin{equation}
(D\otimes D)^T(S\otimes U)(D\otimes D)=R\otimes C.
\label{eq:app-fold-congruence}
\end{equation}
The next lemma gives the exact sparsity pattern of these two factors.

For later geometric statements, let $h_{r,c}$ denote the horizontal edge from
$(r,c)$ to $(r,c+1)$ and let $v_{r,c}$ denote the vertical edge from $(r,c)$
to $(r+1,c)$.  A \emph{row band} $r$ is the pair of adjacent rows $r,r+1$.  A
\emph{cell} $(r,c)$ is the unit square with upper-left corner $(r,c)$, where
$0\le r,c<L-1$.

\begin{lemma}[Exact local form of the folded canonical phase]
\label{lem:gamma-local-factors}
Let $p$ be the largest odd integer not exceeding $m$.  Then
\begin{align}
C={}&\bigoplus_{i=0}^{L-2}e_ie_{i+1}^T
\oplus\!\bigoplus_{\substack{0\le i<L\\i<m-1\ {\rm or}\ i>m}}e_ie_i^T,
\label{eq:app-c-closed}\\
R={}&e_pe_p^T\oplus
\bigoplus_{r=0}^{L-2}
\begin{cases}
e_re_{r+1}^T,&s_r=+1,\\
e_{r+1}e_r^T,&s_r=-1,
\end{cases}
\label{eq:app-r-closed}
\end{align}
where the sign $s_r$ records the orientation of the surviving row link:
\begin{equation}
s_r=
\begin{cases}
(-1)^r,&r<m-1,\\
-1,&r=m-1,\\
(-1)^{r+1},&r\ge m.
\end{cases}
\label{eq:app-row-slopes}
\end{equation}

Equivalently, the off-diagonal phase edges of $q_{R\otimes C}$ are exactly:
(i) one diagonal in every cell, with its orientation determined by $s_r$;
(ii) every vertical edge except those in columns $m-1$ and $m$; and
(iii) the full horizontal path $\{h_{p,c}:0\le c<L-1\}$.
\end{lemma}

\begin{figure}[t]
  \centering
  \includegraphics{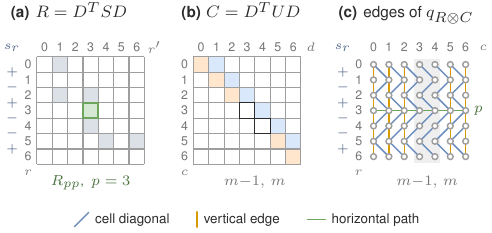}
  \caption{The folded factors of Lemma~\ref{lem:gamma-local-factors} for
  $L=7$ ($m=4$, $p=3$).
  (a)~$R=D^TSD$ has one entry per row link, above the diagonal when $s_r=+1$
  and below it when $s_r=-1$ (the sign $s_r$ is written between rows $r$ and
  $r+1$), plus the single diagonal entry $R_{pp}$ (green).
  (b)~$C=D^TUD$ has the full superdiagonal (blue) and the diagonal (orange)
  except at the two framed entries of the fold columns $m-1$, $m$.
  (c)~The off-diagonal edges of $q_{R\otimes C}$: a row link times the
  superdiagonal of $C$ gives one cell diagonal with orientation $s_r$ (blue); a
  row link times the diagonal of $C$ gives the vertical edges (orange), absent
  in the fold columns $m-1$, $m$ (gray band); $R_{pp}$ times the superdiagonal
  gives the horizontal path in row $p$ (green).}
  \label{fig:app-factors}
\end{figure}

\begin{proof}
Figure~\ref{fig:app-factors} shows the two factors and the resulting edge set
for $L=7$.  It is useful to read $u^TUv$ as
$\bigoplus_{a<b}u_av_b$.  Every vector $d_i$ in
Eq.~\eqref{eq:app-fold-inverse} is supported on two neighboring coordinates,
except for the two singleton vectors at the fold.  Evaluating
$C_{ij}=d_i^TUd_j$ therefore gives
\begin{equation}
C_{ij}=1
\quad\Longleftrightarrow\quad
\bigl(j=i+1\bigr)
\ \text{or}\
\bigl(i=j\notin\{m-1,m\}\bigr),
\label{eq:app-C-entrywise}
\end{equation}
which is Eq.~\eqref{eq:app-c-closed}.  For example, away from the fold,
$d_i$ contains two adjacent coordinates.  Pairing it with itself contributes
one ordered pair, pairing it with $d_{i+1}$ contributes three ordered pairs,
and pairing it with a more distant $d_j$ contributes an even number of
ordered pairs; modulo two, only the diagonal and first superdiagonal survive.  The two
singleton vectors remove the two central diagonal entries.

Next write
\[
S=U^T+E,
\qquad
E_{rr}=\mathbb I(r\text{ is odd}).
\]
Then
\[
R=C^T+D^TED.
\]
Both terms are tridiagonal, so $R$ is tridiagonal.  Its diagonal is especially
simple.  Away from the fold, $d_i$ contains exactly one odd coordinate, so
$(D^TED)_{ii}=1$, which cancels the unit diagonal of $C^T$.  At $i=m-1$ and
$i=m$, the vectors are singletons; exactly the odd one of these two indices
survives.  This index is $p$, proving that $R_{pp}=1$ is the only diagonal
entry.

For every adjacent pair $i,i+1$,
\begin{align*}
R_{i,i+1}\oplus R_{i+1,i}
&=d_i^T(S+S^T)d_{i+1}\\
&=1,
\end{align*}
because $S+S^T$ is the all-ones matrix with zero diagonal and the neighboring
vectors $d_i,d_{i+1}$ have odd cross-pairing.  Thus exactly one orientation of
each row link survives.  Evaluating $D^TED$ on the neighboring supports gives
rightward orientation for even $r<m-1$, leftward orientation for odd
$r<m-1$, leftward orientation at $r=m-1$, and the reversed alternation below
the fold.  This is precisely Eq.~\eqref{eq:app-row-slopes}, and proves
Eq.~\eqref{eq:app-r-closed}.

Finally, combine the supports of $R$ and $C$.  An off-diagonal entry of
$R\otimes C$ can arise in only three ways.  An oriented off-diagonal of $R$
combined with the superdiagonal of $C$ gives one of the two diagonals of a
cell.  The same row link combined with a diagonal entry of $C$ gives a
vertical edge, except in the two columns where $C$ has no diagonal.  The sole
diagonal entry $R_{pp}$ combined with the superdiagonal of $C$ gives the
horizontal path in row $p$.  No other off-diagonal terms are possible.
\end{proof}

\subsection{Packing the folded phase into depth \texorpdfstring{$2g(L)+8$}{2g(L)+8}}
\label{app:gamma-zero-extra}

We now construct the circuit for $L\ge7$.  It has five consecutive blocks:
a column fold, a row fold, eight midpoint layers, an inverse row fold, and an
inverse column fold.  The fold and unfold contribute $2g(L)$ layers.  The purpose of this subsection
is to show that every required phase can be inserted into the eight midpoint
layers or into unused locations of the fold/unfold layers, so no additional
two-qubit layer is needed.

Put $h=m-1$.  Label the horizontal row-fold layers, in execution order, by
$H_0,\ldots,H_{h-1}$; label the eight midpoint layers by
$B_0,\ldots,B_7$; and label the inverse-row-fold layers, again in execution
order, by $U_0,\ldots,U_{h-1}$.  A ``slot'' below means one of these existing
two-qubit layers.  Gates assigned to the same slot must form a matching.
Figure~\ref{fig:app-layers} draws every layer of the circuit for $L=7$.

\begin{figure*}[t]
  \centering
  \includegraphics{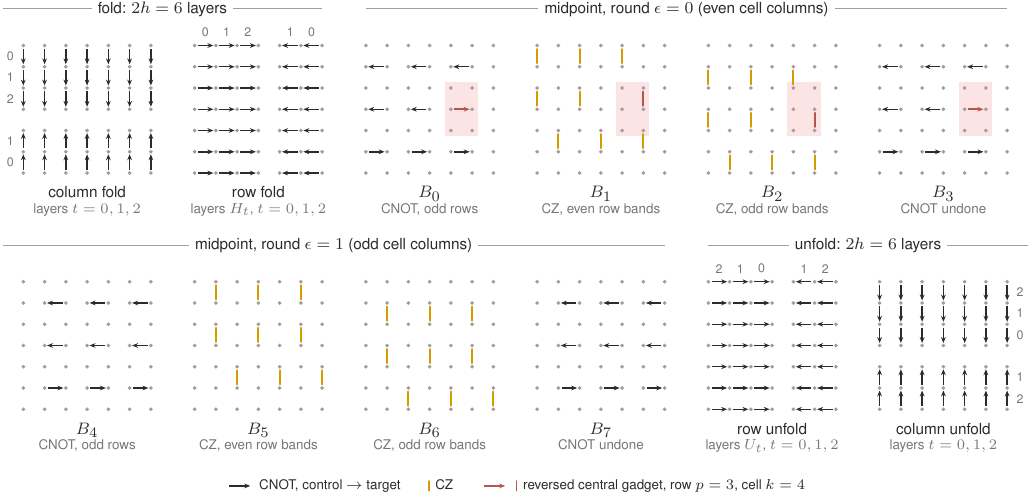}
  \caption{Every two-qubit layer of the folded $\Gamma^\star$ circuit for
  $L=7$ ($h=3$), with gates drawn as edges.  The column fold and the row fold
  consist of three \CNOT layers each, and the unfolds apply the same layers in
  reverse order; the small digits give the layer index $t=0,1,2$ of each
  wavefront in execution order (layers $H_t$ of the row fold and $U_t$ of the
  row unfold), innermost first in the unfolds.  The eight midpoint layers
  $B_0,\ldots,B_7$ form two four-layer gadget rounds, one per cell-column
  parity $\epsilon$, with \CNOT gates in the odd rows and \CZ gates on the
  even and odd row bands at the target columns.  The reversed central gadget
  on row $p=3$, cell $k=4$ is drawn in red; the shaded block marks the cell
  pair $\mathcal G_L=\{(p-1,k),(p,k)\}$ adjacent to its horizontal \CNOT,
  whose phase the reversal changes (cf.\ Fig.~\ref{fig:app-residual}).
  Lemma~\ref{lem:app-zero-extra-placement} inserts the residual \CZ gates into
  unused locations of the row fold, the midpoint layers, and the row unfold.}
  \label{fig:app-layers}
\end{figure*}

\subsubsection{The eight-layer midpoint circuit}

\begin{figure}[tb]
  \centering
  \includegraphics{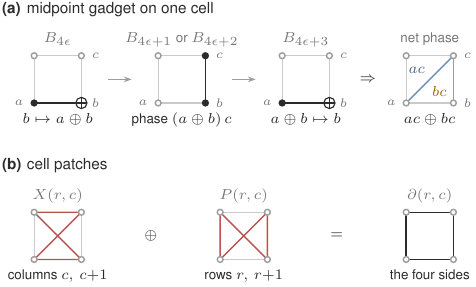}
  \caption{The midpoint gadget and the cell patches on one unit cell
  (independent of $L$).
  (a)~A horizontal \CNOT in the odd row (drawn pointing right, $t_q=+1$; a
  left-pointing \CNOT puts the \CZ on the left side and produces the other
  diagonal), a vertical \CZ at its target column in $B_{4\epsilon+1}$ for the
  even row band drawn (for an odd band the \CZ sits in $B_{4\epsilon+2}$ and
  the cell is mirrored vertically), and the inverse \CNOT restore the bits
  and, by Eq.~\eqref{eq:app-cnot-cz-cnot-identity}, leave the phase
  $(a\oplus b)c=ac\oplus bc$: one cell diagonal (blue) and one vertical side
  (orange).
  (b)~The cell patches $X(r,c)$ and $P(r,c)$, the $K_{2,2}$ on the two cell
  columns and on the two cell rows: both contain the two diagonals, so their
  symmetric difference is the cell boundary $\partial(r,c)$
  [Eq.~\eqref{eq:app-cell-boundary}].}
  \label{fig:app-gadget}
\end{figure}

The midpoint uses the elementary identity
\begin{equation}
(a\oplus b)c=ac\oplus bc.
\label{eq:app-cnot-cz-cnot-identity}
\end{equation}
Thus a horizontal \CNOT, followed by a vertical \CZ at its target, followed by
the inverse \CNOT, creates the desired cell diagonal together with one vertical
side of that cell, while restoring the folded bits.
Figure~\ref{fig:app-gadget}(a) shows one gadget.

For each odd row $q$, define
\[
t_q=
\begin{cases}
s_q,&q<L-1,\\-s_{q-1},&q=L-1.
\end{cases}
\]
For a cell-column parity $\epsilon\in\{0,1\}$, layer $B_{4\epsilon}$ applies
the horizontal \CNOT matching in every odd row on cells
$c\equiv\epsilon\pmod2$.  The gates point right when $t_q=+1$ and left when
$t_q=-1$.  Layers $B_{4\epsilon+1}$ and $B_{4\epsilon+2}$ apply the vertical
\CZ matchings on, respectively, the even-indexed and odd-indexed row bands, at
the target column of the neighboring horizontal \CNOT.  Layer
$B_{4\epsilon+3}$ undoes the horizontal matching.  Denote this initial
midpoint circuit by $N_L^{(0)}$.

The two values of $\epsilon$ cover all cell columns, and the two vertical-
matching layers cover all row bands.  Therefore $N_L^{(0)}$ produces the
required diagonal or anti-diagonal in every cell.  By
Eq.~\eqref{eq:app-cnot-cz-cnot-identity}, it also produces one vertical edge
per cell gadget.

A small modification near the fold will make the remaining phases easier to
pack.  Define
\begin{equation}
p=
\begin{cases}m,&m\text{ odd},\\m-1,&m\text{ even},\end{cases}
\qquad
k=
\begin{cases}m-2,&m\text{ odd},\\m,&m\text{ even}.
\end{cases}
\label{eq:app-central-gadget}
\end{equation}
In the gadget on row $p$ and cell $k$, reverse both the opening and closing
horizontal \CNOTs and move the two neighboring vertical \CZ gates from the old
target column to the new target column.  Call the modified midpoint circuit
$N_L$.  It still acts as the identity on the folded computational-basis bits;
only its phase polynomial changes.

\subsubsection{A schedule-friendly representative}

Let
\[
\sigma_c=\bigoplus_{r=0}^{L-1}x_{r,c}
\]
be the parity of original column $c$.  For $c<d$, define
\begin{equation}
J(c,d)=\sigma_c\sigma_d
=\bigoplus_{r,s=0}^{L-1}x_{r,c}x_{s,d}.
\label{eq:app-j-block}
\end{equation}
Because $J(c,d)$ depends only on whole-column parities, it is invariant under
all within-column row permutations.  Adding such terms therefore does not
change any vertical exchange derivative.

We choose
\begin{equation}
\begin{aligned}
f^\star={}&f_0
\oplus\!\bigoplus_{\substack{d\in\{m-2,m-1\}\\c<d}}J(c,d)\\
&{}\oplus
\begin{cases}
\displaystyle\bigoplus_{\substack{c\in\{m,m+1\}\\d>c}}J(c,d),&L\text{ even},\\
0,&L\text{ odd}.
\end{cases}
\end{aligned}
\label{eq:app-gamma-representative}
\end{equation}
By Lemma~\ref{lem:gamma-vertical-derivative},
$\Gamma^\star\ket{x}=(-1)^{f^\star(x)}\ket{x}$ is a universal vertical
parity correction.

Although Eq.~\eqref{eq:app-gamma-representative} is dense in the original
basis, its effect after folding is confined to a constant-size neighborhood of
the center.  In the fully folded variables $y=Gx$, set
\[
T_c=y_{m-1,c}\oplus y_{m,c}.
\]
Using Eq.~\eqref{eq:app-fold-inverse} first in the row direction and then in the
column direction gives
\begin{equation}
\sigma_c=
\begin{cases}
T_0,&c=0,\\
T_{c-1}\oplus T_c,&1\le c\le m-1,\\
T_c\oplus T_{c+1},&m\le c\le L-2,\\
T_{L-1},&c=L-1.
\end{cases}
\label{eq:app-column-parity-fold}
\end{equation}
The telescoping is now explicit.  For $d\le m-1$,
$\bigoplus_{c<d}\sigma_c=T_{d-1}$; for $c\ge m$,
$\bigoplus_{d>c}\sigma_d=T_{c+1}$.  Hence, after discarding linear terms, each
selected sum in Eq.~\eqref{eq:app-gamma-representative} reduces to a product
of neighboring $T$ variables.  Writing $q\equiv_2q'$ when two Boolean
polynomials have the same off-diagonal quadratic part, we obtain
\begin{equation}
\begin{aligned}
f^\star\oplus f_0\equiv_2{}&
T_{m-3}T_{m-2}\oplus T_{m-2}T_{m-1}\\
&{}\oplus
\begin{cases}
T_mT_{m+1}\oplus T_{m+1}T_{m+2},&L\text{ even},\\
0,&L\text{ odd}.
\end{cases}
\end{aligned}
\label{eq:app-star-telescope}
\end{equation}
For example,
$\sigma_{m-2}\left(\bigoplus_{c<m-2}\sigma_c\right)
=(T_{m-3}\oplus T_{m-2})T_{m-3}
\equiv_2T_{m-3}T_{m-2}$, and the other three terms are identical telescoping
calculations.

\subsubsection{The sparse residual phase}

We next compare the phase produced by the midpoint circuit with the target
folded phase.  The following cell notation makes the comparison transparent.
For a cell $(r,c)$, define two four-edge sets:

\begin{itemize}
    \item $X(r,c)$ is the $K_{2,2}$ whose bipartition is the two cell columns;
    it contains the top and bottom horizontal edges and both diagonals;
    \item $P(r,c)$ is the $K_{2,2}$ whose bipartition is the two cell rows;
    it contains the left and right vertical edges and both diagonals.
\end{itemize}

The two sets share exactly the two diagonals.  Their symmetric difference is
therefore the boundary of the cell:
\begin{equation}
X(r,c)\oplus P(r,c)
=h_{r,c}\oplus h_{r+1,c}\oplus v_{r,c}\oplus v_{r,c+1}
=: \partial(r,c).
\label{eq:app-cell-boundary}
\end{equation}
This identity is the geometric reason for the particular representative in
Eq.~\eqref{eq:app-gamma-representative}; Figure~\ref{fig:app-gadget}(b)
shows the three edge sets.

Let $E_{\rm fold}(f)$ denote the off-diagonal edge set of $f$ in the fully
folded variables, and let $E(N)$ denote the phase-edge set produced by a
midpoint circuit $N$.  First compare the canonical target with the unmodified
midpoint:
\[
B_L:=E_{\rm fold}(f_0)\oplus E(N_L^{(0)}).
\]
By Lemma~\ref{lem:gamma-local-factors}, both sets contain the same chosen
diagonal in every cell, so all cell diagonals cancel.  In row band $r$, the
vertical sides generated by the midpoint gadgets cover every column except one
outer column $o_r$.  In contrast, the canonical target contains every vertical
edge except those in columns $m-1$ and $m$, and it contains the horizontal path
in row $p$.  Therefore
\begin{equation}
\begin{aligned}
B_L={}&\{h_{p,c}:0\le c<L-1\}\\
&{}\cup\{v_{r,c}:0\le r<L-1,\ c\in\{m-1,m,o_r\}\},
\end{aligned}
\label{eq:app-canonical-residual}
\end{equation}
where
\[
\begin{aligned}
o_r&=
\begin{cases}
L-1,&r<m-1,\\
0,&r>m-1,
\end{cases}
&& (r\ne m-1),\\[1mm]
o_{m-1}&=
\begin{cases}
L-1,&m\text{ even},\\
0,&m\text{ odd}.
\end{cases}
\end{aligned}
\]
Thus $B_L$ has a simple exact description: three vertical edges in each row
band and one horizontal edge across every column gap.

Equation~\eqref{eq:app-star-telescope} toggles the column-bipartite patches
$X(r,c)$ on the cell set
\begin{equation}
\begin{aligned}
\mathcal F_L={}&\{(m-1,m-3),(m-1,m-2)\}\\
&{}\cup
\begin{cases}
\{(m-1,m),(m-1,m+1)\},&L\text{ even},\\
\varnothing,&L\text{ odd}.
\end{cases}
\end{aligned}
\label{eq:app-target-cells}
\end{equation}
The central gadget reversal changes the midpoint phase by the row-bipartite
patches on
\[
\mathcal G_L=\{(p-1,k),(p,k)\}.
\]
Indeed, moving a gadget's vertical couplers from one horizontal-CNOT target to
the other toggles the two diagonals and the two vertical sides in each adjacent
cell, which is exactly $P(r,c)$.

Set
$\mathcal P_L=\mathcal F_L\mathbin{\triangle}\mathcal G_L$.
Substituting $X=P\oplus\partial$ from
Eq.~\eqref{eq:app-cell-boundary} gives
\begin{align*}
E_{\rm fold}(f^\star)\oplus E(N_L)
={}&B_L
\oplus\bigoplus_{(r,c)\in\mathcal F_L}X(r,c)\\
&{}\oplus\bigoplus_{(r,c)\in\mathcal G_L}P(r,c)\\
={}&\left(B_L\oplus
\bigoplus_{(r,c)\in\mathcal F_L}\partial(r,c)\right)\\
&{}\oplus\bigoplus_{(r,c)\in\mathcal P_L}P(r,c).
\end{align*}
Therefore
\begin{align}
E_{\rm fold}(f^\star)\oplus E(N_L)
&=R_L\oplus\bigoplus_{(r,c)\in\mathcal P_L}P(r,c),
\label{eq:app-zero-extra-residual}\\
R_L&=B_L\oplus\bigoplus_{(r,c)\in\mathcal F_L}\partial(r,c).
\label{eq:app-rail-boundary}
\end{align}

All cell boundaries in Eq.~\eqref{eq:app-rail-boundary} lie in the central row
band.  The cells in $\mathcal F_L$ occur in adjacent pairs.  Within each pair,
the shared vertical side cancels.  The horizontal path is merely shifted
between the two central rows as it crosses these cells; it still contains
exactly one horizontal edge for each of the $L-1$ column gaps.  The vertical
part still contains exactly three edges in each of the $L-1$ row bands.
Consequently,
\begin{equation}
|R_L|=(L-1)+3(L-1)=4L-4.
\label{eq:app-residual-size}
\end{equation}
This is the complete residual: a constant number of $P$-patches plus the
$O(L)$ individual edges in $R_L$.  Figure~\ref{fig:app-residual} shows $R_L$
and the patches on $\mathcal P_L$ for $L=7$.

\begin{figure}[tb]
  \centering
  \includegraphics{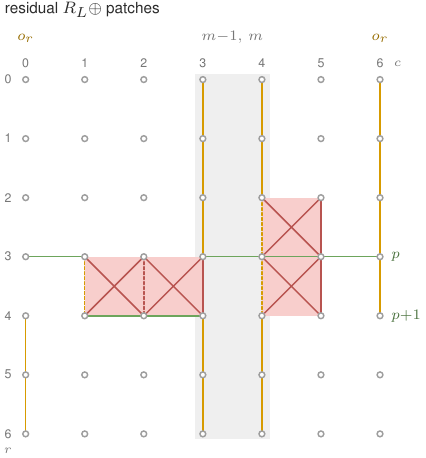}
  \caption{The residual of Eq.~\eqref{eq:app-zero-extra-residual} for $L=7$
  ($m=4$, $p=3$).  $R_L$ consists of three vertical edges (orange) in every
  row band, at columns $m-1$, $m$ and the outer column $o_r$, except in the
  central band $r=m-1$, where for odd $L$ the boundaries of $\mathcal F_L$
  move the edge from column $m-1$ to column $m-3$; and of one horizontal edge
  (green) across every column gap, in row $p$ except across the two cells of
  $\mathcal F_L$, where it is shifted to row $p+1$; $4L-4=24$ edges in total.
  The gray band marks the fold columns $m-1$, $m$.  The row-bipartite patches
  $P(r,c)$ (red) are drawn on the four shaded cells of $\mathcal P_L$: the left
  pair is $\mathcal F_L$ and the right pair, in cell column $k=4$, is
  $\mathcal G_L$.  The sets combine by symmetric difference, so a side that
  occurs twice cancels and is dashed: the dashed orange sides in columns $1$
  and $4$ belong to $R_L$ and to a patch, the dashed red side in column $2$ to
  both patches of $\mathcal F_L$.}
  \label{fig:app-residual}
\end{figure}

\begin{lemma}[The residual requires no additional two-qubit layer]
\label{lem:app-zero-extra-placement}
For every $L\ge7$, every patch and every individual edge in
Eq.~\eqref{eq:app-zero-extra-residual} can be inserted into an unused location
of the row fold, the eight midpoint layers, or the inverse row fold.  Gates
assigned to the same layer are pairwise disjoint.
\end{lemma}

\begin{proof}
We place the residual in four classes.  In each class we verify two facts: the
chosen physical \CZ multiplies the intended folded parity variables, and its
endpoints are unused by every other two-qubit gate in the same layer.

\medskip\noindent\emph{Row-bipartite cell patches.}
Immediately before $H_0$, the horizontal fold has not yet started.  Along each
row, the physical register therefore stores $Dy$.  By
Eq.~\eqref{eq:app-fold-inverse}, an interior physical column $j$ stores
$y_{j-1}\oplus y_j$ on the left of the fold and
$y_j\oplus y_{j+1}$ on the right.  A vertical physical \CZ in row band $r$
therefore realizes $P(r,c)$ by choosing
\begin{equation}
j=
\begin{cases}
c+1,&c<m-1,\\c,&c\ge m.
\end{cases}
\label{eq:app-patch-physical-column}
\end{equation}
Assign every patch in $\mathcal P_L$ to $H_0$, except
$P(m-1,m)$ when $L\equiv3\pmod4$; assign that single exceptional patch to
$U_{h-1}$, where the same $Dy$ basis has been restored.  Patches placed in the
same layer use distinct physical columns.  Even when two patches occupy
adjacent row bands, their distinct columns keep their endpoints disjoint.  For
$L\ge7$, these vertical \CZ gates also avoid the outer horizontal \CNOTs in
$H_0$ or $U_{h-1}$.

\medskip\noindent\emph{Noncentral vertical edges.}
For every row band $r\ne m-1$, the three vertical edges of $R_L$ are still at
the columns shown in Eq.~\eqref{eq:app-canonical-residual}.  Define two pairs
of midpoint slots
\begin{equation}
(A,B)=
\begin{cases}
((B_1,B_2),(B_5,B_6)),&m\text{ even},\\
((B_5,B_6),(B_1,B_2)),&m\text{ odd}.
\end{cases}
\label{eq:app-rail-slot-pairs}
\end{equation}
Above the fold ($r<m-1$), use pair $A$ for column $m-1$ and pair $B$ for
columns $m$ and $L-1$.  Below the fold ($r>m-1$), use pair $B$ for column
$m-1$ and pair $A$ for columns $0$ and $m$.  Within either pair, an even row
band uses the first slot and an odd row band uses the second.

At each selected midpoint slot, both endpoints are either untouched by
the active horizontal \CNOTs or are CNOT sources.  Its stored bit is therefore the
required folded variable.  The midpoint circuit's own vertical \CZ uses the
corresponding target column, so the inserted gate is disjoint from the base
gate.  Alternating the two slots by the parity of $r$ prevents adjacent row
bands from sharing a qubit, and the three residual columns are distinct.
Thus every assigned slot remains a matching.

\medskip\noindent\emph{Noncentral horizontal edges.}
For each horizontal edge of $R_L$ outside the finite fold-seam cases listed
below, define
\begin{equation}
t(c)=
\begin{cases}
c+2,&c\le m-2,\\L-c,&c\ge m.
\end{cases}
\label{eq:app-horizontal-rail-time}
\end{equation}
Place the edge crossing columns $c,c+1$ in layer $H_{t(c)}$.  This is the first
later row-fold layer at which both endpoints already contain their final folded
values and the CNOT wavefront has moved at least two columns away.  For the
noncentral range, $2\le t(c)<h$.  At a fixed value of $t$, there is at most one
assigned edge on the left half and one on the right half; because $t<h$, their
endpoint intervals are disjoint.

\medskip\noindent\emph{The fold seam.}
Only a constant number of edges near the central columns are not covered by
the preceding uniform rules.  Their local CNOT orientation depends only on
$L\bmod4$, so four finite cases suffice.  In the first table, an entry
$(c,B_j)$ means that the central-band vertical edge $v_{m-1,c}$ is placed in
midpoint layer $B_j$; $H$ abbreviates $H_{h-1}$.
\begin{center}
\begin{tabular}{@{}cl@{}}
\toprule
$L\bmod4$ & central-band vertical placements \\
\midrule
$3$ & $(m-3,B_2),(m,B_6),(L-1,B_0)$ \\
$0$ & $(m-3,B_2),(m+2,B_6),(L-1,B_2)$ \\
$1$ & $(0,B_1),(m-3,B_1),(m,H)$ \\
$2$ & $(0,B_1),(m-3,B_1),(m+2,B_5)$ \\
\bottomrule
\end{tabular}
\end{center}

In the next table, $(r,c,B_j)$ means that the horizontal edge $h_{r,c}$ is
placed in $B_j$; $U$ abbreviates $U_0$.
\begin{center}
\small
\begin{tabular}{@{}cl@{}}
\toprule
$L\bmod4$ & central horizontal placements \\
\midrule
$3$ & $(m-1,m-1,B_2),(m-1,m,U)$ \\
    & $(m,m-3,B_0),(m,m-2,B_3)$ \\
\addlinespace[2pt]
$0$ & $(m-1,m-1,B_2),(m,m-3,B_0)$ \\
    & $(m,m-2,B_3),(m,m,B_0),(m,m+1,B_3)$ \\
\addlinespace[2pt]
$1$ & $(m-1,m-3,B_0),(m-1,m-2,B_3)$ \\
    & $(m,m-1,B_5),(m,m,U)$ \\
\addlinespace[2pt]
$2$ & $(m-1,m-3,B_0),(m-1,m-2,B_3)$ \\
    & $(m-1,m,B_0),(m-1,m+1,B_3),(m,m-1,B_5)$ \\
\bottomrule
\end{tabular}
\end{center}

To verify the tables, substitute the four values of $L\bmod4$ into
Eqs.~\eqref{eq:app-row-slopes} and~\eqref{eq:app-central-gadget}.  In every
entry, each endpoint is either unchanged in the indicated layer or is a CNOT
source, so the inserted \CZ realizes the named folded edge.  Entries sharing a
slot have disjoint endpoints and avoid both the base midpoint matching and the
bulk placements above.  All listed indices exist for $L\ge7$.  This completes
the locality, phase, and collision checks for every residual term.
\end{proof}

\subsubsection{Removing the induced linear terms}

The preceding edge accounting matches the off-diagonal quadratic part of the
target phase.  A \CZ applied to two parity masks can additionally generate
linear terms because Boolean variables satisfy $x_i^2=x_i$.  Immediately before
a \CZ gate $e$, write its two endpoint values as $a_e\cdot x$ and
$b_e\cdot x$.  Then
\begin{equation}
(a_e\cdot x)(b_e\cdot x)
=(a_e\mathbin{\odot}b_e)\cdot x
\oplus\bigoplus_{i<j}
(a_{e,i}b_{e,j}\oplus a_{e,j}b_{e,i})x_ix_j,
\label{eq:app-boolean-product}
\end{equation}
where $\odot$ is entrywise multiplication.  Therefore the total unwanted
linear mask is
\begin{equation}
z=\bigoplus_{e\in\mathcal E_{\rm CZ}}
\bigl(a_e\mathbin{\odot}b_e\bigr),
\label{eq:app-zero-extra-z}
\end{equation}
with $\mathcal E_{\rm CZ}$ the set of all \CZ gates in the construction.

After the inverse fold, the CNOT network has restored the original variables
$x$.  Every monomial in $f^\star$ joins two distinct columns, so the target
polynomial has no linear term in this basis.  Applying $Z_i$ exactly when
$z_i=1$ therefore cancels all induced linear terms.  These $Z$ gates form one
parallel single-qubit layer and do not change the two-qubit or CNOT depth.

\begin{theorem}[Folded synthesis]
\label{thm:app-gamma-zero-extra}
For every $L\ge7$, the circuit above implements the selected universal
representative $\Gamma^\star$ using $O(N)$ gates, no ancillas, and
nearest-neighbor CNOT depth
\begin{equation}
d_{\Gamma^\star}(L)\le2g(L)+8
=\begin{cases}
2L+4,&L\text{ even},\\
2L+6,&L\text{ odd}.
\end{cases}
\label{eq:app-zero-extra-depth}
\end{equation}
\end{theorem}

\begin{proof}
\emph{Correct phase.}
Starting from the complete cell-diagonal pattern of $N_L^{(0)}$, the modified
midpoint $N_L$ differs only by the explicitly identified local $P$-patches.
Equations~\eqref{eq:app-zero-extra-residual} and
\eqref{eq:app-rail-boundary} identify the exact remaining off-diagonal phase,
and Lemma~\ref{lem:app-zero-extra-placement} inserts every one of those terms
into existing layers.  Equation~\eqref{eq:app-zero-extra-z} removes the induced
linear terms.  The resulting Boolean phase is therefore exactly $f^\star$,
which is a universal vertical correction by
Eq.~\eqref{eq:app-gamma-representative} and
Lemma~\ref{lem:gamma-vertical-derivative}.

\emph{Restoration of the data basis.}
Every midpoint horizontal \CNOT is uncomputed inside its four-layer gadget,
and the modified central gadget is likewise opened and closed with inverse
CNOTs.  The inverse row and column folds then undo $G$.  Hence the circuit is
diagonal in the original computational basis and acts as
$\ket{x}\mapsto(-1)^{f^\star(x)}\ket{x}$.

\emph{Depth.}
The forward fold has depth $g(L)$, the midpoint has eight layers, and the
inverse fold has depth $g(L)$.  Lemma~\ref{lem:app-zero-extra-placement} adds no
new two-qubit layer, and the final $Z$ correction is single-qubit.  Thus the
depth is $2g(L)+8$, which evaluates to the two cases in
Eq.~\eqref{eq:app-zero-extra-depth}.

\emph{Gate count.}
The fold, inverse fold, and eight-layer midpoint each contain $O(N)$ gates.
The residual contains $O(L)$ two-qubit phase gates by
Eq.~\eqref{eq:app-residual-size}, and the final repair uses at most $N$
single-qubit $Z$ gates.  The total gate count is therefore $O(N)$.
\end{proof}

\subsection{A matching light-cone lower bound}
\label{app:gamma-lower-bound}

\begin{theorem}[Optimal leading coefficient]
\label{thm:gamma-lb}
Consider an ancilla-free nearest-neighbor unitary circuit on an $L\times L$
data grid, with no measurement or feedforward.  Suppose it implements a
diagonal phase correction $\Gamma$ with $\Gamma^2=I$ and satisfies the
universal vertical parity-correction condition.  For $L\ge2$, its two-qubit
layer depth is at least $2L-2$.
\end{theorem}

\begin{proof}
Write
\[
\Gamma\ket{x}=(-1)^{f(x)}\ket{x}
\]
for the Boolean phase implemented by an arbitrary valid correction, and let
$h=f\oplus f_0$.  The exchange derivatives of $f$ and $f_0$ agree on every
vertical edge.  Hence $h(x)=h(\tau_{r,c}x)$ for every adjacent exchange within
a column.  Such exchanges generate all permutations of the rows within that
column, so $h$ depends on each column only through its Hamming weight.

Now set every column except $0$ and $L-1$ to zero, and restrict each of these
two outer columns to contain at most one occupied site.  Let
$a=(r,0)$ and $b=(s,L-1)$.  On this restricted set, the four values
$h(0)$, $h(e_a)$, $h(e_b)$, and $h(e_a\oplus e_b)$ depend only on whether the
two outer columns have weights $(0,0),(1,0),(0,1)$, or $(1,1)$, and not on the
rows $r,s$.  Therefore the mixed Boolean derivative
\begin{equation}
\partial_a\partial_b h(0)
=h(0)\oplus h(e_a)\oplus h(e_b)\oplus h(e_a\oplus e_b)
=:\lambda
\label{eq:app-lower-h-lambda}
\end{equation}
is independent of $r$ and $s$.

For the canonical phase $f_0$, the coefficient of
$x_{r,0}x_{s,L-1}$ is exactly $S_{r,s}$, because $0<L-1$ in
Eq.~\eqref{eq:app-gamma-f0}.  Thus
\begin{equation}
\partial_a\partial_b f(0)=S_{r,s}\oplus\lambda.
\label{eq:app-lower-f-derivative}
\end{equation}
If $\lambda=0$, choose
$a=(L-1,0)$ and $b=(0,L-1)$; then $S_{L-1,0}=1$.  If $\lambda=1$, choose
$a=(0,0)$ and $b=(L-1,L-1)$; then $S_{0,L-1}=0$.  In either case,
$\partial_a\partial_bf(0)=1$, and $a,b$ are opposite corners of the grid.

It remains to translate this nonzero derivative into operator support.  Define
\[
O_a=\Gamma X_a\Gamma.
\]
Because $\Gamma^2=I$,
\begin{equation}
O_a\ket{x}
=(-1)^{f(x)\oplus f(x\oplus e_a)}\ket{x\oplus e_a}.
\label{eq:app-lower-conjugated-X}
\end{equation}
The condition $\partial_a\partial_bf(0)=1$ says that the phase in
Eq.~\eqref{eq:app-lower-conjugated-X} changes when the occupation at $b$ is
flipped.  Therefore $O_a$ acts nontrivially on site $b$; otherwise this matrix
element could not depend on the state of $b$.

Conjugating a single-site operator through one nearest-neighbor two-qubit layer
can enlarge its support by at most one graph edge.  Hence a depth-$d$ circuit
can spread the support of $X_a$ only to the graph-distance-$d$ neighborhood of
$a$.  The opposite corners $a$ and $b$ have Manhattan distance $2L-2$, so
$d\ge2L-2$.
\end{proof}

\section{2D Nearest-Neighbor Compilation of the Reconfigurable Merge-Sort Permutation}
\label{app:reconf-2dnn}

This appendix gives an explicit 2D nearest-neighbor (2D-NN) compilation of the merge-sort fermionic permutation of Maskara et al.~\cite{maskara2025fastsimulationfermionsreconfigurable}. We report its asymptotic depth in Table~\ref{tab:fp-comparison} of Section~\ref{sec:eval-setup}. The original algorithm targets reconfigurable architectures. In this model, any subset of qubits can move to arbitrary positions in one $O(1)$-depth ``reconfigure'' step, independent of distance. On a fixed grid, each such step requires an explicit nearest-neighbor routing circuit. A direct translation loses much of the original asymptotic advantage. We identify four uses of reconfiguration and give a compatible 2D-NN implementation for each. Mid-circuit measurement and classical feedforward already work with 2D-NN connectivity, so we keep them unchanged.

The original algorithm has $\log_2 N$ levels of recursive merging. Each merge combines two sorted halves through parity-cascade builds, reconfigure steps, and one \CZ layer (\textsc{ApplyPhases}) that applies the inversion phases. We use the depth-$1$ \textsc{ApplyPhases} construction from Appendix~A.2 of~\cite{maskara2025fastsimulationfermionsreconfigurable}, rather than the depth-up-to-$L$ construction in their main-text Figure~2.

We use an $L \times (2L{+}1)$ physical grid and assume for simplicity that $L$ is a power of two. The data occupy the left $L \times L$ block, and the ancillas occupy the right $L \times (L{+}1)$ strip. The extra ancilla column ensures that, after the ancillas move in, every data column lies between two ancilla columns. We use the following \CNOT-depth conventions: a \FSWAP layer has depth $2$, a \SWAP layer has depth $3$, and a \CZ layer has depth $1$. Single-qubit gates and classical processing are free.

\subsection{Compilation choices}
\label{app:reconf-choices}

\paragraph{Horizontal-only recursion to single rows.}
The original recursion halves sub-blocks down to single elements. In our 2D-NN compilation, we instead stop at one row of $L$ qubits. We sort that row with the textbook 1D \FSWAP odd--even-transposition (OET) network of depth $2L$ (Section~\ref{sec:oet-bg}), which matches the row diameter up to a constant. At each higher level, we cut the sub-block horizontally. This halves its number of rows while preserving the row direction, so the same column-interleaved ancilla layout can be used at every level. All $L(L{+}1)$ ancillas---$L{+}1$ per row---move in or out in parallel, with depth $2L$ for each move.

\paragraph{Efficient SWAP for ancilla transport.}
The ancillas start and return to $\ket{0}$ around each move. A \SWAP can therefore use two \CNOT gates instead of three, so each move has depth $2L$.

\paragraph{Short--long--short three-stage routing.}
The original implements an arbitrary permutation on each $X \times L$ sub-block with one $O(1)$ reconfigure step. We replace this step with the Hall RCR routing of Section~\ref{sec:hall}, which decomposes the permutation into row, column, and row stages. We orient the grid so that the repeated stage uses the shorter axis. Thus, on an $X \times L$ sub-block, the stages are short--long--short (col $\to$ row $\to$ col). Horizontal recursion always gives $X \le L$, with $X=L$ only at the top level and $X=2$ at the bottom. Short--long--short routing uses $2X+L$ \SWAP layers, whereas long--short--long routing uses $2L+X$. This saves $L-X$ \SWAP layers per route. Since each \SWAP layer has \CNOT depth $3$, each route has \CNOT depth $3(2X+L)$.

\paragraph{Snake-corner cascade compilation.}
The original parity build is a depth-$2$ cascade gadget that uses mid-circuit measurement and classical feedforward~\cite{baumer2025measurement}. It maps the data to the parity basis, where \textsc{ApplyPhases} has depth $1$. Each cascade-\CNOT acts on a pair that is adjacent in the parity ordering. The reconfigurable model brings each such pair together directly. Our column-interleaved layout preserves the snake structure of this adjacency graph at every recursion level. Every cascade-\CNOT within a row is therefore physically local.

The only exception occurs at a row turn of a multi-row sub-block. There, two ancilla cells are stacked vertically between consecutive logical positions on the snake. We treat this stacked pair as one logical edge. Each cascade-\CNOT at a turn is compiled into four physical \CNOTs: hop in, hop across, hop out, and undo the temporary transfer. This sequence skips the middle ancilla. The gadget still has two logical \CNOT layers, but in the worst case each layer expands to depth at most $4$ when it includes a turn. The two parity-group builds run in parallel on disjoint qubit and ancilla sets. The resource count below therefore uses
\begin{equation}
\kappa \;:=\; c_{\mathrm{parity}} + c_{\mathrm{phase}} + c_{\mathrm{undo}} \;\le\; 8 + 1 + 8 \;=\; 17.
\label{eq:reconf-kappa}
\end{equation}

\subsection{Algorithm and resource analysis}
\label{app:reconf-algo}

Algorithm~\ref{alg:reconf-2dnn} gives the recursion. The merge subroutine \textsc{Interleave2D} applies the four substitutions of Section~\ref{app:reconf-choices} to Figure~2 of~\cite{maskara2025fastsimulationfermionsreconfigurable}; Figure~\ref{fig:reconf-2dnn-merge} shows the compiled result, and Table~\ref{tab:reconf-merge} gives its \CNOT-depth accounting.

\begin{figure*}[t]
  \centering
  \includegraphics[width=\textwidth]{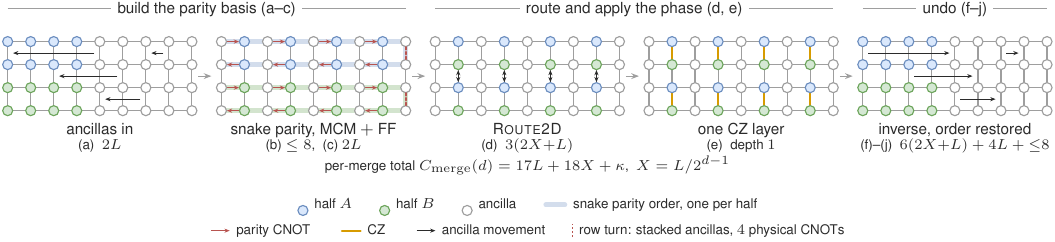}
  \caption{The 2D-NN-compiled merge subroutine \textsc{Interleave2D}, adapted from Figure~2 of Maskara et al.~\cite{maskara2025fastsimulationfermionsreconfigurable} using the four substitutions in Section~\ref{app:reconf-choices}, drawn one tile per step for the top recursion level, where $X=L$: the $X \times L$ sub-block has $L=4$ data columns per row, $L{+}1=5$ interleaved ancilla columns and $X=4$ rows, two per half.
    Table~\ref{tab:reconf-merge} gives the \CNOT depth of every stage; the depths under the tiles are the general expressions.
    \textbf{(a)} The $L{+}1$ ancillas of each row move in parallel from the right strip into the data block, which places every data column between two ancilla columns.
    \textbf{(b)} That layout preserves the snake parity ordering, so each half is built at once along its own snake---threading data and ancilla cells alike and turning at the row end---by the depth-$2$ cascade of~\cite{baumer2025measurement}, which uses mid-circuit measurement and classical feedforward. The two halves build in parallel on disjoint qubit and ancilla sets, and every cascade \CNOT within a row is physically local. The only exceptions are the row turns, where two ancilla cells are stacked between consecutive snake positions; each such \CNOT compiles into four physical \CNOT s, so the two logical layers expand to depth at most~$8$.
    \textbf{(c)} The ancillas move back out.
    \textbf{(d)} \textsc{Route2D} is the Hall RCR 2D-NN permutation of Section~\ref{sec:hall}; it brings parity-basis qubits from the two halves into vertically adjacent pairs, exchanging the two middle rows.
    \textbf{(e)} \textsc{ApplyPhases} is then one depth-$1$ \CZ layer.
    \textbf{(f)--(j)} The inverse: the route is undone, the parity unbuilt, and a final route restores the order.}
  \label{fig:reconf-2dnn-merge}
\end{figure*}

\begin{algorithm}[H]
\caption{2D-NN compiled merge-sort fermionic permutation}
\label{alg:reconf-2dnn}
\begin{algorithmic}[1]
\Procedure{FermPerm}{$\mathrm{grid}: L_{\mathrm{eff}} \times L,\; \pi$}
\If{$L_{\mathrm{eff}} = 1$}
\State \textsc{FSWAP-RowSort}($\mathrm{grid}, \pi$)
       \Comment{depth $2L$}
\State \Return
\EndIf
\State $(\mathrm{top}, \mathrm{bot}) \gets$ horizontal split of grid at $L_{\mathrm{eff}}/2$
\State \Call{FermPerm}{$\mathrm{top}, \pi|_{\mathrm{top}}$}
       \Comment{this and the next call execute in parallel}
\State \Call{FermPerm}{$\mathrm{bot}, \pi|_{\mathrm{bot}}$}
\State \textsc{Interleave2D}($\mathrm{top} \cup \mathrm{bot}, \pi$)
       \Comment{merge}
\EndProcedure
\end{algorithmic}
\end{algorithm}
\begin{table}[!t]
\centering
\renewcommand{\arraystretch}{1.05}
\begin{tabular}{@{}clc@{}}
\toprule
\# & Stage & \CNOT depth \\
\midrule
(1)  & \textsc{MoveAncillasIn}                          & $2L$ \\
(2a) & \textsc{BuildParity} (group $B$)                 & $\le 8$ \\
(2b) & \textsc{BuildPartialParity} (group $A_j$)        & $\le 8$ \\
(3)  & \textsc{MoveAncillasOut}                         & $2L$ \\
(4)  & \textsc{Route2D}($\pi_1$, $X \times L$)          & $3(2X + L)$ \\
(5)  & \textsc{ApplyPhases} (one \CZ layer)           & $1$ \\
(6)  & \textsc{Route2D}($\pi_1^{-1}$, $X \times L$)     & $3(2X + L)$ \\
(7)  & \textsc{MoveAncillasIn}                          & $2L$ \\
(8a) & \textsc{UnbuildPartialParity} ($A_j$)            & $\le 8$ \\
(8b) & \textsc{UnbuildParity} ($B$)                     & $\le 8$ \\
(9)  & \textsc{MoveAncillasOut}                         & $2L$ \\
(10) & \textsc{Route2D}($\pi_{\mathrm{final}}$, $X \times L$) & $3(2X + L)$ \\
\midrule
\multicolumn{2}{l}{\textbf{Per-merge total} $C_{\mathrm{merge}}(d)$} & $17 L + 18 X + \kappa$ \\
\bottomrule
\end{tabular}
\caption{Stage-by-stage \CNOT depth of \textsc{Interleave2D} at recursion level $d$ on an $X \times L$ sub-block, where $X=L/2^{d-1}$. The stages run from top to bottom. Stages (2a)/(2b) and (8a)/(8b) run in parallel on disjoint qubit and ancilla sets.}
\label{tab:reconf-merge}
\end{table}

There are $\log_2 L$ recursion levels above the base case. At level $d$, the $2^{d-1}$ active merges occupy disjoint sub-grids and run in parallel. The depth of the level therefore equals the depth of one merge on an $X \times L$ sub-block, where $X=L/2^{d-1}$.

We now sum the entries in Table~\ref{tab:reconf-merge}. The four ancilla moves contribute $4\cdot2L=8L$. The three routes contribute $3\cdot3(2X+L)=18X+9L$. The parity/phase/undo block contributes $\kappa$ from Eq.~\eqref{eq:reconf-kappa}. With $X=L/2^{d-1}$,
\begin{equation}
C_{\mathrm{merge}}(d) \;=\; 17 L + 18 X + \kappa \;=\; 17 L + \frac{36 L}{2^{d}} + \kappa.
\label{eq:reconf-cmerge}
\end{equation}
We add the base-case row sort of depth $2L$ and sum over all levels. Using $\sum_{d=1}^{\log_2 L}2^{-d}=1-1/L$ gives
\begin{align}
D(L) &= 2L + \sum_{d=1}^{\log_2 L} C_{\mathrm{merge}}(d) \notag \\
     &= 2L + 17L\log_2 L + 36L\!\left(1 - \tfrac{1}{L}\right) + \kappa\log_2 L \notag \\
     &= (17L + \kappa)\log_2 L + 38L - 36. \notag \\
\intertext{Substituting $\kappa \le 17$ from~\eqref{eq:reconf-kappa},}
D(L) &\le 17(L+1)\log_2 L + 38L - 36 \notag \\
     &= \tfrac{17}{2}\sqrt{N}\log_2 N + 38\sqrt{N} + o(\sqrt{N}). \label{eq:reconf-DL-bound}
\end{align}
This gives the $O(\sqrt{N}\log N)$ scaling cited in Section~\ref{sec:intro}. Compared with our $10\sqrt{N}+O(1)$ result in Section~\ref{sec:full-alg}, this baseline has an additional multiplicative $\log_2N$ factor in its asymptotic depth. The compiled circuit uses $L^2$ data qubits and $L(L{+}1)$ ancillas. It therefore uses $L(2L{+}1)=2N+\sqrt{N}$ physical qubits in total.

\section{2D Nearest-Neighbor Compilation of the All-to-All Ancilla-Free Staircase Permutation}
\label{app:constantinides-2dnn}

This appendix gives an explicit 2D nearest-neighbor (2D-NN) compilation of the ancilla-free staircase fermionic permutation of Constantinides et al.~\cite{constantinides2025lowdepthfermionroutingancillas}. We report its asymptotic depth in Table~\ref{tab:fp-comparison} of Section~\ref{sec:eval-setup}. The original algorithm assumes all-to-all connectivity and has depth $O(\log^2 N)$. It uses no ancillas, mid-circuit measurements, or classical feedforward. On a fixed grid, however, every long-range gate needs nearest-neighbor routing. A direct translation therefore loses the original asymptotic advantage. We replace each all-to-all step with a compatible 2D-NN circuit while keeping the construction ancilla-free and free of measurement and feedforward.

The original algorithm first factors any $\pi \in S_N$ classically into $\log_2 N$ sequential levels of \emph{staircase} permutations by recursive bisection~\cite[Lemma~1]{constantinides2025lowdepthfermionroutingancillas}. A staircase acts on a contiguous range $T$ of length $w := |T|$. It is an ordered set of $k$ pairwise-disjoint transpositions $(m_1,n_1),\ldots,(m_k,n_k)$ satisfying
\[
m_1 < \cdots < m_k < n_1 < \cdots < n_k.
\]
The name refers to the step pattern formed by these transpositions in a permutation diagram. At level $d$, there are $2^{d-1}$ staircases on disjoint sub-ranges of length $N/2^{d-1}$, and they run in parallel. Each level handles the transpositions whose endpoints lie on opposite sides of the current midpoint. The algorithm then recurses independently on the two halves, much like one merge level in Appendix~\ref{app:reconf-2dnn}.

Within one staircase, the long-range \FSWAPs can be reordered so that all \SWAPs form one final layer. The remaining circuit has three gadgets built from \CZ-fanout layers~\cite{constantinides2025lowdepthfermionroutingancillas}:
\begin{equation}
\begin{aligned}
\textsc{Staircase}_T &\;=\; \textsc{FinalSWAPs} \cdot G_3 \cdot G_2 \cdot G_1, \\
G_i &\;=\; P^\dagger \cdot L_2 \cdot L_1 \cdot P,
\end{aligned}
\label{eq:constantinides-staircase}
\end{equation}
Each gadget $G_i$ applies the prefix-XOR (parity) transform $P$, two non-overlapping \CZ-fanout layers $L_1$ and $L_2$, and then $P^\dagger$ on their common target sub-range~\cite[Lemma~2]{constantinides2025lowdepthfermionroutingancillas}. The final layer $\textsc{FinalSWAPs}$ contains $k \le w/2$ disjoint \SWAPs. Each of $P$, $L_1$, $L_2$, and $P^\dagger$ is ancilla-free and has depth at most $2\log_2 w$ in \emph{logical CNOT layers}~\cite{moore1999quantum, fang2003quantum, remaud2025ancilla}. These are layers in the original all-to-all circuit, before 2D-NN routing. Thus one staircase contains at most $3 \cdot 4 \cdot 2\log_2 w = 24\log_2 w$ logical \CNOT layers, followed by $\textsc{FinalSWAPs}$.

We use an $L \times L$ grid, where $L$ is a power of two and $N=L^2$. Qubits follow the snake row-major order. Every contiguous sub-range $[a,b] \subseteq [0,N{-}1]$ therefore occupies a contiguous part of the grid. We use the \CNOT-depth conventions of Appendix~\ref{app:reconf-2dnn}: \CZ layer $=1$, \FSWAP layer $=2$, and \SWAP layer $=3$.

\subsection{Compilation choices}
\label{app:constantinides-choices}

\paragraph{Horizontal recursion.}
We index the recursion by $d=1,\ldots,\log_2 L$, as in Appendix~\ref{app:reconf-2dnn}. At level $d$, there are $2^{d-1}$ disjoint sub-grids of shape $X_d \times L$, where $X_d=L/2^{d-1}$. Their staircases run in parallel. The top level has $X_1=L$, while the deepest level above the base has $X_{\log_2 L}=2$. Under the snake layout, each staircase range $T$ has length $w_d \le X_dL=N/2^{d-1}$ and lies within its sub-grid. All routing at level $d$ can therefore remain inside that sub-grid.

\paragraph{Short--long--short routing.}
Before each logical \CNOT layer, we apply a Hall RCR confined to the same sub-grid (Section~\ref{sec:hall}). As in Appendix~\ref{app:reconf-2dnn}, we orient it short--long--short. Its depth is
\[
H_d \;=\; 3(2 X_d + L) \;\le\; 9L,
\]
whose $6X_d$ term halves at each level. The trailing $\textsc{FinalSWAPs}$ operation is also a permutation within the sub-grid, so one more Hall RCR implements it.

\paragraph{Stopping at single rows.}
After $\log_2 L$ levels, each $2 \times L$ sub-grid splits into two $1 \times L$ rows. Each row then contains exactly the modes assigned to it. The remaining task is a 1D fermionic permutation of $L$ modes in each row. We perform these permutations in parallel with the 1D \FSWAP odd--even-transposition (OET) network of Section~\ref{sec:oet-bg}, which has depth $2L$. Continuing the staircase recursion for another $\log_2 L$ levels would instead cost $\Theta(L\log L)$ depth on each row. We therefore stop at single rows, as in Appendix~\ref{app:reconf-2dnn}.

\subsection{Algorithm and resource analysis}
\label{app:constantinides-resources}

Algorithm~\ref{alg:constantinides-2dnn} gives the recursive structure. Each non-base call applies the staircase for its sub-grid, splits the sub-grid horizontally, and recurses on the two halves. The two recursive calls run in parallel. At the base case, each row uses a 1D \FSWAP-OET. At level $d$, the $2^{d-1}$ calls act at the same time on disjoint $X_d \times L$ sub-grids. The depth of the level is therefore the depth of one staircase on one such sub-grid. Table~\ref{tab:constantinides-stages} lists the depth of each part.

\begin{algorithm}[H]
\caption{2D-NN compiled all-to-all staircase fermionic permutation}
\label{alg:constantinides-2dnn}
\begin{algorithmic}[1]
\Procedure{FermPerm}{$\mathrm{grid}: L_{\mathrm{eff}} \times L,\; \pi$}
\If{$L_{\mathrm{eff}} = 1$}
\State \textsc{FSWAP-RowSort}($\mathrm{grid}, \pi$) \Comment{depth $2L$}
\State \Return
\EndIf
\State $\sigma \gets$ staircase from Lemma~1 of~\cite{constantinides2025lowdepthfermionroutingancillas} applied to $\pi$
\State \textsc{ApplyStaircase}($\mathrm{grid}, \sigma$)
       \Comment{$\textsc{FinalSWAPs} \cdot G_3 \cdot G_2 \cdot G_1$}
\State $(\mathrm{top}, \mathrm{bot}) \gets$ horizontal split of grid at $L_{\mathrm{eff}}/2$
\State \Call{FermPerm}{$\mathrm{top}, \pi|_{\mathrm{top}}$}
       \Comment{this and the next call execute in parallel}
\State \Call{FermPerm}{$\mathrm{bot}, \pi|_{\mathrm{bot}}$}
\EndProcedure
\end{algorithmic}
\end{algorithm}
\begin{table}[t]
\centering
\renewcommand{\arraystretch}{1.15}
\begin{tabular}{@{}lc@{}}
\toprule
\textbf{Component} & \textbf{CNOT depth} \\
\midrule
Any of $P$, $L_1$, $L_2$, $P^\dagger$                                                                 & $\le 2\, (\log_2 w_d)(H_d + 1)$ \\
Gadget $G_i = P^\dagger\!\cdot\!L_2\!\cdot\!L_1\!\cdot\!P$                                            & $\le 8\, (\log_2 w_d)(H_d + 1)$ \\
\begin{tabular}[c]{@{}l@{}}Level-$d$ staircase\\ $(\textsc{FinalSWAPs} \cdot G_3 \cdot G_2 \cdot G_1)$\end{tabular} & $\le 24\, (\log_2 w_d)(H_d + 1) + H_d$ \\
\bottomrule
\end{tabular}
\caption{Stage-by-stage \CNOT-depth budget at level $d$, with $w_d \le X_d L = N/2^{d-1}$ and $H_d = 3(2X_d + L) \le 9L$. In the first row, each logical \CNOT layer costs $H_d+1$: one Hall RCR of depth $H_d$, followed by one \CNOT layer of depth one.}
\label{tab:constantinides-stages}
\end{table}

The bound $w_d \le X_dL$ gives $\log_2 w_d \le 2\log_2 L+1-d$, and $H_d=3L(1+2^{2-d})$. Substituting these expressions into Table~\ref{tab:constantinides-stages}, summing over $d=1,\ldots,\log_2 L$, and adding the base-case depth $2L$ gives
\begin{align}
D(L) &\le 2L + \sum_{d=1}^{\log_2 L} \!\bigl[24(2\log_2 L + 1 - d)(H_d + 1) + H_d\bigr] \notag \\
     &= 108\, L\, (\log_2 L)^2 + 615\, L \log_2 L - 274\, L \notag \\
     &\qquad + 36\, (\log_2 L)^2 - 276\, \log_2 L + 276 \notag \\
     &= 27\,\sqrt{N}\,(\log_2 N)^2 + \tfrac{615}{2}\,\sqrt{N}\log_2 N \notag \\
     &\qquad - 274\sqrt{N} + o(\sqrt{N}). \label{eq:constantinides-DL-bound}
\end{align}
The resulting 2D-NN compilation has depth $O(\sqrt{N}\log^2 N)$. It is a factor of $\log_2 N$ above the $O(\sqrt{N}\log N)$ merge-sort compilation of Appendix~\ref{app:reconf-2dnn}. The merge-sort compilation is itself a factor of $\log_2 N$ above our $10\sqrt{N}+O(1)$ result in Section~\ref{sec:full-alg}. The constant $27$ comes from the geometric decrease of the $X_d$-dependent part of $H_d$. It is smaller than a bound that charges every logical \CNOT layer the worst-case full-grid Hall RCR depth $9L$. The compiled circuit uses $L^2=N$ data qubits and no ancillas, mid-circuit measurements, or classical feedforward. It therefore matches the original in these resources. Its $O(N\sqrt{N}\log^2 N)$ entangling-gate count is dominated by the Hall RCRs at each level. Our construction in Section~\ref{sec:full-alg} uses $O(N\sqrt{N})$ entangling gates.

\section{Hilbert-Curve Disjointness for BK\texorpdfstring{$\to$}{ to }JW Rounds}
\label{app:hilbert-disjoint}

This appendix proves the claim used in Section~\ref{sec:ferm-perm-bk-parity}. In each round of the BK$\to$JW conversion, the \CNOT intervals on the inorder line have pairwise-disjoint dyadic containers of comparable length. Here, $J$ denotes one such access interval. Under the Hilbert layout, these containers become pairwise-disjoint axis-aligned rectangles of diameter $O(\sqrt{|J|})$. For the fully occupied $N=2^k$ layout, we implement a remote \CNOT without ancillas by swapping one endpoint along a shortest path, applying the \CNOT, and reversing those swaps. Because the rectangles are disjoint, these routes run in parallel. Summing their depths gives a total conversion depth of $O(\sqrt{N})$.

We consider the standard BK trees for $N=2^k$ and $N=2^k-1$, where $k\ge2$ is even. We index the cells of the order-$k/2$ square Hilbert traversal $H_k$ by $\{0,1,\dots,2^k-1\}$. For $N=2^k$, every cell is occupied. For the perfect-tree case $N=2^k-1$, qubit~$i$ occupies $H_k(i)$ and the final cell is vacant. The first case gives the $L\times L$ grid used in the main result, with $L=2^{k/2}$ and $N=L^2$. The second includes the $N=63$ example in Figure~\ref{fig:hilbert-bk-jw}.

\subsection{Dyadic intervals and the Hilbert layout}
\label{app:dyadic}

A \emph{dyadic interval of order $q$} is an interval
\[
[a2^q,\,(a+1)2^q-1]\subseteq [0,2^k-1].
\]
It has length $2^q$. We use two facts about these intervals: aligned translations preserve them, and the Hilbert layout maps them to rectangles.

\begin{lemma}[Translation invariance]
\label{lem:translate-dyadic}
Let $D=[a2^q,(a+1)2^q-1]$ be a dyadic interval of order $q$, and let $t$ be a nonnegative multiple of $2^q$. Then
\[
t+D\;=\;[t+a2^q,\,t+(a+1)2^q-1]
\]
is a dyadic interval of order $q$.
\end{lemma}

\begin{proof}
Writing $t=b\cdot 2^q$ gives $t+D=[(a+b)2^q,\,(a+b+1)2^q-1]$.
\end{proof}

\begin{proposition}[Hilbert dyadic-rectangle property]
\label{prop:hilbert-dyadic}
Under $H_k$, every dyadic interval of order $q$ maps to an axis-aligned rectangle of area $2^q$. Its side lengths are $2^{\lfloor q/2\rfloor}$ and $2^{\lceil q/2\rceil}$. Disjoint dyadic intervals map to non-overlapping rectangles.
\end{proposition}

\begin{proof}
The finite Hilbert layout $H_k$ fills a $2^{\lceil k/2\rceil}\times 2^{\lfloor k/2\rfloor}$ rectangle through four recursive $H_{k-2}$ calls on its four quarters~\cite{hilbert1935stetige, bader2012space}. Any two consecutive calls together cover an axis-aligned half of their parent square.

Let $J=[a2^q,(a+1)2^q-1]$. If $q\equiv k\pmod 2$, then $J$ is the domain of one recursive call. Thus, $H_k(J)$ is a subrectangle of area $2^q$. If $q\not\equiv k\pmod 2$, write $J=J_1\sqcup J_2$, where $J_1$ and $J_2$ are consecutive dyadic intervals of order $q-1$. Each $J_t$ is the domain of one call. These two calls are consecutive within the same parent call, whose area is $2^{q+1}$. Their images fill an axis-aligned half of that parent, so $H_k(J)$ is again an axis-aligned rectangle of area $2^q$. The stated side lengths follow in both cases. Because $H_k$ is a bijection on cells, it also preserves disjointness.
\end{proof}

\subsection{The local subproblem}
\label{app:local-pattern}

We first define the interval pattern for one round of \CNOTs.

\begin{definition}[Local-subproblem family $\mathcal{I}_{d,r}$]
\label{def:I-family}
For integers $d\ge 1$ and $1\le r\le d$, define the collection of intervals $\mathcal{I}_{d,r}$ in $[0,2^d-1]$ recursively by
\[
\mathcal{I}_{d,1}\;=\;\bigl\{[2^{d-1}-1,\,2^d-1]\bigr\},
\]
and, for $r\ge 2$,
\[
\mathcal{I}_{d,r}\;=\;\mathcal{I}_{d-1,r-1}\;\sqcup\;\bigl(2^{d-1}+\mathcal{I}_{d-1,r-1}\bigr),
\]
where $t+[a,b]:=[t+a,t+b]$ and $t+\mathcal{A}:=\{t+J:J\in\mathcal{A}\}$.
\end{definition}

The family $\mathcal{I}_{d,r}$ is the BK$\to$JW access pattern. To see this, consider a \emph{depth-$d$ local subproblem}. It consists of a balanced BK subtree on cells $0,\dots,2^d-2$ and an external parent at cell $2^d-1$. The subtree root is at the median cell $2^{d-1}-1$. Thus, sub-round $1$ applies one \CNOT to the endpoints $(2^{d-1}-1,\,2^d-1)$, as specified by $\mathcal{I}_{d,1}$.

After this \CNOT, two depth-$(d-1)$ local subproblems remain. The left child occupies cells $[0,2^{d-1}-2]$ and uses cell $2^{d-1}-1$ as its external parent. The right child occupies cells $[2^{d-1},2^d-2]$ and keeps the original external parent. Subtracting $2^{d-1}$ from the right child's cells makes the two child problems identical. Therefore, for $r\ge2$, sub-round $r$ consists of sub-round $r-1$ in the left child and its translation by $2^{d-1}$ in the right child. This is exactly the recursion that defines $\mathcal{I}_{d,r}$. Thus, $\mathcal{I}_{d,r}$ lists the inorder intervals of the \CNOTs in sub-round $r$.

For the standard $N=2^k$ BK tree, cell $2^k-1$ is the root above the perfect subtree on cells $0,\dots,2^k-2$. The full tree is therefore one depth-$k$ local subproblem. Its round-$r$ access pattern is $\mathcal{I}_{k,r}$ for $r=1,\dots,k$.

We next show that the intervals in $\mathcal{I}_{d,r}$ have pairwise-disjoint dyadic containers.

\begin{lemma}[Dyadic containers for $\mathcal{I}_{d,r}$]
\label{lem:local-dyadic}
For every $d\ge 1$ and $1\le r\le d$, there is an assignment $J\mapsto D(J)$. For each $J\in\mathcal{I}_{d,r}$, the interval $D(J)\subseteq[0,2^d-1]$ is dyadic and satisfies
\[
|J|\;\le\;|D(J)|\;\le\;2|J|.
\]
The family $\{D(J):J\in\mathcal{I}_{d,r}\}$ is pairwise disjoint.
\end{lemma}

\begin{proof}
We induct on $r$. For $r=1$, the only interval is $J=[2^{d-1}-1,\,2^d-1]$. Choose $D(J)=[0,2^d-1]$. Then $|J|=2^{d-1}+1$ and $|D(J)|=2^d$, so the bounds hold.

For $r\ge 2$, write
$\mathcal{I}_{d,r}=\mathcal{I}_{d-1,r-1}\sqcup(2^{d-1}+\mathcal{I}_{d-1,r-1})$.
By induction, each $J\in\mathcal{I}_{d-1,r-1}$ has a dyadic container $D(J)\subseteq[0,2^{d-1}-1]$ with $|J|\le|D(J)|\le 2|J|$. These containers are pairwise disjoint. Each $|D(J)|$ is a power of two no larger than $2^{d-1}$, so $2^{d-1}$ is a multiple of $|D(J)|$. Lemma~\ref{lem:translate-dyadic} shows that $2^{d-1}+D(J)$ is a dyadic interval of the same order. It is a container for $2^{d-1}+J$ with the same length bounds. The original containers lie in $[0,2^{d-1}-1]$, while the translated containers lie in $[2^{d-1},2^d-1]$. The full family is therefore pairwise disjoint.
\end{proof}

\subsection{Perfect-tree rounds and depth bound}
\label{app:global-rounds}

For the $N=2^k-1$ perfect-tree case, the full BK$\to$JW conversion processes the right-spine vertices one at a time. Let $v_0,v_1,\dots,v_{k-1}$ be these vertices from top to bottom. For each $i=0,1,\dots,k-2$, vertex $v_i$ and its balanced left subtree of depth $k-1-i$ occupy the contiguous block
\[
B_i\;=\;[\beta_i,\,\beta_i+L_i-1],\qquad L_i\;=\;2^{k-1-i}.
\]
The vertex $v_i$ is in the rightmost cell, $\beta_i+L_i-1$. This top-down order gives the offsets
\[
\beta_0\;=\;0,\qquad \beta_i\;=\;\sum_{j<i}L_j\;=\;2^k-2^{k-i}\quad(i\ge 1).
\]
Thus, the blocks $B_i$ are pairwise disjoint and $L_i\mid\beta_i$. In fact, $\beta_i/L_i=2^{i+1}-2$ for $i\ge1$, and $\beta_0=0$. Each block $B_i$ is a depth-$(k-1-i)$ local subproblem as defined in Section~\ref{app:local-pattern}. Here, $v_i$ is the external parent and its left subtree is the balanced subtree. The block is active in round $r$ exactly when $r\le k-1-i$, or equivalently, $i\le k-1-r$. Therefore, the round-$r$ intervals are
\[
\mathcal{G}_{k,r}\;=\;\bigsqcup_{i=0}^{k-1-r}\bigl(\beta_i+\mathcal{I}_{k-1-i,\,r}\bigr).
\]

\begin{theorem}[Per-round disjoint dyadic rectangles]
\label{thm:round-dyadic}
For every round $r\in\{1,\dots,k-1\}$, each $J\in\mathcal{G}_{k,r}$ has a dyadic container $D(J)\subseteq[0,2^k-1]$ with $|J|\le|D(J)|\le 2|J|$. The family $\{D(J):J\in\mathcal{G}_{k,r}\}$ is pairwise disjoint. Under the Hilbert layout, each round-$r$ \CNOT is therefore confined to an axis-aligned rectangle of diameter $O\!\left(\sqrt{|J|}\right)$. These rectangles are pairwise cell-disjoint.
\end{theorem}

\begin{proof}
Fix an active block $B_i$. By Lemma~\ref{lem:local-dyadic}, each $J'\in\mathcal{I}_{k-1-i,r}$ has a dyadic container inside $[0,L_i-1]$ of length at most $L_i$. The offset $\beta_i$ is a multiple of $L_i$ and hence of every power of two up to $L_i$. Lemma~\ref{lem:translate-dyadic} then shows that each translated container is also dyadic. This gives containers for $\beta_i+\mathcal{I}_{k-1-i,r}$ inside $B_i$. Lemma~\ref{lem:local-dyadic} gives pairwise disjointness within each $B_i$. Containers in different blocks are also disjoint because the blocks $B_i$ are disjoint. Finally, Proposition~\ref{prop:hilbert-dyadic} maps each container to a rectangle whose side lengths are $O(\sqrt{|D(J)|})$. Since $|D(J)|\le2|J|$, its diameter is $O\!\left(\sqrt{|J|}\right)$.
\end{proof}

For example, let $k=3$ and $N=7$. The top-level blocks are $B_0=[0,3]$ and $B_1=[4,5]$. Round $1$ applies the \CNOTs in $\mathcal{G}_{3,1}=\{[1,3],[4,5]\}$. Their dyadic containers are $[0,3]$ and $[4,5]$. Round $2$ applies the \CNOTs in $\mathcal{G}_{3,2}=\{[0,1],[2,3]\}$, whose intervals are already dyadic.

We obtain the depth bound by summing the rectangle diameters over all rounds. The recursion in Definition~\ref{def:I-family} halves the local problem size at each step, which gives $\max_{J\in\mathcal{I}_{d,r}}|J|=2^{d-r}+1$. The deepest active block, $B_0$, has the largest intervals. Therefore,
\[
\max_{J\in\mathcal{G}_{k,r}}|J|\;=\;2^{k-1-r}+1.
\]

\begin{corollary}[Hilbert-layout conversion depth]
\label{cor:bk-jw-depth}
For even $k$ and $N=2^k$, the total BK$\to$JW conversion depth under the square Hilbert layout is $O(\sqrt{N})$. The construction occupies an $L\times L$ grid with $L=2^{k/2}$ and uses no ancillas.
\end{corollary}

\begin{proof}
For $N=2^k$, round~$r$ is $\mathcal{I}_{k,r}$. Lemma~\ref{lem:local-dyadic} gives pairwise-disjoint dyadic containers, and the recursion gives a maximum interval length of $2^{k-r}+1$. The route--\CNOT--unroute construction therefore has depth $O(2^{(k-r)/2})$ in round $r$. Summing over the rounds gives
\[
\sum_{r=1}^{k}O\!\left(2^{(k-r)/2}\right)
=O\!\left(2^{k/2}\right)=O(\sqrt{N}).\qedhere
\]
\end{proof}

\begin{corollary}[Parity conversion depth]
\label{cor:parity-jw-depth}
For $N=2^k$ with even $k$, Parity$\leftrightarrow$JW conversion has depth $O(\sqrt{N})$ on the square Hilbert-ordered nearest-neighbor grid and uses no ancillas.
\end{corollary}

\begin{proof}
In the standard zero-based $N$-mode Parity and Fenwick-BK trees, qubit $N{-}1$ is the common root. The other qubits, $0,\ldots,N{-}2$, form a left spine in the Parity tree and a perfect balanced subtree in the BK tree. On this subtree, BK$\to$Parity is the left--right mirror of the perfect-tree BK$\to$JW schedule $\mathcal{G}_{k,r}$. Because \CNOT is self-inverse, Parity$\to$BK uses the same mirrored, labeled \CNOTs in reverse round order.

Define $\rho(i)=N{-}2{-}i$ for $i<N{-}1$ and $\rho(N{-}1)=N{-}1$. Initially place Parity qubit~$i$ at cell $H_k(\rho(i))$. Reflection maps each mirrored access interval to the corresponding interval in $\mathcal{G}_{k,r}$. Theorem~\ref{thm:round-dyadic} therefore gives the same disjoint Hilbert rectangles. Reversing the round order does not change their route--gate--unroute depth. The largest interval in perfect-tree round~$r$ has length $2^{k-1-r}+1$. Thus, the geometric sum of the round depths is $O(2^{k/2})=O(\sqrt{N})$ for Parity$\to$BK.

Next, an ordinary Hall-grid permutation sends qubit~$i$ from $H_k(\rho(i))$ to $H_k(i)$ in depth $O(\sqrt{N})$ without ancillas~\cite{hall-routing}. Corollary~\ref{cor:bk-jw-depth} then converts BK to JW in another $O(\sqrt{N})$ depth. Reversing these steps gives JW$\to$Parity with the same bound.
\end{proof}

Theorem~\ref{thm:round-dyadic} gives the same disjoint-container structure for the $N=2^k-1$ perfect tree shown in Figure~\ref{fig:hilbert-bk-jw}. That layout leaves one grid cell vacant, so we use it only to illustrate the structure. The zero-ancilla square-grid extension in the main text uses Corollaries~\ref{cor:bk-jw-depth} and~\ref{cor:parity-jw-depth} with $N=2^k$.

\bibliography{bib/references}

\begin{thebibliography}{60}%
\makeatletter
\providecommand \@ifxundefined [1]{%
 \@ifx{#1\undefined}
}%
\providecommand \@ifnum [1]{%
 \ifnum #1\expandafter \@firstoftwo
 \else \expandafter \@secondoftwo
 \fi
}%
\providecommand \@ifx [1]{%
 \ifx #1\expandafter \@firstoftwo
 \else \expandafter \@secondoftwo
 \fi
}%
\providecommand \natexlab [1]{#1}%
\providecommand \enquote  [1]{``#1''}%
\providecommand \bibnamefont  [1]{#1}%
\providecommand \bibfnamefont [1]{#1}%
\providecommand \citenamefont [1]{#1}%
\providecommand \href@noop [0]{\@secondoftwo}%
\providecommand \href [0]{\begingroup \@sanitize@url \@href}%
\providecommand \@href[1]{\@@startlink{#1}\@@href}%
\providecommand \@@href[1]{\endgroup#1\@@endlink}%
\providecommand \@sanitize@url [0]{\catcode `\\12\catcode `\$12\catcode `\&12\catcode `\#12\catcode `\^12\catcode `\_12\catcode `\%12\relax}%
\providecommand \@@startlink[1]{}%
\providecommand \@@endlink[0]{}%
\providecommand \url  [0]{\begingroup\@sanitize@url \@url }%
\providecommand \@url [1]{\endgroup\@href {#1}{\urlprefix }}%
\providecommand \urlprefix  [0]{URL }%
\providecommand \Eprint [0]{\href }%
\providecommand \doibase [0]{https://doi.org/}%
\providecommand \selectlanguage [0]{\@gobble}%
\providecommand \bibinfo  [0]{\@secondoftwo}%
\providecommand \bibfield  [0]{\@secondoftwo}%
\providecommand \translation [1]{[#1]}%
\providecommand \BibitemOpen [0]{}%
\providecommand \bibitemStop [0]{}%
\providecommand \bibitemNoStop [0]{.\EOS\space}%
\providecommand \EOS [0]{\spacefactor3000\relax}%
\providecommand \BibitemShut  [1]{\csname bibitem#1\endcsname}%
\let\auto@bib@innerbib\@empty
\bibitem [{\citenamefont {Abrams}\ and\ \citenamefont {Lloyd}(1997)}]{Abrams_1997}%
  \BibitemOpen
  \bibfield  {author} {\bibinfo {author} {\bibfnamefont {D.~S.}\ \bibnamefont {Abrams}}\ and\ \bibinfo {author} {\bibfnamefont {S.}~\bibnamefont {Lloyd}},\ }\bibfield  {title} {\bibinfo {title} {Simulation of many-body fermi systems on a universal quantum computer},\ }\href {https://doi.org/10.1103/physrevlett.79.2586} {\bibfield  {journal} {\bibinfo  {journal} {Physical Review Letters}\ }\textbf {\bibinfo {volume} {79}},\ \bibinfo {pages} {2586–2589} (\bibinfo {year} {1997})}\BibitemShut {NoStop}%
\bibitem [{\citenamefont {Ortiz}\ \emph {et~al.}(2001)\citenamefont {Ortiz}, \citenamefont {Gubernatis}, \citenamefont {Knill},\ and\ \citenamefont {Laflamme}}]{Ortiz_2001}%
  \BibitemOpen
  \bibfield  {author} {\bibinfo {author} {\bibfnamefont {G.}~\bibnamefont {Ortiz}}, \bibinfo {author} {\bibfnamefont {J.~E.}\ \bibnamefont {Gubernatis}}, \bibinfo {author} {\bibfnamefont {E.}~\bibnamefont {Knill}},\ and\ \bibinfo {author} {\bibfnamefont {R.}~\bibnamefont {Laflamme}},\ }\bibfield  {title} {\bibinfo {title} {Quantum algorithms for fermionic simulations},\ }\bibfield  {journal} {\bibinfo  {journal} {Physical Review A}\ }\textbf {\bibinfo {volume} {64}},\ \href {https://doi.org/10.1103/physreva.64.022319} {10.1103/physreva.64.022319} (\bibinfo {year} {2001})\BibitemShut {NoStop}%
\bibitem [{\citenamefont {Lanyon}\ \emph {et~al.}(2010)\citenamefont {Lanyon}, \citenamefont {Whitfield}, \citenamefont {Gillett}, \citenamefont {Goggin}, \citenamefont {Almeida}, \citenamefont {Kassal}, \citenamefont {Biamonte}, \citenamefont {Mohseni}, \citenamefont {Powell}, \citenamefont {Barbieri}, \citenamefont {Aspuru-Guzik},\ and\ \citenamefont {White}}]{Lanyon_2010}%
  \BibitemOpen
  \bibfield  {author} {\bibinfo {author} {\bibfnamefont {B.~P.}\ \bibnamefont {Lanyon}}, \bibinfo {author} {\bibfnamefont {J.~D.}\ \bibnamefont {Whitfield}}, \bibinfo {author} {\bibfnamefont {G.~G.}\ \bibnamefont {Gillett}}, \bibinfo {author} {\bibfnamefont {M.~E.}\ \bibnamefont {Goggin}}, \bibinfo {author} {\bibfnamefont {M.~P.}\ \bibnamefont {Almeida}}, \bibinfo {author} {\bibfnamefont {I.}~\bibnamefont {Kassal}}, \bibinfo {author} {\bibfnamefont {J.~D.}\ \bibnamefont {Biamonte}}, \bibinfo {author} {\bibfnamefont {M.}~\bibnamefont {Mohseni}}, \bibinfo {author} {\bibfnamefont {B.~J.}\ \bibnamefont {Powell}}, \bibinfo {author} {\bibfnamefont {M.}~\bibnamefont {Barbieri}}, \bibinfo {author} {\bibfnamefont {A.}~\bibnamefont {Aspuru-Guzik}},\ and\ \bibinfo {author} {\bibfnamefont {A.~G.}\ \bibnamefont {White}},\ }\bibfield  {title} {\bibinfo {title} {Towards quantum chemistry on a quantum computer},\ }\href {https://doi.org/10.1038/nchem.483} {\bibfield  {journal} {\bibinfo  {journal} {Nature Chemistry}\ }\textbf {\bibinfo {volume} {2}},\ \bibinfo {pages} {106–111} (\bibinfo {year} {2010})}\BibitemShut {NoStop}%
\bibitem [{\citenamefont {Wecker}\ \emph {et~al.}(2014)\citenamefont {Wecker}, \citenamefont {Bauer}, \citenamefont {Clark}, \citenamefont {Hastings},\ and\ \citenamefont {Troyer}}]{Wecker_2014}%
  \BibitemOpen
  \bibfield  {author} {\bibinfo {author} {\bibfnamefont {D.}~\bibnamefont {Wecker}}, \bibinfo {author} {\bibfnamefont {B.}~\bibnamefont {Bauer}}, \bibinfo {author} {\bibfnamefont {B.~K.}\ \bibnamefont {Clark}}, \bibinfo {author} {\bibfnamefont {M.~B.}\ \bibnamefont {Hastings}},\ and\ \bibinfo {author} {\bibfnamefont {M.}~\bibnamefont {Troyer}},\ }\bibfield  {title} {\bibinfo {title} {Gate-count estimates for performing quantum chemistry on small quantum computers},\ }\bibfield  {journal} {\bibinfo  {journal} {Physical Review A}\ }\textbf {\bibinfo {volume} {90}},\ \href {https://doi.org/10.1103/physreva.90.022305} {10.1103/physreva.90.022305} (\bibinfo {year} {2014})\BibitemShut {NoStop}%
\bibitem [{\citenamefont {Reiher}\ \emph {et~al.}(2017)\citenamefont {Reiher}, \citenamefont {Wiebe}, \citenamefont {Svore}, \citenamefont {Wecker},\ and\ \citenamefont {Troyer}}]{Reiher_2017}%
  \BibitemOpen
  \bibfield  {author} {\bibinfo {author} {\bibfnamefont {M.}~\bibnamefont {Reiher}}, \bibinfo {author} {\bibfnamefont {N.}~\bibnamefont {Wiebe}}, \bibinfo {author} {\bibfnamefont {K.~M.}\ \bibnamefont {Svore}}, \bibinfo {author} {\bibfnamefont {D.}~\bibnamefont {Wecker}},\ and\ \bibinfo {author} {\bibfnamefont {M.}~\bibnamefont {Troyer}},\ }\bibfield  {title} {\bibinfo {title} {Elucidating reaction mechanisms on quantum computers},\ }\href {https://doi.org/10.1073/pnas.1619152114} {\bibfield  {journal} {\bibinfo  {journal} {Proceedings of the National Academy of Sciences}\ }\textbf {\bibinfo {volume} {114}},\ \bibinfo {pages} {7555–7560} (\bibinfo {year} {2017})}\BibitemShut {NoStop}%
\bibitem [{\citenamefont {Bauer}\ \emph {et~al.}(2020)\citenamefont {Bauer}, \citenamefont {Bravyi}, \citenamefont {Motta},\ and\ \citenamefont {Chan}}]{Bauer_2020}%
  \BibitemOpen
  \bibfield  {author} {\bibinfo {author} {\bibfnamefont {B.}~\bibnamefont {Bauer}}, \bibinfo {author} {\bibfnamefont {S.}~\bibnamefont {Bravyi}}, \bibinfo {author} {\bibfnamefont {M.}~\bibnamefont {Motta}},\ and\ \bibinfo {author} {\bibfnamefont {G.~K.-L.}\ \bibnamefont {Chan}},\ }\bibfield  {title} {\bibinfo {title} {Quantum algorithms for quantum chemistry and quantum materials science},\ }\href {https://doi.org/10.1021/acs.chemrev.9b00829} {\bibfield  {journal} {\bibinfo  {journal} {Chemical Reviews}\ }\textbf {\bibinfo {volume} {120}},\ \bibinfo {pages} {12685–12717} (\bibinfo {year} {2020})}\BibitemShut {NoStop}%
\bibitem [{\citenamefont {McArdle}\ \emph {et~al.}(2020)\citenamefont {McArdle}, \citenamefont {Endo}, \citenamefont {Aspuru-Guzik}, \citenamefont {Benjamin},\ and\ \citenamefont {Yuan}}]{McArdle_2020}%
  \BibitemOpen
  \bibfield  {author} {\bibinfo {author} {\bibfnamefont {S.}~\bibnamefont {McArdle}}, \bibinfo {author} {\bibfnamefont {S.}~\bibnamefont {Endo}}, \bibinfo {author} {\bibfnamefont {A.}~\bibnamefont {Aspuru-Guzik}}, \bibinfo {author} {\bibfnamefont {S.~C.}\ \bibnamefont {Benjamin}},\ and\ \bibinfo {author} {\bibfnamefont {X.}~\bibnamefont {Yuan}},\ }\bibfield  {title} {\bibinfo {title} {Quantum computational chemistry},\ }\bibfield  {journal} {\bibinfo  {journal} {Reviews of Modern Physics}\ }\textbf {\bibinfo {volume} {92}},\ \href {https://doi.org/10.1103/revmodphys.92.015003} {10.1103/revmodphys.92.015003} (\bibinfo {year} {2020})\BibitemShut {NoStop}%
\bibitem [{\citenamefont {Di~Meglio}\ \emph {et~al.}(2024)\citenamefont {Di~Meglio}, \citenamefont {Jansen}, \citenamefont {Tavernelli}, \citenamefont {Alexandrou}, \citenamefont {Arunachalam}, \citenamefont {Bauer}, \citenamefont {Borras}, \citenamefont {Carrazza}, \citenamefont {Crippa}, \citenamefont {Croft} \emph {et~al.}}]{di2024quantum}%
  \BibitemOpen
  \bibfield  {author} {\bibinfo {author} {\bibfnamefont {A.}~\bibnamefont {Di~Meglio}}, \bibinfo {author} {\bibfnamefont {K.}~\bibnamefont {Jansen}}, \bibinfo {author} {\bibfnamefont {I.}~\bibnamefont {Tavernelli}}, \bibinfo {author} {\bibfnamefont {C.}~\bibnamefont {Alexandrou}}, \bibinfo {author} {\bibfnamefont {S.}~\bibnamefont {Arunachalam}}, \bibinfo {author} {\bibfnamefont {C.~W.}\ \bibnamefont {Bauer}}, \bibinfo {author} {\bibfnamefont {K.}~\bibnamefont {Borras}}, \bibinfo {author} {\bibfnamefont {S.}~\bibnamefont {Carrazza}}, \bibinfo {author} {\bibfnamefont {A.}~\bibnamefont {Crippa}}, \bibinfo {author} {\bibfnamefont {V.}~\bibnamefont {Croft}}, \emph {et~al.},\ }\bibfield  {title} {\bibinfo {title} {Quantum computing for high-energy physics: State of the art and challenges},\ }\href@noop {} {\bibfield  {journal} {\bibinfo  {journal} {PRX Quantum}\ }\textbf {\bibinfo {volume} {5}},\ \bibinfo {pages} {037001} (\bibinfo {year} {2024})}\BibitemShut {NoStop}%
\bibitem [{\citenamefont {Manin}(1980)}]{manin1980computable}%
  \BibitemOpen
  \bibfield  {author} {\bibinfo {author} {\bibfnamefont {Y.}~\bibnamefont {Manin}},\ }\bibfield  {title} {\bibinfo {title} {Computable and uncomputable},\ }\href@noop {} {\bibfield  {journal} {\bibinfo  {journal} {Sovetskoye Radio, Moscow}\ }\textbf {\bibinfo {volume} {128}},\ \bibinfo {pages} {28} (\bibinfo {year} {1980})}\BibitemShut {NoStop}%
\bibitem [{\citenamefont {Feynman}(2018)}]{feynman2018simulating}%
  \BibitemOpen
  \bibfield  {author} {\bibinfo {author} {\bibfnamefont {R.~P.}\ \bibnamefont {Feynman}},\ }\bibfield  {title} {\bibinfo {title} {Simulating physics with computers},\ }in\ \href@noop {} {\emph {\bibinfo {booktitle} {Feynman and computation}}}\ (\bibinfo  {publisher} {CRC Press},\ \bibinfo {year} {2018})\ pp.\ \bibinfo {pages} {133--153}\BibitemShut {NoStop}%
\bibitem [{\citenamefont {Jordan}\ and\ \citenamefont {Wigner}(1928)}]{jordan1928paulische}%
  \BibitemOpen
  \bibfield  {author} {\bibinfo {author} {\bibfnamefont {P.}~\bibnamefont {Jordan}}\ and\ \bibinfo {author} {\bibfnamefont {E.}~\bibnamefont {Wigner}},\ }\bibfield  {title} {\bibinfo {title} {{\"U}ber das paulische {\"a}quivalenzverbot},\ }\href@noop {} {\bibfield  {journal} {\bibinfo  {journal} {Zeitschrift f{\"u}r Physik}\ }\textbf {\bibinfo {volume} {47}},\ \bibinfo {pages} {631} (\bibinfo {year} {1928})}\BibitemShut {NoStop}%
\bibitem [{\citenamefont {Bravyi}\ and\ \citenamefont {Kitaev}(2002)}]{bk}%
  \BibitemOpen
  \bibfield  {author} {\bibinfo {author} {\bibfnamefont {S.~B.}\ \bibnamefont {Bravyi}}\ and\ \bibinfo {author} {\bibfnamefont {A.~Y.}\ \bibnamefont {Kitaev}},\ }\bibfield  {title} {\bibinfo {title} {Fermionic quantum computation},\ }\href@noop {} {\bibfield  {journal} {\bibinfo  {journal} {Annals of Physics}\ }\textbf {\bibinfo {volume} {298}},\ \bibinfo {pages} {210} (\bibinfo {year} {2002})}\BibitemShut {NoStop}%
\bibitem [{\citenamefont {Yordanov}\ \emph {et~al.}(2020)\citenamefont {Yordanov}, \citenamefont {Arvidsson-Shukur},\ and\ \citenamefont {Barnes}}]{parity}%
  \BibitemOpen
  \bibfield  {author} {\bibinfo {author} {\bibfnamefont {Y.~S.}\ \bibnamefont {Yordanov}}, \bibinfo {author} {\bibfnamefont {D.~R.}\ \bibnamefont {Arvidsson-Shukur}},\ and\ \bibinfo {author} {\bibfnamefont {C.~H.}\ \bibnamefont {Barnes}},\ }\bibfield  {title} {\bibinfo {title} {Efficient quantum circuits for quantum computational chemistry},\ }\href@noop {} {\bibfield  {journal} {\bibinfo  {journal} {Physical Review A}\ }\textbf {\bibinfo {volume} {102}},\ \bibinfo {pages} {062612} (\bibinfo {year} {2020})}\BibitemShut {NoStop}%
\bibitem [{\citenamefont {Liu}\ \emph {et~al.}(2024)\citenamefont {Liu}, \citenamefont {Che}, \citenamefont {Zhou}, \citenamefont {Shi},\ and\ \citenamefont {Li}}]{fermihedral}%
  \BibitemOpen
  \bibfield  {author} {\bibinfo {author} {\bibfnamefont {Y.}~\bibnamefont {Liu}}, \bibinfo {author} {\bibfnamefont {S.}~\bibnamefont {Che}}, \bibinfo {author} {\bibfnamefont {J.}~\bibnamefont {Zhou}}, \bibinfo {author} {\bibfnamefont {Y.}~\bibnamefont {Shi}},\ and\ \bibinfo {author} {\bibfnamefont {G.}~\bibnamefont {Li}},\ }\bibfield  {title} {\bibinfo {title} {{Fermihedral}: On the optimal compilation for fermion-to-qubit encoding},\ }in\ \href {https://doi.org/10.1145/3620666.3651371} {\emph {\bibinfo {booktitle} {Proceedings of the 29th ACM International Conference on Architectural Support for Programming Languages and Operating Systems, Volume 3 ({ASPLOS})}}}\ (\bibinfo  {publisher} {ACM},\ \bibinfo {year} {2024})\ pp.\ \bibinfo {pages} {382--397}\BibitemShut {NoStop}%
\bibitem [{\citenamefont {Liu}\ \emph {et~al.}(2025)\citenamefont {Liu}, \citenamefont {Yao}, \citenamefont {Hong}, \citenamefont {Froustey}, \citenamefont {Rrapaj}, \citenamefont {Iancu}, \citenamefont {Li},\ and\ \citenamefont {Shi}}]{hatt}%
  \BibitemOpen
  \bibfield  {author} {\bibinfo {author} {\bibfnamefont {Y.}~\bibnamefont {Liu}}, \bibinfo {author} {\bibfnamefont {K.}~\bibnamefont {Yao}}, \bibinfo {author} {\bibfnamefont {J.}~\bibnamefont {Hong}}, \bibinfo {author} {\bibfnamefont {J.}~\bibnamefont {Froustey}}, \bibinfo {author} {\bibfnamefont {E.}~\bibnamefont {Rrapaj}}, \bibinfo {author} {\bibfnamefont {C.}~\bibnamefont {Iancu}}, \bibinfo {author} {\bibfnamefont {G.}~\bibnamefont {Li}},\ and\ \bibinfo {author} {\bibfnamefont {Y.}~\bibnamefont {Shi}},\ }\bibfield  {title} {\bibinfo {title} {{HATT}: {Hamiltonian} {Adaptive} {Ternary} {Tree} for optimizing fermion-to-qubit mapping},\ }in\ \href {https://doi.org/10.1109/HPCA61900.2025.00022} {\emph {\bibinfo {booktitle} {2025 IEEE International Symposium on High Performance Computer Architecture ({HPCA})}}}\ (\bibinfo {organization} {IEEE},\ \bibinfo {year} {2025})\ pp.\ \bibinfo {pages} {143--157}\BibitemShut {NoStop}%
\bibitem [{\citenamefont {Maskara}\ \emph {et~al.}(2025)\citenamefont {Maskara}, \citenamefont {Kalinowski}, \citenamefont {Gonzalez-Cuadra},\ and\ \citenamefont {Lukin}}]{maskara2025fastsimulationfermionsreconfigurable}%
  \BibitemOpen
  \bibfield  {author} {\bibinfo {author} {\bibfnamefont {N.}~\bibnamefont {Maskara}}, \bibinfo {author} {\bibfnamefont {M.}~\bibnamefont {Kalinowski}}, \bibinfo {author} {\bibfnamefont {D.}~\bibnamefont {Gonzalez-Cuadra}},\ and\ \bibinfo {author} {\bibfnamefont {M.~D.}\ \bibnamefont {Lukin}},\ }\href {https://arxiv.org/abs/2509.08898} {\bibinfo {title} {Fast simulation of fermions with reconfigurable qubits}} (\bibinfo {year} {2025}),\ \Eprint {https://arxiv.org/abs/2509.08898} {arXiv:2509.08898 [quant-ph]} \BibitemShut {NoStop}%
\bibitem [{\citenamefont {Constantinides}\ \emph {et~al.}(2025)\citenamefont {Constantinides}, \citenamefont {Yu}, \citenamefont {Devulapalli}, \citenamefont {Fahimniya}, \citenamefont {Schaeffer}, \citenamefont {Childs}, \citenamefont {Gullans}, \citenamefont {Schuckert},\ and\ \citenamefont {Gorshkov}}]{constantinides2025lowdepthfermionroutingancillas}%
  \BibitemOpen
  \bibfield  {author} {\bibinfo {author} {\bibfnamefont {N.}~\bibnamefont {Constantinides}}, \bibinfo {author} {\bibfnamefont {J.}~\bibnamefont {Yu}}, \bibinfo {author} {\bibfnamefont {D.}~\bibnamefont {Devulapalli}}, \bibinfo {author} {\bibfnamefont {A.}~\bibnamefont {Fahimniya}}, \bibinfo {author} {\bibfnamefont {L.}~\bibnamefont {Schaeffer}}, \bibinfo {author} {\bibfnamefont {A.~M.}\ \bibnamefont {Childs}}, \bibinfo {author} {\bibfnamefont {M.~J.}\ \bibnamefont {Gullans}}, \bibinfo {author} {\bibfnamefont {A.}~\bibnamefont {Schuckert}},\ and\ \bibinfo {author} {\bibfnamefont {A.~V.}\ \bibnamefont {Gorshkov}},\ }\href {https://arxiv.org/abs/2510.05099} {\bibinfo {title} {Low-depth fermion routing without ancillas}} (\bibinfo {year} {2025}),\ \Eprint {https://arxiv.org/abs/2510.05099} {arXiv:2510.05099 [quant-ph]} \BibitemShut {NoStop}%
\bibitem [{\citenamefont {Ferris}(2014)}]{ffft}%
  \BibitemOpen
  \bibfield  {author} {\bibinfo {author} {\bibfnamefont {A.~J.}\ \bibnamefont {Ferris}},\ }\bibfield  {title} {\bibinfo {title} {Fourier transform for fermionic systems and the spectral tensor network},\ }\bibfield  {journal} {\bibinfo  {journal} {Physical Review Letters}\ }\textbf {\bibinfo {volume} {113}},\ \href {https://doi.org/10.1103/physrevlett.113.010401} {10.1103/physrevlett.113.010401} (\bibinfo {year} {2014})\BibitemShut {NoStop}%
\bibitem [{\citenamefont {Babbush}\ \emph {et~al.}(2018)\citenamefont {Babbush}, \citenamefont {Wiebe}, \citenamefont {McClean}, \citenamefont {McClain}, \citenamefont {Neven},\ and\ \citenamefont {Chan}}]{ffft-ct}%
  \BibitemOpen
  \bibfield  {author} {\bibinfo {author} {\bibfnamefont {R.}~\bibnamefont {Babbush}}, \bibinfo {author} {\bibfnamefont {N.}~\bibnamefont {Wiebe}}, \bibinfo {author} {\bibfnamefont {J.}~\bibnamefont {McClean}}, \bibinfo {author} {\bibfnamefont {J.}~\bibnamefont {McClain}}, \bibinfo {author} {\bibfnamefont {H.}~\bibnamefont {Neven}},\ and\ \bibinfo {author} {\bibfnamefont {G.~K.-L.}\ \bibnamefont {Chan}},\ }\bibfield  {title} {\bibinfo {title} {Low-depth quantum simulation of materials},\ }\href@noop {} {\bibfield  {journal} {\bibinfo  {journal} {Physical Review X}\ }\textbf {\bibinfo {volume} {8}},\ \bibinfo {pages} {011044} (\bibinfo {year} {2018})}\BibitemShut {NoStop}%
\bibitem [{\citenamefont {Kivlichan}\ \emph {et~al.}(2018)\citenamefont {Kivlichan}, \citenamefont {McClean}, \citenamefont {Wiebe}, \citenamefont {Gidney}, \citenamefont {Aspuru-Guzik}, \citenamefont {Chan},\ and\ \citenamefont {Babbush}}]{fswap1Dnetwork}%
  \BibitemOpen
  \bibfield  {author} {\bibinfo {author} {\bibfnamefont {I.~D.}\ \bibnamefont {Kivlichan}}, \bibinfo {author} {\bibfnamefont {J.}~\bibnamefont {McClean}}, \bibinfo {author} {\bibfnamefont {N.}~\bibnamefont {Wiebe}}, \bibinfo {author} {\bibfnamefont {C.}~\bibnamefont {Gidney}}, \bibinfo {author} {\bibfnamefont {A.}~\bibnamefont {Aspuru-Guzik}}, \bibinfo {author} {\bibfnamefont {G.~K.-L.}\ \bibnamefont {Chan}},\ and\ \bibinfo {author} {\bibfnamefont {R.}~\bibnamefont {Babbush}},\ }\bibfield  {title} {\bibinfo {title} {Quantum simulation of electronic structure with linear depth and connectivity},\ }\bibfield  {journal} {\bibinfo  {journal} {Physical Review Letters}\ }\textbf {\bibinfo {volume} {120}},\ \href {https://doi.org/10.1103/physrevlett.120.110501} {10.1103/physrevlett.120.110501} (\bibinfo {year} {2018})\BibitemShut {NoStop}%
\bibitem [{\citenamefont {Devulapalli}\ \emph {et~al.}(2024)\citenamefont {Devulapalli}, \citenamefont {Schoute}, \citenamefont {Bapat}, \citenamefont {Childs},\ and\ \citenamefont {Gorshkov}}]{devulapalli2022routing}%
  \BibitemOpen
  \bibfield  {author} {\bibinfo {author} {\bibfnamefont {D.}~\bibnamefont {Devulapalli}}, \bibinfo {author} {\bibfnamefont {E.}~\bibnamefont {Schoute}}, \bibinfo {author} {\bibfnamefont {A.}~\bibnamefont {Bapat}}, \bibinfo {author} {\bibfnamefont {A.~M.}\ \bibnamefont {Childs}},\ and\ \bibinfo {author} {\bibfnamefont {A.~V.}\ \bibnamefont {Gorshkov}},\ }\bibfield  {title} {\bibinfo {title} {Quantum routing with teleportation},\ }\bibfield  {journal} {\bibinfo  {journal} {Physical Review Research}\ }\textbf {\bibinfo {volume} {6}},\ \href {https://doi.org/10.1103/physrevresearch.6.033313} {10.1103/physrevresearch.6.033313} (\bibinfo {year} {2024})\BibitemShut {NoStop}%
\bibitem [{\citenamefont {{Google Quantum AI and Collaborators}}(2025)}]{google2025quantum}%
  \BibitemOpen
  \bibfield  {author} {\bibinfo {author} {\bibnamefont {{Google Quantum AI and Collaborators}}},\ }\bibfield  {title} {\bibinfo {title} {Quantum error correction below the surface code threshold},\ }\href {https://doi.org/10.1038/s41586-024-08449-y} {\bibfield  {journal} {\bibinfo  {journal} {Nature}\ }\textbf {\bibinfo {volume} {638}},\ \bibinfo {pages} {920} (\bibinfo {year} {2025})}\BibitemShut {NoStop}%
\bibitem [{\citenamefont {Kim}\ \emph {et~al.}(2023)\citenamefont {Kim}, \citenamefont {Eddins}, \citenamefont {Anand}, \citenamefont {Wei}, \citenamefont {van~den Berg}, \citenamefont {Rosenblatt}, \citenamefont {Nayfeh}, \citenamefont {Wu}, \citenamefont {Zaletel}, \citenamefont {Temme} \emph {et~al.}}]{kim2023evidence}%
  \BibitemOpen
  \bibfield  {author} {\bibinfo {author} {\bibfnamefont {Y.}~\bibnamefont {Kim}}, \bibinfo {author} {\bibfnamefont {A.}~\bibnamefont {Eddins}}, \bibinfo {author} {\bibfnamefont {S.}~\bibnamefont {Anand}}, \bibinfo {author} {\bibfnamefont {K.~X.}\ \bibnamefont {Wei}}, \bibinfo {author} {\bibfnamefont {E.}~\bibnamefont {van~den Berg}}, \bibinfo {author} {\bibfnamefont {S.}~\bibnamefont {Rosenblatt}}, \bibinfo {author} {\bibfnamefont {H.}~\bibnamefont {Nayfeh}}, \bibinfo {author} {\bibfnamefont {Y.}~\bibnamefont {Wu}}, \bibinfo {author} {\bibfnamefont {M.}~\bibnamefont {Zaletel}}, \bibinfo {author} {\bibfnamefont {K.}~\bibnamefont {Temme}}, \emph {et~al.},\ }\bibfield  {title} {\bibinfo {title} {Evidence for the utility of quantum computing before fault tolerance},\ }\href@noop {} {\bibfield  {journal} {\bibinfo  {journal} {Nature}\ }\textbf {\bibinfo {volume} {618}},\ \bibinfo {pages} {500} (\bibinfo {year} {2023})}\BibitemShut {NoStop}%
\bibitem [{\citenamefont {Jiang}\ \emph {et~al.}(2018)\citenamefont {Jiang}, \citenamefont {Sung}, \citenamefont {Kechedzhi}, \citenamefont {Smelyanskiy},\ and\ \citenamefont {Boixo}}]{gamma}%
  \BibitemOpen
  \bibfield  {author} {\bibinfo {author} {\bibfnamefont {Z.}~\bibnamefont {Jiang}}, \bibinfo {author} {\bibfnamefont {K.~J.}\ \bibnamefont {Sung}}, \bibinfo {author} {\bibfnamefont {K.}~\bibnamefont {Kechedzhi}}, \bibinfo {author} {\bibfnamefont {V.~N.}\ \bibnamefont {Smelyanskiy}},\ and\ \bibinfo {author} {\bibfnamefont {S.}~\bibnamefont {Boixo}},\ }\bibfield  {title} {\bibinfo {title} {Quantum algorithms to simulate many-body physics of correlated fermions},\ }\bibfield  {journal} {\bibinfo  {journal} {Physical Review Applied}\ }\textbf {\bibinfo {volume} {9}},\ \href {https://doi.org/10.1103/physrevapplied.9.044036} {10.1103/physrevapplied.9.044036} (\bibinfo {year} {2018})\BibitemShut {NoStop}%
\bibitem [{\citenamefont {Hilbert}(1935)}]{hilbert1935stetige}%
  \BibitemOpen
  \bibfield  {author} {\bibinfo {author} {\bibfnamefont {D.}~\bibnamefont {Hilbert}},\ }\bibfield  {title} {\bibinfo {title} {{\"U}ber die stetige abbildung einer linie auf ein fl{\"a}chenst{\"u}ck},\ }in\ \href@noop {} {\emph {\bibinfo {booktitle} {Dritter Band: Analysis{\textperiodcentered} Grundlagen der Mathematik{\textperiodcentered} Physik Verschiedenes: Nebst Einer Lebensgeschichte}}}\ (\bibinfo  {publisher} {Springer},\ \bibinfo {year} {1935})\ pp.\ \bibinfo {pages} {1--2}\BibitemShut {NoStop}%
\bibitem [{\citenamefont {Aigner}\ \emph {et~al.}(2026)\citenamefont {Aigner}, \citenamefont {Klaver}, \citenamefont {Lanthaler},\ and\ \citenamefont {Lechner}}]{aigner2026fermionlattices}%
  \BibitemOpen
  \bibfield  {author} {\bibinfo {author} {\bibfnamefont {G.}~\bibnamefont {Aigner}}, \bibinfo {author} {\bibfnamefont {B.}~\bibnamefont {Klaver}}, \bibinfo {author} {\bibfnamefont {M.}~\bibnamefont {Lanthaler}},\ and\ \bibinfo {author} {\bibfnamefont {W.}~\bibnamefont {Lechner}},\ }\bibfield  {title} {\bibinfo {title} {Fermion lattices can be simulated by same-size qubit lattices with {$\mathcal{O}(1)$} interaction overhead},\ }\bibfield  {journal} {\bibinfo  {journal} {arXiv preprint arXiv:2605.12600}\ }\href {https://doi.org/10.48550/arXiv.2605.12600} {10.48550/arXiv.2605.12600} (\bibinfo {year} {2026})\BibitemShut {NoStop}%
\bibitem [{\citenamefont {Miller}\ \emph {et~al.}(2025)\citenamefont {Miller}, \citenamefont {Favre}, \citenamefont {Holmes}, \citenamefont {Salehi}, \citenamefont {Chakraborty}, \citenamefont {Nyk{\"a}nen}, \citenamefont {Zimbor{\'a}s}, \citenamefont {Glos},\ and\ \citenamefont {Garc{\'i}a-P{\'e}rez}}]{majoranaprop}%
  \BibitemOpen
  \bibfield  {author} {\bibinfo {author} {\bibfnamefont {A.}~\bibnamefont {Miller}}, \bibinfo {author} {\bibfnamefont {J.}~\bibnamefont {Favre}}, \bibinfo {author} {\bibfnamefont {Z.}~\bibnamefont {Holmes}}, \bibinfo {author} {\bibfnamefont {{\"O}.}~\bibnamefont {Salehi}}, \bibinfo {author} {\bibfnamefont {R.}~\bibnamefont {Chakraborty}}, \bibinfo {author} {\bibfnamefont {A.}~\bibnamefont {Nyk{\"a}nen}}, \bibinfo {author} {\bibfnamefont {Z.}~\bibnamefont {Zimbor{\'a}s}}, \bibinfo {author} {\bibfnamefont {A.}~\bibnamefont {Glos}},\ and\ \bibinfo {author} {\bibfnamefont {G.}~\bibnamefont {Garc{\'i}a-P{\'e}rez}},\ }\bibfield  {title} {\bibinfo {title} {Simulation of {Fermionic} circuits using {Majorana} {Propagation}},\ }\href@noop {} {\bibfield  {journal} {\bibinfo  {journal} {arXiv preprint arXiv:2503.18939}\ } (\bibinfo {year} {2025})}\BibitemShut {NoStop}%
\bibitem [{\citenamefont {Vatan}\ and\ \citenamefont {Williams}(2004)}]{fswap}%
  \BibitemOpen
  \bibfield  {author} {\bibinfo {author} {\bibfnamefont {F.}~\bibnamefont {Vatan}}\ and\ \bibinfo {author} {\bibfnamefont {C.}~\bibnamefont {Williams}},\ }\bibfield  {title} {\bibinfo {title} {Optimal quantum circuits for general two-qubit gates},\ }\bibfield  {journal} {\bibinfo  {journal} {Physical Review A}\ }\textbf {\bibinfo {volume} {69}},\ \href {https://doi.org/10.1103/physreva.69.032315} {10.1103/physreva.69.032315} (\bibinfo {year} {2004})\BibitemShut {NoStop}%
\bibitem [{\citenamefont {Litinski}\ and\ \citenamefont {von Oppen}(2018)}]{litinski2018lattice}%
  \BibitemOpen
  \bibfield  {author} {\bibinfo {author} {\bibfnamefont {D.}~\bibnamefont {Litinski}}\ and\ \bibinfo {author} {\bibfnamefont {F.}~\bibnamefont {von Oppen}},\ }\bibfield  {title} {\bibinfo {title} {Lattice surgery with a twist: Simplifying {Clifford} gates of surface codes},\ }\href {https://doi.org/10.22331/q-2018-05-04-62} {\bibfield  {journal} {\bibinfo  {journal} {Quantum}\ }\textbf {\bibinfo {volume} {2}},\ \bibinfo {pages} {62} (\bibinfo {year} {2018})}\BibitemShut {NoStop}%
\bibitem [{\citenamefont {Knuth}(1998)}]{oet}%
  \BibitemOpen
  \bibfield  {author} {\bibinfo {author} {\bibfnamefont {D.~E.}\ \bibnamefont {Knuth}},\ }\href@noop {} {\emph {\bibinfo {title} {The art of computer programming, volume 3: (2nd ed.) sorting and searching}}}\ (\bibinfo  {publisher} {Addison Wesley Longman Publishing Co., Inc.},\ \bibinfo {address} {USA},\ \bibinfo {year} {1998})\BibitemShut {NoStop}%
\bibitem [{\citenamefont {Alon}\ \emph {et~al.}(1994)\citenamefont {Alon}, \citenamefont {Chung},\ and\ \citenamefont {Graham}}]{hall-routing}%
  \BibitemOpen
  \bibfield  {author} {\bibinfo {author} {\bibfnamefont {N.}~\bibnamefont {Alon}}, \bibinfo {author} {\bibfnamefont {F.~R.~K.}\ \bibnamefont {Chung}},\ and\ \bibinfo {author} {\bibfnamefont {R.~L.}\ \bibnamefont {Graham}},\ }\bibfield  {title} {\bibinfo {title} {Routing permutations on graphs via matchings},\ }\href {https://doi.org/10.1137/S0895480192236628} {\bibfield  {journal} {\bibinfo  {journal} {SIAM Journal on Discrete Mathematics}\ }\textbf {\bibinfo {volume} {7}},\ \bibinfo {pages} {513} (\bibinfo {year} {1994})}\BibitemShut {NoStop}%
\bibitem [{\citenamefont {König}(1916)}]{Lregularmatching}%
  \BibitemOpen
  \bibfield  {author} {\bibinfo {author} {\bibfnamefont {D.}~\bibnamefont {König}},\ }\bibfield  {title} {\bibinfo {title} {Über graphen und ihre anwendung auf determinantentheorie und mengenlehre},\ }\href {http://eudml.org/doc/158740} {\bibfield  {journal} {\bibinfo  {journal} {Mathematische Annalen}\ }\textbf {\bibinfo {volume} {77}},\ \bibinfo {pages} {453} (\bibinfo {year} {1916})}\BibitemShut {NoStop}%
\bibitem [{\citenamefont {Gabow}\ and\ \citenamefont {Kariv}(1982)}]{edgecoloring}%
  \BibitemOpen
  \bibfield  {author} {\bibinfo {author} {\bibfnamefont {H.~N.}\ \bibnamefont {Gabow}}\ and\ \bibinfo {author} {\bibfnamefont {O.}~\bibnamefont {Kariv}},\ }\bibfield  {title} {\bibinfo {title} {Algorithms for edge coloring bipartite graphs and multigraphs},\ }\href {https://doi.org/10.1137/0211009} {\bibfield  {journal} {\bibinfo  {journal} {SIAM Journal on Computing}\ }\textbf {\bibinfo {volume} {11}},\ \bibinfo {pages} {117} (\bibinfo {year} {1982})}\BibitemShut {NoStop}%
\bibitem [{\citenamefont {Jiang}\ \emph {et~al.}(2020)\citenamefont {Jiang}, \citenamefont {Kalev}, \citenamefont {Mruczkiewicz},\ and\ \citenamefont {Neven}}]{Jiang_2020}%
  \BibitemOpen
  \bibfield  {author} {\bibinfo {author} {\bibfnamefont {Z.}~\bibnamefont {Jiang}}, \bibinfo {author} {\bibfnamefont {A.}~\bibnamefont {Kalev}}, \bibinfo {author} {\bibfnamefont {W.}~\bibnamefont {Mruczkiewicz}},\ and\ \bibinfo {author} {\bibfnamefont {H.}~\bibnamefont {Neven}},\ }\bibfield  {title} {\bibinfo {title} {Optimal fermion-to-qubit mapping via ternary trees with applications to reduced quantum states learning},\ }\href {https://doi.org/10.22331/q-2020-06-04-276} {\bibfield  {journal} {\bibinfo  {journal} {Quantum}\ }\textbf {\bibinfo {volume} {4}},\ \bibinfo {pages} {276} (\bibinfo {year} {2020})}\BibitemShut {NoStop}%
\bibitem [{\citenamefont {Yu}\ \emph {et~al.}(2025)\citenamefont {Yu}, \citenamefont {Liu}, \citenamefont {Sugiura}, \citenamefont {Van~Voorhis},\ and\ \citenamefont {Zeytinoğlu}}]{Yu_2025}%
  \BibitemOpen
  \bibfield  {author} {\bibinfo {author} {\bibfnamefont {J.}~\bibnamefont {Yu}}, \bibinfo {author} {\bibfnamefont {Y.}~\bibnamefont {Liu}}, \bibinfo {author} {\bibfnamefont {S.}~\bibnamefont {Sugiura}}, \bibinfo {author} {\bibfnamefont {T.}~\bibnamefont {Van~Voorhis}},\ and\ \bibinfo {author} {\bibfnamefont {S.}~\bibnamefont {Zeytinoğlu}},\ }\bibfield  {title} {\bibinfo {title} {Clifford circuit-based heuristic optimization of fermion-to-qubit mappings},\ }\href {https://doi.org/10.1021/acs.jctc.5c00794} {\bibfield  {journal} {\bibinfo  {journal} {Journal of Chemical Theory and Computation}\ }\textbf {\bibinfo {volume} {21}},\ \bibinfo {pages} {9430–9443} (\bibinfo {year} {2025})}\BibitemShut {NoStop}%
\bibitem [{\citenamefont {Miller}\ \emph {et~al.}(2023)\citenamefont {Miller}, \citenamefont {Zimborás}, \citenamefont {Knecht}, \citenamefont {Maniscalco},\ and\ \citenamefont {García-Pérez}}]{Miller_2023}%
  \BibitemOpen
  \bibfield  {author} {\bibinfo {author} {\bibfnamefont {A.}~\bibnamefont {Miller}}, \bibinfo {author} {\bibfnamefont {Z.}~\bibnamefont {Zimborás}}, \bibinfo {author} {\bibfnamefont {S.}~\bibnamefont {Knecht}}, \bibinfo {author} {\bibfnamefont {S.}~\bibnamefont {Maniscalco}},\ and\ \bibinfo {author} {\bibfnamefont {G.}~\bibnamefont {García-Pérez}},\ }\bibfield  {title} {\bibinfo {title} {Bonsai algorithm: Grow your own fermion-to-qubit mappings},\ }\bibfield  {journal} {\bibinfo  {journal} {PRX Quantum}\ }\textbf {\bibinfo {volume} {4}},\ \href {https://doi.org/10.1103/prxquantum.4.030314} {10.1103/prxquantum.4.030314} (\bibinfo {year} {2023})\BibitemShut {NoStop}%
\bibitem [{\citenamefont {Chiew}\ \emph {et~al.}(2024)\citenamefont {Chiew}, \citenamefont {Harrison},\ and\ \citenamefont {Strelchuk}}]{chiew2024ternarytreetransformationsequivalent}%
  \BibitemOpen
  \bibfield  {author} {\bibinfo {author} {\bibfnamefont {M.}~\bibnamefont {Chiew}}, \bibinfo {author} {\bibfnamefont {B.}~\bibnamefont {Harrison}},\ and\ \bibinfo {author} {\bibfnamefont {S.}~\bibnamefont {Strelchuk}},\ }\href {https://arxiv.org/abs/2412.07578} {\bibinfo {title} {Ternary tree transformations are equivalent to linear encodings of the fock basis}} (\bibinfo {year} {2024}),\ \Eprint {https://arxiv.org/abs/2412.07578} {arXiv:2412.07578 [quant-ph]} \BibitemShut {NoStop}%
\bibitem [{\citenamefont {Vlasov}(2024)}]{vlasov2024mutualtransformationsarbitraryternary}%
  \BibitemOpen
  \bibfield  {author} {\bibinfo {author} {\bibfnamefont {A.~Y.}\ \bibnamefont {Vlasov}},\ }\href {https://arxiv.org/abs/2404.16693} {\bibinfo {title} {Mutual transformations of arbitrary ternary qubit trees by clifford gates}} (\bibinfo {year} {2024}),\ \Eprint {https://arxiv.org/abs/2404.16693} {arXiv:2404.16693 [quant-ph]} \BibitemShut {NoStop}%
\bibitem [{\citenamefont {{Cirq Developers}}(2025)}]{cirq}%
  \BibitemOpen
  \bibfield  {author} {\bibinfo {author} {\bibnamefont {{Cirq Developers}}},\ }\href {https://doi.org/10.5281/ZENODO.4062499} {\emph {\bibinfo {title} {Cirq}}}\ (\bibinfo  {publisher} {Zenodo},\ \bibinfo {year} {2025})\BibitemShut {NoStop}%
\bibitem [{\citenamefont {Fowler}\ \emph {et~al.}(2012)\citenamefont {Fowler}, \citenamefont {Mariantoni}, \citenamefont {Martinis},\ and\ \citenamefont {Cleland}}]{Fowler_2012}%
  \BibitemOpen
  \bibfield  {author} {\bibinfo {author} {\bibfnamefont {A.~G.}\ \bibnamefont {Fowler}}, \bibinfo {author} {\bibfnamefont {M.}~\bibnamefont {Mariantoni}}, \bibinfo {author} {\bibfnamefont {J.~M.}\ \bibnamefont {Martinis}},\ and\ \bibinfo {author} {\bibfnamefont {A.~N.}\ \bibnamefont {Cleland}},\ }\bibfield  {title} {\bibinfo {title} {Surface codes: Towards practical large-scale quantum computation},\ }\bibfield  {journal} {\bibinfo  {journal} {Physical Review A}\ }\textbf {\bibinfo {volume} {86}},\ \href {https://doi.org/10.1103/physreva.86.032324} {10.1103/physreva.86.032324} (\bibinfo {year} {2012})\BibitemShut {NoStop}%
\bibitem [{\citenamefont {Beverland}\ \emph {et~al.}(2022)\citenamefont {Beverland}, \citenamefont {Murali}, \citenamefont {Troyer}, \citenamefont {Svore}, \citenamefont {Hoefler}, \citenamefont {Kliuchnikov}, \citenamefont {Low}, \citenamefont {Soeken}, \citenamefont {Sundaram},\ and\ \citenamefont {Vaschillo}}]{beverland2022assessing}%
  \BibitemOpen
  \bibfield  {author} {\bibinfo {author} {\bibfnamefont {M.~E.}\ \bibnamefont {Beverland}}, \bibinfo {author} {\bibfnamefont {P.}~\bibnamefont {Murali}}, \bibinfo {author} {\bibfnamefont {M.}~\bibnamefont {Troyer}}, \bibinfo {author} {\bibfnamefont {K.~M.}\ \bibnamefont {Svore}}, \bibinfo {author} {\bibfnamefont {T.}~\bibnamefont {Hoefler}}, \bibinfo {author} {\bibfnamefont {V.}~\bibnamefont {Kliuchnikov}}, \bibinfo {author} {\bibfnamefont {G.~H.}\ \bibnamefont {Low}}, \bibinfo {author} {\bibfnamefont {M.}~\bibnamefont {Soeken}}, \bibinfo {author} {\bibfnamefont {A.}~\bibnamefont {Sundaram}},\ and\ \bibinfo {author} {\bibfnamefont {A.}~\bibnamefont {Vaschillo}},\ }\bibfield  {title} {\bibinfo {title} {Assessing requirements to scale to practical quantum advantage},\ }\href@noop {} {\bibfield  {journal} {\bibinfo  {journal} {arXiv preprint arXiv:2211.07629}\ } (\bibinfo {year} {2022})}\BibitemShut {NoStop}%
\bibitem [{\citenamefont {McKay}\ \emph {et~al.}(2017)\citenamefont {McKay}, \citenamefont {Wood}, \citenamefont {Sheldon}, \citenamefont {Chow},\ and\ \citenamefont {Gambetta}}]{mckay2017efficient}%
  \BibitemOpen
  \bibfield  {author} {\bibinfo {author} {\bibfnamefont {D.~C.}\ \bibnamefont {McKay}}, \bibinfo {author} {\bibfnamefont {C.~J.}\ \bibnamefont {Wood}}, \bibinfo {author} {\bibfnamefont {S.}~\bibnamefont {Sheldon}}, \bibinfo {author} {\bibfnamefont {J.~M.}\ \bibnamefont {Chow}},\ and\ \bibinfo {author} {\bibfnamefont {J.~M.}\ \bibnamefont {Gambetta}},\ }\bibfield  {title} {\bibinfo {title} {Efficient {$Z$} gates for quantum computing},\ }\href {https://doi.org/10.1103/PhysRevA.96.022330} {\bibfield  {journal} {\bibinfo  {journal} {Physical Review A}\ }\textbf {\bibinfo {volume} {96}},\ \bibinfo {pages} {022330} (\bibinfo {year} {2017})}\BibitemShut {NoStop}%
\bibitem [{\citenamefont {Riesebos}\ \emph {et~al.}(2017)\citenamefont {Riesebos}, \citenamefont {Fu}, \citenamefont {Varsamopoulos}, \citenamefont {Almud{\'e}ver},\ and\ \citenamefont {Bertels}}]{riesebos2017pauli}%
  \BibitemOpen
  \bibfield  {author} {\bibinfo {author} {\bibfnamefont {L.}~\bibnamefont {Riesebos}}, \bibinfo {author} {\bibfnamefont {X.}~\bibnamefont {Fu}}, \bibinfo {author} {\bibfnamefont {S.}~\bibnamefont {Varsamopoulos}}, \bibinfo {author} {\bibfnamefont {C.~G.}\ \bibnamefont {Almud{\'e}ver}},\ and\ \bibinfo {author} {\bibfnamefont {K.}~\bibnamefont {Bertels}},\ }\bibfield  {title} {\bibinfo {title} {Pauli frames for quantum computer architectures},\ }in\ \href {https://doi.org/10.1145/3061639.3062300} {\emph {\bibinfo {booktitle} {Proceedings of the 54th Annual Design Automation Conference 2017}}}\ (\bibinfo  {publisher} {ACM},\ \bibinfo {year} {2017})\ pp.\ \bibinfo {pages} {1--6}\BibitemShut {NoStop}%
\bibitem [{\citenamefont {Gidney}(2021)}]{stim}%
  \BibitemOpen
  \bibfield  {author} {\bibinfo {author} {\bibfnamefont {C.}~\bibnamefont {Gidney}},\ }\bibfield  {title} {\bibinfo {title} {Stim: a fast stabilizer circuit simulator},\ }\href {https://doi.org/10.22331/q-2021-07-06-497} {\bibfield  {journal} {\bibinfo  {journal} {{Quantum}}\ }\textbf {\bibinfo {volume} {5}},\ \bibinfo {pages} {497} (\bibinfo {year} {2021})}\BibitemShut {NoStop}%
\bibitem [{\citenamefont {Huggins}\ \emph {et~al.}(2025)\citenamefont {Huggins}, \citenamefont {Khattar}, \citenamefont {Xu}, \citenamefont {Harrigan}, \citenamefont {Kang}, \citenamefont {Low}, \citenamefont {Fowler}, \citenamefont {Rubin},\ and\ \citenamefont {Babbush}}]{huggins2025flasq}%
  \BibitemOpen
  \bibfield  {author} {\bibinfo {author} {\bibfnamefont {W.~J.}\ \bibnamefont {Huggins}}, \bibinfo {author} {\bibfnamefont {T.}~\bibnamefont {Khattar}}, \bibinfo {author} {\bibfnamefont {A.}~\bibnamefont {Xu}}, \bibinfo {author} {\bibfnamefont {M.}~\bibnamefont {Harrigan}}, \bibinfo {author} {\bibfnamefont {C.}~\bibnamefont {Kang}}, \bibinfo {author} {\bibfnamefont {G.~H.}\ \bibnamefont {Low}}, \bibinfo {author} {\bibfnamefont {A.}~\bibnamefont {Fowler}}, \bibinfo {author} {\bibfnamefont {N.~C.}\ \bibnamefont {Rubin}},\ and\ \bibinfo {author} {\bibfnamefont {R.}~\bibnamefont {Babbush}},\ }\bibfield  {title} {\bibinfo {title} {The {FL}uid {A}llocation of {S}urface code {Q}ubits ({FLASQ}) cost model for early fault-tolerant quantum algorithms},\ }\bibfield  {journal} {\bibinfo  {journal} {arXiv preprint arXiv:2511.08508}\ }\href {https://doi.org/10.48550/arXiv.2511.08508} {10.48550/arXiv.2511.08508} (\bibinfo {year} {2025})\BibitemShut {NoStop}%
\bibitem [{\citenamefont {Ross}\ and\ \citenamefont {Selinger}(2016)}]{ross2016optimal}%
  \BibitemOpen
  \bibfield  {author} {\bibinfo {author} {\bibfnamefont {N.~J.}\ \bibnamefont {Ross}}\ and\ \bibinfo {author} {\bibfnamefont {P.}~\bibnamefont {Selinger}},\ }\bibfield  {title} {\bibinfo {title} {Optimal ancilla-free {Clifford}+{T} approximation of {$z$}-rotations},\ }\href {https://doi.org/10.26421/QIC16.11-12-1} {\bibfield  {journal} {\bibinfo  {journal} {Quantum Information \& Computation}\ }\textbf {\bibinfo {volume} {16}},\ \bibinfo {pages} {901} (\bibinfo {year} {2016})}\BibitemShut {NoStop}%
\bibitem [{\citenamefont {Litinski}(2019)}]{litinski2019magic}%
  \BibitemOpen
  \bibfield  {author} {\bibinfo {author} {\bibfnamefont {D.}~\bibnamefont {Litinski}},\ }\bibfield  {title} {\bibinfo {title} {Magic state distillation: Not as costly as you think},\ }\href {https://doi.org/10.22331/q-2019-12-02-205} {\bibfield  {journal} {\bibinfo  {journal} {Quantum}\ }\textbf {\bibinfo {volume} {3}},\ \bibinfo {pages} {205} (\bibinfo {year} {2019})}\BibitemShut {NoStop}%
\bibitem [{\citenamefont {Gidney}\ \emph {et~al.}(2024)\citenamefont {Gidney}, \citenamefont {Shutty},\ and\ \citenamefont {Jones}}]{gidney2024cultivation}%
  \BibitemOpen
  \bibfield  {author} {\bibinfo {author} {\bibfnamefont {C.}~\bibnamefont {Gidney}}, \bibinfo {author} {\bibfnamefont {N.}~\bibnamefont {Shutty}},\ and\ \bibinfo {author} {\bibfnamefont {C.}~\bibnamefont {Jones}},\ }\bibfield  {title} {\bibinfo {title} {Magic state cultivation: growing {T} states as cheap as {CNOT} gates},\ }\bibfield  {journal} {\bibinfo  {journal} {arXiv preprint arXiv:2409.17595}\ }\href {https://doi.org/10.48550/arXiv.2409.17595} {10.48550/arXiv.2409.17595} (\bibinfo {year} {2024})\BibitemShut {NoStop}%
\bibitem [{\citenamefont {Verstraete}\ \emph {et~al.}(2009)\citenamefont {Verstraete}, \citenamefont {Cirac},\ and\ \citenamefont {Latorre}}]{Verstraete_2009}%
  \BibitemOpen
  \bibfield  {author} {\bibinfo {author} {\bibfnamefont {F.}~\bibnamefont {Verstraete}}, \bibinfo {author} {\bibfnamefont {J.~I.}\ \bibnamefont {Cirac}},\ and\ \bibinfo {author} {\bibfnamefont {J.~I.}\ \bibnamefont {Latorre}},\ }\bibfield  {title} {\bibinfo {title} {Quantum circuits for strongly correlated quantum systems},\ }\bibfield  {journal} {\bibinfo  {journal} {Physical Review A}\ }\textbf {\bibinfo {volume} {79}},\ \href {https://doi.org/10.1103/physreva.79.032316} {10.1103/physreva.79.032316} (\bibinfo {year} {2009})\BibitemShut {NoStop}%
\bibitem [{\citenamefont {McClean}\ \emph {et~al.}(2020)\citenamefont {McClean}, \citenamefont {Rubin}, \citenamefont {Sung}, \citenamefont {Kivlichan}, \citenamefont {Bonet-Monroig}, \citenamefont {Cao}, \citenamefont {Dai}, \citenamefont {Fried}, \citenamefont {Gidney}, \citenamefont {Gimby} \emph {et~al.}}]{openfermion}%
  \BibitemOpen
  \bibfield  {author} {\bibinfo {author} {\bibfnamefont {J.~R.}\ \bibnamefont {McClean}}, \bibinfo {author} {\bibfnamefont {N.~C.}\ \bibnamefont {Rubin}}, \bibinfo {author} {\bibfnamefont {K.~J.}\ \bibnamefont {Sung}}, \bibinfo {author} {\bibfnamefont {I.~D.}\ \bibnamefont {Kivlichan}}, \bibinfo {author} {\bibfnamefont {X.}~\bibnamefont {Bonet-Monroig}}, \bibinfo {author} {\bibfnamefont {Y.}~\bibnamefont {Cao}}, \bibinfo {author} {\bibfnamefont {C.}~\bibnamefont {Dai}}, \bibinfo {author} {\bibfnamefont {E.~S.}\ \bibnamefont {Fried}}, \bibinfo {author} {\bibfnamefont {C.}~\bibnamefont {Gidney}}, \bibinfo {author} {\bibfnamefont {B.}~\bibnamefont {Gimby}}, \emph {et~al.},\ }\bibfield  {title} {\bibinfo {title} {Openfermion: the electronic structure package for quantum computers},\ }\href@noop {} {\bibfield  {journal} {\bibinfo  {journal} {Quantum Science \& Technology}\ }\textbf {\bibinfo {volume} {5}},\ \bibinfo {pages} {034014} (\bibinfo {year} {2020})}\BibitemShut {NoStop}%
\bibitem [{\citenamefont {Derby}\ \emph {et~al.}(2021)\citenamefont {Derby}, \citenamefont {Klassen}, \citenamefont {Bausch},\ and\ \citenamefont {Cubitt}}]{compactencoding}%
  \BibitemOpen
  \bibfield  {author} {\bibinfo {author} {\bibfnamefont {C.}~\bibnamefont {Derby}}, \bibinfo {author} {\bibfnamefont {J.}~\bibnamefont {Klassen}}, \bibinfo {author} {\bibfnamefont {J.}~\bibnamefont {Bausch}},\ and\ \bibinfo {author} {\bibfnamefont {T.}~\bibnamefont {Cubitt}},\ }\bibfield  {title} {\bibinfo {title} {Compact fermion to qubit mappings},\ }\href {https://doi.org/10.1103/PhysRevB.104.035118} {\bibfield  {journal} {\bibinfo  {journal} {Phys. Rev. B}\ }\textbf {\bibinfo {volume} {104}},\ \bibinfo {pages} {035118} (\bibinfo {year} {2021})}\BibitemShut {NoStop}%
\bibitem [{\citenamefont {Orman}\ \emph {et~al.}(2025)\citenamefont {Orman}, \citenamefont {Gharibyan},\ and\ \citenamefont {Preskill}}]{syk}%
  \BibitemOpen
  \bibfield  {author} {\bibinfo {author} {\bibfnamefont {P.}~\bibnamefont {Orman}}, \bibinfo {author} {\bibfnamefont {H.}~\bibnamefont {Gharibyan}},\ and\ \bibinfo {author} {\bibfnamefont {J.}~\bibnamefont {Preskill}},\ }\bibfield  {title} {\bibinfo {title} {Quantum chaos in the sparse {SYK} model},\ }\href {https://doi.org/10.1007/JHEP02(2025)173} {\bibfield  {journal} {\bibinfo  {journal} {Journal of High Energy Physics}\ }\textbf {\bibinfo {volume} {2025}},\ \bibinfo {pages} {173} (\bibinfo {year} {2025})},\ \Eprint {https://arxiv.org/abs/2403.13884} {arXiv:2403.13884 [hep-th]} \BibitemShut {NoStop}%
\bibitem [{\citenamefont {Sachdev}\ and\ \citenamefont {Ye}(1993)}]{Sachdev_1993}%
  \BibitemOpen
  \bibfield  {author} {\bibinfo {author} {\bibfnamefont {S.}~\bibnamefont {Sachdev}}\ and\ \bibinfo {author} {\bibfnamefont {J.}~\bibnamefont {Ye}},\ }\bibfield  {title} {\bibinfo {title} {Gapless spin-fluid ground state in a random quantum heisenberg magnet},\ }\href {https://doi.org/10.1103/physrevlett.70.3339} {\bibfield  {journal} {\bibinfo  {journal} {Physical Review Letters}\ }\textbf {\bibinfo {volume} {70}},\ \bibinfo {pages} {3339–3342} (\bibinfo {year} {1993})}\BibitemShut {NoStop}%
\bibitem [{\citenamefont {Xu}\ \emph {et~al.}(2020)\citenamefont {Xu}, \citenamefont {Susskind}, \citenamefont {Su},\ and\ \citenamefont {Swingle}}]{xu2020sparsemodelquantumholography}%
  \BibitemOpen
  \bibfield  {author} {\bibinfo {author} {\bibfnamefont {S.}~\bibnamefont {Xu}}, \bibinfo {author} {\bibfnamefont {L.}~\bibnamefont {Susskind}}, \bibinfo {author} {\bibfnamefont {Y.}~\bibnamefont {Su}},\ and\ \bibinfo {author} {\bibfnamefont {B.}~\bibnamefont {Swingle}},\ }\href {https://arxiv.org/abs/2008.02303} {\bibinfo {title} {A sparse model of quantum holography}} (\bibinfo {year} {2020}),\ \Eprint {https://arxiv.org/abs/2008.02303} {arXiv:2008.02303 [cond-mat.str-el]} \BibitemShut {NoStop}%
\bibitem [{\citenamefont {Setia}\ and\ \citenamefont {Whitfield}(2018)}]{setia2018superfast}%
  \BibitemOpen
  \bibfield  {author} {\bibinfo {author} {\bibfnamefont {K.}~\bibnamefont {Setia}}\ and\ \bibinfo {author} {\bibfnamefont {J.~D.}\ \bibnamefont {Whitfield}},\ }\bibfield  {title} {\bibinfo {title} {Bravyi-{K}itaev {S}uperfast simulation of electronic structure on a quantum computer},\ }\href {https://doi.org/10.1063/1.5019371} {\bibfield  {journal} {\bibinfo  {journal} {The Journal of Chemical Physics}\ }\textbf {\bibinfo {volume} {148}},\ \bibinfo {pages} {164104} (\bibinfo {year} {2018})}\BibitemShut {NoStop}%
\bibitem [{\citenamefont {B{\"a}umer}\ and\ \citenamefont {Woerner}(2025)}]{baumer2025measurement}%
  \BibitemOpen
  \bibfield  {author} {\bibinfo {author} {\bibfnamefont {E.}~\bibnamefont {B{\"a}umer}}\ and\ \bibinfo {author} {\bibfnamefont {S.}~\bibnamefont {Woerner}},\ }\bibfield  {title} {\bibinfo {title} {Measurement-based long-range entangling gates in constant depth},\ }\href@noop {} {\bibfield  {journal} {\bibinfo  {journal} {Physical Review Research}\ }\textbf {\bibinfo {volume} {7}},\ \bibinfo {pages} {023120} (\bibinfo {year} {2025})}\BibitemShut {NoStop}%
\bibitem [{\citenamefont {Moore}(1999)}]{moore1999quantum}%
  \BibitemOpen
  \bibfield  {author} {\bibinfo {author} {\bibfnamefont {C.}~\bibnamefont {Moore}},\ }\bibfield  {title} {\bibinfo {title} {Quantum circuits: Fanout, parity, and counting},\ }\href@noop {} {\bibfield  {journal} {\bibinfo  {journal} {arXiv preprint quant-ph/9903046}\ } (\bibinfo {year} {1999})}\BibitemShut {NoStop}%
\bibitem [{\citenamefont {Fang}\ \emph {et~al.}(2003)\citenamefont {Fang}, \citenamefont {Fenner}, \citenamefont {Green}, \citenamefont {Homer},\ and\ \citenamefont {Zhang}}]{fang2003quantum}%
  \BibitemOpen
  \bibfield  {author} {\bibinfo {author} {\bibfnamefont {M.}~\bibnamefont {Fang}}, \bibinfo {author} {\bibfnamefont {S.}~\bibnamefont {Fenner}}, \bibinfo {author} {\bibfnamefont {F.}~\bibnamefont {Green}}, \bibinfo {author} {\bibfnamefont {S.}~\bibnamefont {Homer}},\ and\ \bibinfo {author} {\bibfnamefont {Y.}~\bibnamefont {Zhang}},\ }\bibfield  {title} {\bibinfo {title} {Quantum lower bounds for fanout},\ }\href@noop {} {\bibfield  {journal} {\bibinfo  {journal} {arXiv preprint quant-ph/0312208}\ } (\bibinfo {year} {2003})}\BibitemShut {NoStop}%
\bibitem [{\citenamefont {Remaud}\ and\ \citenamefont {Vandaele}(2025)}]{remaud2025ancilla}%
  \BibitemOpen
  \bibfield  {author} {\bibinfo {author} {\bibfnamefont {M.}~\bibnamefont {Remaud}}\ and\ \bibinfo {author} {\bibfnamefont {V.}~\bibnamefont {Vandaele}},\ }\bibfield  {title} {\bibinfo {title} {Ancilla-free quantum adder with sublinear depth},\ }in\ \href@noop {} {\emph {\bibinfo {booktitle} {International Conference on Reversible Computation}}}\ (\bibinfo {organization} {Springer},\ \bibinfo {year} {2025})\ pp.\ \bibinfo {pages} {137--154}\BibitemShut {NoStop}%
\bibitem [{\citenamefont {Bader}(2012)}]{bader2012space}%
  \BibitemOpen
  \bibfield  {author} {\bibinfo {author} {\bibfnamefont {M.}~\bibnamefont {Bader}},\ }\href@noop {} {\emph {\bibinfo {title} {Space-filling curves: an introduction with applications in scientific computing}}},\ Vol.~\bibinfo {volume} {9}\ (\bibinfo  {publisher} {Springer Science \& Business Media},\ \bibinfo {year} {2012})\BibitemShut {NoStop}%
\end{thebibliography}%

\end{document}